\RequirePackage{fix-cm}
\documentclass[onecolumn,prx,amsmath,11pt]{revtex4-2}
\usepackage{microtype} 

\usepackage[margin=1in]{geometry}
\usepackage{setspace}

\usepackage{graphicx}
\usepackage{dcolumn}
\usepackage{amsmath,mathtools,braket,float}
\mathtoolsset{showonlyrefs=true} 
\usepackage{color,comment,footnote}
\allowdisplaybreaks 
\usepackage{booktabs} 

\usepackage{tcolorbox}
\usepackage{amsthm}
\theoremstyle{plain}
\newtheorem{thm}{Theorem}

\newtheorem{lem}{Lemma}

\theoremstyle{definition}
\newtheorem{dfn}{Definition}

\usepackage[hyphens]{xurl} 
\usepackage{hyperref}
\usepackage{xcolor}
\hypersetup{
colorlinks=true,
citecolor=blue,
linkcolor=blue,
urlcolor=blue,
}

\usepackage[T1]{fontenc}
\usepackage{newtxtext} 
\usepackage[smallerops,nosymbolsc]{newtxmath}
\usepackage{bm}

\DeclareMathAlphabet{\mathcal}{OMS}{cmsy}{m}{n}
\SetMathAlphabet{\mathcal}{bold}{OMS}{cmsy}{b}{n}
\DeclareMathAlphabet{\mathbb}{U}{msb}{m}{n}
\SetMathAlphabet{\mathbb}{bold}{U}{msb}{m}{n}
\DeclareMathAlphabet{\mathfrak}{U}{euf}{m}{n}
\SetMathAlphabet{\mathfrak}{bold}{U}{euf}{b}{n}

\DeclareSymbolFont{cmlargesymbols}{OMX}{cmex}{m}{n}
\SetSymbolFont{cmlargesymbols}{bold}{OMX}{cmex}{m}{n}

\let\sum\relax       \DeclareMathSymbol{\sum}{\mathop}{cmlargesymbols}{"50}
\let\prod\relax      \DeclareMathSymbol{\prod}{\mathop}{cmlargesymbols}{"51}
\let\coprod\relax    \DeclareMathSymbol{\coprod}{\mathop}{cmlargesymbols}{"60}

\let\bigcup\relax    \DeclareMathSymbol{\bigcup}{\mathop}{cmlargesymbols}{"53}
\let\bigcap\relax    \DeclareMathSymbol{\bigcap}{\mathop}{cmlargesymbols}{"54}
\let\biguplus\relax  \DeclareMathSymbol{\biguplus}{\mathop}{cmlargesymbols}{"55}
\let\bigsqcup\relax  \DeclareMathSymbol{\bigsqcup}{\mathop}{cmlargesymbols}{"46}

\let\bigwedge\relax  \DeclareMathSymbol{\bigwedge}{\mathop}{cmlargesymbols}{"56}
\let\bigvee\relax    \DeclareMathSymbol{\bigvee}{\mathop}{cmlargesymbols}{"57}

\let\bigodot\relax   \DeclareMathSymbol{\bigodot}{\mathop}{cmlargesymbols}{"4A}
\let\bigoplus\relax  \DeclareMathSymbol{\bigoplus}{\mathop}{cmlargesymbols}{"4C}
\let\bigotimes\relax \DeclareMathSymbol{\bigotimes}{\mathop}{cmlargesymbols}{"4E}

\DeclareSymbolFont{cmoperators}{OT1}{cmr}{m}{n}
\SetSymbolFont{cmoperators}{bold}{OT1}{cmr}{bx}{n}
\DeclareSymbolFont{cmsymbols}{OMS}{cmsy}{m}{n}
\SetSymbolFont{cmsymbols}{bold}{OMS}{cmsy}{b}{n}

\DeclareSymbolFontAlphabet{\mathcal}{cmsymbols}

\DeclareMathDelimiter{(}{\mathopen} {cmoperators}{"28}{cmlargesymbols}{"00}
\DeclareMathDelimiter{)}{\mathclose}{cmoperators}{"29}{cmlargesymbols}{"01}
\DeclareMathDelimiter{[}{\mathopen} {cmoperators}{"5B}{cmlargesymbols}{"02}
\DeclareMathDelimiter{]}{\mathclose}{cmoperators}{"5D}{cmlargesymbols}{"03}
\DeclareMathDelimiter{|}{\mathord}  {cmsymbols}  {"6A}{cmlargesymbols}{"0C}
\DeclareMathDelimiter{/}{\mathord}  {cmoperators}{"2F}{cmlargesymbols}{"0E}

\let\lbrace\relax    \DeclareMathDelimiter{\lbrace}{\mathopen} {cmsymbols}{"66}{cmlargesymbols}{"08}
\let\rbrace\relax    \DeclareMathDelimiter{\rbrace}{\mathclose}{cmsymbols}{"67}{cmlargesymbols}{"09}
\let\langle\relax    \DeclareMathDelimiter{\langle}{\mathopen} {cmsymbols}{"68}{cmlargesymbols}{"0A}
\let\rangle\relax    \DeclareMathDelimiter{\rangle}{\mathclose}{cmsymbols}{"69}{cmlargesymbols}{"0B}
\let\lfloor\relax    \DeclareMathDelimiter{\lfloor}{\mathopen} {cmsymbols}{"62}{cmlargesymbols}{"04}
\let\rfloor\relax    \DeclareMathDelimiter{\rfloor}{\mathclose}{cmsymbols}{"63}{cmlargesymbols}{"05}
\let\lceil\relax     \DeclareMathDelimiter{\lceil}{\mathopen}  {cmsymbols}{"64}{cmlargesymbols}{"06}
\let\rceil\relax     \DeclareMathDelimiter{\rceil}{\mathclose} {cmsymbols}{"65}{cmlargesymbols}{"07}
\let\vert\relax      \DeclareMathDelimiter{\vert}{\mathord}    {cmsymbols}{"6A}{cmlargesymbols}{"0C}
\let\Vert\relax      \DeclareMathDelimiter{\Vert}{\mathord}    {cmsymbols}{"6B}{cmlargesymbols}{"0D}
\let\|\relax         \DeclareMathDelimiter{\|}{\mathord}       {cmsymbols}{"6B}{cmlargesymbols}{"0D}
\let\backslash\relax \DeclareMathDelimiter{\backslash}{\mathord}{cmsymbols}{"6E}{cmlargesymbols}{"0F}

\let\lvert\relax     \DeclareMathDelimiter{\lvert}{\mathopen}  {cmsymbols}{"6A}{cmlargesymbols}{"0C}
\let\rvert\relax     \DeclareMathDelimiter{\rvert}{\mathclose} {cmsymbols}{"6A}{cmlargesymbols}{"0C}
\let\lVert\relax     \DeclareMathDelimiter{\lVert}{\mathopen}  {cmsymbols}{"6B}{cmlargesymbols}{"0D}
\let\rVert\relax     \DeclareMathDelimiter{\rVert}{\mathclose} {cmsymbols}{"6B}{cmlargesymbols}{"0D}

\let\lgroup\relax     \DeclareMathDelimiter{\lgroup}{\mathopen}    {cmlargesymbols}{"3A}{cmlargesymbols}{"3A}
\let\rgroup\relax     \DeclareMathDelimiter{\rgroup}{\mathclose}   {cmlargesymbols}{"3B}{cmlargesymbols}{"3B}
\let\lmoustache\relax \DeclareMathDelimiter{\lmoustache}{\mathopen}{cmlargesymbols}{"40}{cmlargesymbols}{"40}
\let\rmoustache\relax \DeclareMathDelimiter{\rmoustache}{\mathclose}{cmlargesymbols}{"41}{cmlargesymbols}{"41}

\newcommand{\mc}{\mathcal}

\newcommand{\tr}{\operatorname{tr}}

\newcommand{\mD}{\mathcal{D}}

\newcommand{\mG}{\mathcal{G}}
\newcommand{\mH}{\mathcal{H}}

\newcommand{\mO}{\mathcal{O}}
\newcommand{\mP}{\mathcal{P}}

\newcommand{\mX}{\mathcal{X}}
\newcommand{\mY}{\mathcal{Y}}

\newcommand{\mOt}{\tilde{\mathcal{O}}}

\newcommand{\bt}{\bm{t}}

\newcommand{\bw}{\bm{w}}
\newcommand{\bx}{\bm{x}}
\newcommand{\by}{\bm{y}}
\newcommand{\bz}{\bm{z}}

\newcommand{\balpha}{\bm{\alpha}}

\newcommand{\btheta}{\bm{\theta}}

\newcommand{\AGL}{\mathcal{A}_{\rm GL}}
\newcommand{\supp}{\operatorname{supp}}
\newcommand{\dist}{\operatorname{dist}}

\newcommand{\hc}{\hat{c}}
\newcommand{\tc}{\tilde{c}}
\newcommand{\barc}{\bar{c}}

\newcommand{\LCN}{\operatorname{LCN}}
\newcommand{\LFC}{\operatorname{LFC}}
\newcommand{\clip}{\operatorname{clip}}
\newcommand{\sgn}{\operatorname{sgn}}
\newcommand{\vectorize}{\operatorname{vec}}
\newcommand{\Var}{\operatorname{Var}}
\newcommand{\Unif}{\operatorname{Unif}}

\makeatletter
\newcommand{\disabletoc}{%
  \let\oldaddcontentsline\addcontentsline
  \renewcommand{\addcontentsline}[3]{}%
}
\newcommand{\enabletoc}{%
  \let\addcontentsline\oldaddcontentsline
}
\makeatother

\begin{document}
\disabletoc
\singlespacing

\title{Provably scalable learning from quantum many-body data via locality}
\author{Koki Chinzei}\thanks{chinzei.koki@fujitsu.com}
\author{Quoc Hoan Tran}
\author{Norifumi Matsumoto}
\author{Yasuhiro Endo}
\author{Hirotaka Oshima}
\affiliation{Quantum Laboratory, Fujitsu Research, Fujitsu Limited,
	4-1-1 Kawasaki, Kanagawa 211-8588, Japan}


\begin{abstract}
Machine learning (ML) holds great promise for extracting insights from complex quantum many-body data obtained in quantum experiments, potentially solving certain classically intractable problems. 
However, addressing large-scale problems still requires significant amounts of data that exceed the capabilities of near-term quantum devices.
We propose a scalable ML framework, the Geometrically Local Quantum Kernel (GLQK), that efficiently learns from quantum many-body data by leveraging the exponential decay of correlations prevalent in noncritical systems. For learning an unknown polynomial of quantum expectation values, we rigorously prove that the GLQK substantially improves the sample complexity in the system size $n$ over the existing shadow kernel, by constructing a feature space from local quantum information at the correlation length scale. This improvement is particularly notable when each term of the target polynomial involves few local subsystems. Remarkably, for translationally symmetric data, the GLQK achieves constant sample complexity, independent of $n$. 
Numerical experiments on two quantum many-body tasks corroborate these theoretical results, demonstrating the high scalability of the GLQK.
Beyond the polynomial setting, we further show that it can efficiently learn certain non-polynomial functions.
Our framework establishes new avenues for utilizing experimental data to advance the understanding of quantum many-body physics.
\end{abstract}

\maketitle

\section{Introduction}

\begin{figure*}[t]
    \centering
    \includegraphics[width=\linewidth]{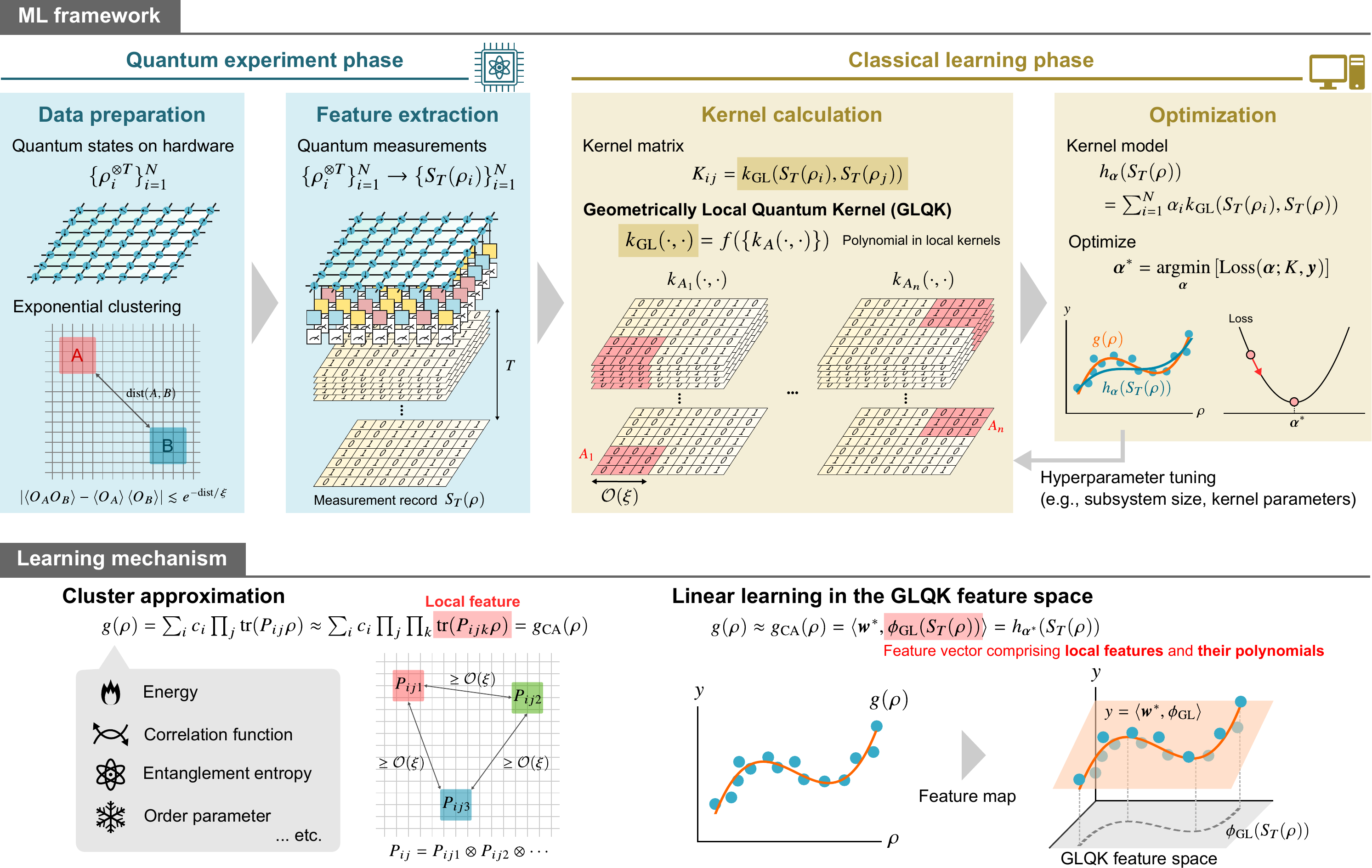}
    \caption{
    Overview of our ML framework and its mechanism.
    (Top) Our ML framework comprises the quantum experiment phase and the classical learning phase.
    In the quantum experiment phase, quantum data $\rho_i^{\otimes T}$ is prepared on quantum hardware (e.g., digital quantum computers, analog quantum simulators) and then measured in several bases.
    This process extracts quantum features $S_T(\rho_i)$, which record the measurement bases and outcomes.
    In the subsequent learning phase, the extracted quantum features are learned on a classical computer.
    In this work, we propose the GLQK to leverage the ECP of quantum data, thereby enhancing learning efficiency.
    Specifically, the GLQK is calculated from the quantum features by incorporating local quantum kernels $k_A$ on subsystems of size $\mO(\xi)$ into a polynomial $f$.
    Based on the calculated kernel functions, we optimize the kernel model $h_{\balpha}(S_T(\rho))$ to approximate $g(\rho)$.
    Hyperparameters (e.g., subsystem size, kernel exponent) can be tuned adaptively for a given data distribution using a validation dataset, without requiring additional quantum computational resources, providing a flexible learning framework.
    (Bottom) The validity of the GLQK is guaranteed by the ECP, which enables the approximation of the polynomial $g(\rho)$ by an alternative polynomial $g_{\rm CA}(\rho)$ in local features.    
    Given the kernel construction, the GLQK can represent $g_{\rm CA}(\rho)$ as a linear function within its feature space, which is composed of polynomials in local features, thereby enabling efficient learning.
    }
    \label{fig: overview}
\end{figure*}

Understanding complex quantum many-body phenomena is a pivotal challenge across various fields, including physics, chemistry, and biology. 
Classical computational approaches often struggle to capture the intricate interplay of interactions in these systems due to the exponential dimensionality of the Hilbert space. 
Recent advances in experimental control over quantum systems offer a promising avenue for probing these phenomena. 
Specifically, digital quantum computers~\cite{Arute2019-im} and analog quantum simulators~\cite{Greiner2002-ol} hold the potential to solve classically intractable problems by directly accessing quantum many-body states. 
In parallel, machine learning (ML) has emerged as a novel approach to understanding quantum many-body systems~\cite{Carleo2019-gk}. 
ML techniques have demonstrated remarkable capabilities in capturing complex correlations and patterns within quantum systems, potentially surpassing traditional numerical methods in certain scenarios~\cite{Carleo2017-nw, Gao2017-vn, Pfau2020-fr, Snyder2012-sk, Liu2017-ye, Sanchez-Lengeling2018-vy, Stanev2021-yq}. 
The ability of ML to learn from data and generalize to unseen configurations offers new perspectives and insights that complement traditional theoretical and computational approaches.

The convergence of quantum technologies and ML presents a unique opportunity to accelerate scientific discovery in the realm of quantum many-body physics~\cite{Dunjko2018-jv, Cerezo2022-xy, Acampora2025-pe}.
This synergy allows us to leverage the strengths of both approaches: quantum computers and simulators generate data from complex quantum systems, while ML algorithms analyze and extract meaningful insights from these experimental data.
Theoretical results~\cite{Huang2021-wb, Huang2022-ip, Molteni2024-ja} have demonstrated the existence of quantum many-body problems that can be solved in polynomial time with ML approaches based on data (typically collected from quantum experiments), even when classical algorithms without such data access cannot.
This indicates the potential for exponential advantages from utilizing quantum data.
Classical shadows~\cite{Aaronson2020-uc, Huang2020-ti}, an efficient classical representation of a quantum state, often serve as a crucial input to ML models that learn from quantum data prepared on quantum devices~\cite{Che2024-ue, Wang2022-sq, Du2023-hz, Yao2024-jo, Tang2024-ei, Tang2025-kp}. 
In particular, the shadow kernel method~\cite{Huang2022-ip} has shown the ability to learn quantum phases of matter from classical shadows using polynomial-sized datasets and computation times.
These prior results highlight the fundamental promise of combining quantum technologies and ML, inspiring further investigation to advance this burgeoning field and confront its inherent challenges.

Despite theoretical efficiency, applying this approach to large-scale problems poses a significant challenge due to the polynomial but substantial data requirements and the constrained computational capabilities of near-term quantum devices. 
These limitations hinder the practical scalability of existing techniques. 
For instance, when using the shadow kernel to learn quantum phases, the sample complexity increases as a polynomial with a high degree in the system size $n$~\cite{Huang2022-ip}. 
This fast growth becomes prohibitive for larger quantum systems, restricting the feasibility of these techniques. 
Addressing this challenge is important not only for the development of effective ML algorithms but also for advancing our understanding of the fundamental limits within quantum learning theory.

\renewcommand{\arraystretch}{1.2}
\begin{table*}
\centering
    \caption{
    Learning costs for shadow kernel~\cite{Huang2022-ip} and GLQK.
    The task here is to learn an unknown $m$-body, degree-$p$ polynomial $g(\rho)$ in an $n$-qubit quantum state $\rho$ with correlation length bounded by $\xi$ on a $D$-dimensional hypercubic lattice. 
    The table shows the scaling of the sufficient number of training and validation data points $N_z$ and $N_v$, and shadow size $T$ with respect to the system size $n$, correlation length $\xi$, error $\epsilon$, failure probability $\delta$, and the number of candidate hyperparameter values $|\mP|$, with and without translation symmetry of quantum data.
    This scaling assumes that the weight $m$, degree $p$, norm $\|g\|_1$, and upper bound $C$ of the target polynomial are fixed.
    The quantity $\alpha_g(\leq mp)$ represents the {\it local-cover number} of $g$, characterizing the minimum number of local subsystems required to encompass the support of each term of $g$.
    The quantity $\beta_g (p\leq \beta_g \leq mp)$ denotes the {\it local-factor count} of $g$, characterizing the number of local factors when each term of $g$ is decomposed into the product of local expectation values.    
    These quantities take small values when $g(\rho)$ is local relative to $\mO(\xi)$.
    For instance, if $g(\rho)$ is a sum of local linear/nonlinear quantities (e.g., local Hamiltonian, local purity, local entanglement entropy), $\alpha_g=1$ and $\beta_g=p$.
    See Eqs.~\eqref{eq: LCN} and \eqref{eq: LFC} in Sec.~\ref{sec: methods_LCN} for the detailed definitions of $\alpha_g$ and $\beta_g$.
    The tilde in $\mOt(\cdot)$ hides logarithmic factors in $\delta$.
    }
    \begin{tabular}{wl{0.3cm}wl{3cm}wc{2cm}wc{3cm}wc{3cm}wc{3cm}}
        \hline \hline
         && Trans. sym. & Training ($N_z$) & Validation ($N_v$) & Shadow ($T$) \\
         \hline \\[-10pt]
        &Shadow kernel & No & $\mOt(n^{mp}/\epsilon^4)$ & $\mOt(\log(|\mP|)/\epsilon^2)$ & $\mO(1/\epsilon^2)$ \\[5pt] 
        &GLQK & No & $\mOt(n^{\alpha_g}\xi^{mpD}/\epsilon^4)$ & $\mOt(\log(|\mP|)/\epsilon^2)$ & $\mO(1/\epsilon^2)$  \\[4pt]
        \hline\\[-10pt]
        &Shadow kernel & Yes & $\mOt(n^{mp-\beta_g}/\epsilon^4)$ & $\mOt(\log(|\mP|)/\epsilon^2)$ & $\mO(1/\epsilon^2)$ \\[5pt] 
        &GLQK & Yes & $\mOt(\xi^{mpD}/\epsilon^4)$ & $\mOt(\log(|\mP|)/\epsilon^2)$ & $\mO(1/\epsilon^2)$  \\[4pt]
        
        \hline \hline
    \end{tabular}
    \label{tab: scaling}
\end{table*}

In this paper, we propose an ML framework called {\it Geometrically Local Quantum Kernel} (GLQK), which leverages locality, as characterized by the exponential clustering property (ECP)~\cite{Hastings2010-aq, Hastings2004-vw, Hastings2006-fp, Nachtergaele2006-pg, Lieb1972-mo, Nachtergaele2006-un, Bravyi2006-df, Poulin2010-qy}, to learn efficiently from quantum many-body experimental data with rigorous guarantees.
The ECP, widely observed in noncritical quantum many-body systems, describes the exponential decay of correlations in space, suggesting that quantum information is concentrated in local subsystems of size $\mO(\xi)$, where $\xi$ is the correlation length. 
For such systems, we aim to learn an unknown function $g:\rho\mapsto y$ from data, where $\rho$ is a quantum state with a correlation length bounded by $\xi$, and $y\in\mathbb{R}$.
This problem is typical in supervised learning, and $g(\rho)$ represents an unknown physical property, such as order parameters of unexplored phase transitions.
Here, we primarily focus on the case where $g(\rho)$ is a polynomial of quantum expectation values (though we also extend our analysis to certain non-polynomial functions both theoretically and numerically).
Our ML framework consists of the quantum experiment phase and the classical learning phase (Fig.~\ref{fig: overview}).
In the quantum experiment phase, we prepare quantum data on quantum hardware and extract their features through measurements (e.g., classical shadows).
In the subsequent learning phase, we classically construct the GLQK from these features by incorporating local information on subsystems of size $\mO(\xi)$. 
Kernel hyperparameters associated with typically unknown target characteristics, such as the polynomial degree and correlation length, are optimized via classical post-processing using a validation dataset, thereby providing a flexible learning framework.
Owing to the ECP, this approach enables accurate and efficient learning of $g(\rho)$.
We rigorously prove that the GLQK substantially improves the polynomial sample complexity of the existing shadow kernel in the system size $n$ (Table~\ref{tab: scaling}).
Moreover, when data exhibits translation symmetry, the GLQK achieves constant sample complexity, independent of $n$, showing its outstanding scalability.
Through two numerical experiments on quantum many-body phenomena, including tasks involving non-polynomial functions, we demonstrate improved learning efficiency compared to the shadow kernel and confirm the predicted constant scaling for translationally symmetric data.

Exploiting locality has long been a key principle for reducing the complexity of interacting many-body systems.
By rigorously incorporating the ECP into quantum kernel methods, the GLQK translates this physical structure into provable improvements in learning complexity.
Our results demonstrate that embedding locality at the level of representation fundamentally alters the scaling behavior of quantum data learning, thereby establishing a theoretically scalable and practically viable framework for extracting information from quantum many-body experimental data.

\section{Results}
\subsection{Problem: polynomial learning} \label{sec: problem}
The goal is to learn an unknown function $g:\rho \mapsto  y$ over a data distribution $\mD$ in a supervised learning manner, where $\rho$ is an $n$-qubit quantum state, and $y\in\mathbb{R}$.
Specifically, we aim to obtain a model $h$ that minimizes the expected loss $L_{\mD}(h)=\mathbb{E}_{\rho\sim\mD}[(g(\rho)-h(\rho))^2]$.
A training dataset comprises $T$ copies of $N$ quantum states and their corresponding labels: $\{\rho_i^{\otimes T}, y_i\}_{i=1}^N$, where each $\rho_i$ is sampled from $\mD$ and $y_i=g(\rho_i)$.
This problem setting can be applied not only to regression tasks but also to classification tasks, where $g(\rho)=0$ serves as the decision boundary, and the sign of $g(\rho)$ corresponds to the class label.

We make two assumptions for this task. 
First, we assume that any quantum state $\rho$ sampled from $\mD$ satisfies the ECP with a correlation length bounded by $\xi$ (see the next subsection for details). 
Second, we assume that $g(\rho)$ can be represented as a polynomial in $\rho$ (though we will also consider certain non-polynomial functions later to demonstrate the broader applicability of our framework).
To characterize the polynomial, we define an $m$-body, degree-$p$ polynomial as follows:
\begin{dfn}[$m$-body, degree-$p$ polynomial] \label{def: polynomial in quantum state}
Consider the following function $g(\rho)$ of an $n$-qubit quantum state $\rho$:
\begin{align}
g(\rho) = \sum_i c_i \prod_{j=1}^p \tr \left[ P_{ij} \rho \right], \label{eq: m-body degree-p polynomial}
\end{align}
where $P_{ij}$ is an $n$-qubit Pauli string, and $c_i$ is an expansion coefficient.
If the Pauli weights of all $P_{ij}$'s are less than or equal to $m$, we say that $g(\rho)$ is an $m$-body, degree-$p$ polynomial in $\rho$.
We also define the $\ell_1$-norm of Pauli coefficients as $\|g\|_{1}=\sum_i |c_i|$.
\end{dfn}
This definition encompasses various physically important quantities, consisting of linear ones with $p=1$ (e.g., energy, magnetization, correlation functions) and nonlinear ones with $p\geq 2$ (e.g., purity).
Furthermore, it can approximate logarithmic and exponential functions by truncating their high-degree terms in $\rho$.
This allows for representing, for example, von Neumann entropy and (topological) entanglement entropy with arbitrary accuracy.

\subsection{Exponential clustering and cluster approximation} \label{sec: exponential clustering property}
The ECP is a fairly generic quantum many-body phenomenon that describes the exponential decay of correlations, typically arising from the locality of quantum systems~\cite{Hastings2010-aq, Hastings2004-vw, Hastings2006-fp, Nachtergaele2006-pg, Lieb1972-mo, Nachtergaele2006-un, Bravyi2006-df, Poulin2010-qy}.
Leveraging locality can enhance the efficiency of many quantum algorithms by reducing problems across the entire Hilbert space to those concerning smaller subspaces~\cite{Cerezo2021-tq, Cerezo2023-hz, Thanasilp2024-qc, Cramer2010-yw, Votto2025-ij, Guo2024-ly, Vermersch2024-bg, Rouze2024-gm, Mizuta2022-ua, Kanasugi2023-te, Huang2024-vx}.
Although there exist ML algorithms that leverage locality to learn unknown properties of quantum systems, they assume specific situations or lack theoretical guarantees~\cite{Lewis2024-yh, Wanner2024-ei, Wu2024-az}.
Our work offers a provably efficient framework applicable to more general situations. 
See Sec.~\ref{sec_sm: related_works} for detailed background.

To formalize the ECP, let us consider an $n$-qubit quantum state $\rho$ on the $D$-dimensional hypercubic lattice (one can easily extend the results of this paper to general lattices).
We say that $\rho$ satisfies the ECP if the following inequality holds for any observables $O_A$ and $O_B$, each acting on subsystems $A$ and $B$, respectively ($A,B\subseteq [n]$, $[n]=\{1,\cdots,n\}$ denotes the set of $n$ qubits):
\begin{align}
    \left| \braket{O_A O_B} - \braket{O_A}\braket{O_B}\right| 
    \leq \| O_A \|_{S} \| O_B \|_{S} \, e^{-\dist(A,B)/\xi},
\end{align}
where $\dist(A,B)$ is the shortest distance between $A$ and $B$ on the lattice, $\xi$ is the correlation length, $\braket{X}=\tr(X\rho)$ is the expectation value, and $\|X\|_{S}$ denotes the spectral norm.
This property indicates that quantum correlations decay exponentially in distance, justifying the approximation of $\braket{O_A O_B} \approx \braket{O_A}\braket{O_B}$ for any observables $O_A$ and $O_B$ with $\dist(A,B) \gg \xi$.

Based on this property, we introduce the cluster approximation of the polynomial $g(\rho)=\sum_i c_i \prod_{j} \tr \left[ P_{ij} \rho \right]$, which is crucial in understanding the validity of the GLQK.
The cluster approximation, characterized by a parameter $\Delta$, decomposes the support of $P_{ij}$, denoted as $\supp(P_{ij})$, into some clusters such that qubits within a distance $\Delta$ are grouped into the same cluster, while distinct clusters are separated by a distance of at least $\Delta$.
Let $P_{ijk}$ be the partial Pauli string of $P_{ij}$ acting on the $k$th cluster, i.e., $P_{ij}=P_{ij1}\otimes P_{ij2}\otimes \cdots$.
Then, the $\Delta$-cluster approximation of $g(\rho)$ is defined as
\begin{align}
&g_{\rm CA}(\rho) 
= \sum_i c_i \prod_{j=1}^p \prod_{k} \tr \left[ P_{ijk} \rho \right]. \label{eq: cluster approximation}
\end{align}

For quantum states exhibiting finite correlation length, $g_{\rm CA}(\rho)$ well approximates the original polynomial $g(\rho)$ if $\Delta$ is sufficiently large, since correlations between clusters, $\supp(P_{ijk})$, are suppressed exponentially in $\Delta$.
The following lemma quantifies this fact, implying that quantum information is concentrated in local subsystems of size $\mO(\xi)$ (the proof is provided in Sec.~\ref{sec_sm: clustering}):
\begin{lem} \label{lem: approx} 
    Let $g(\rho)$ be an $m$-body, degree-$p$ polynomial.
    For any $\epsilon$ and $\xi$, the $\Delta$-cluster approximation $g_{\rm CA}(\rho)$ with $\Delta=\xi \log(\|g\|_{1} mp/\epsilon)$ satisfies
    \begin{align}
        &\left| g(\rho) - g_{\rm CA}(\rho) \right| \leq \epsilon, 
    \end{align}
    for any $\rho$ with a correlation length bounded by $\xi$.
\end{lem}    

The cluster approximation and Lemma~\ref{lem: approx} underpin the validity of the GLQK.
To show this, we define a set of local subsystems $\AGL(\zeta)$ as
\begin{align}
\AGL(\zeta) = \{A_i(\zeta) \subseteq [n] \,|\, i\in [n]\},
\end{align}
where $A_i(\zeta)$ is the $D$-dimensional hypercubic local subsystem with side length $\zeta$ and corner at the $i$th qubit.
In Eq.~\eqref{eq: cluster approximation}, one can easily show that each cluster, $\supp(P_{ijk})$, is encompassed by some $A\in\AGL(\zeta)$ of size $\zeta\geq m\Delta$, since the number of qubits included in each cluster is at most $m$, and the distance between neighboring qubits within the cluster is less than $\Delta$.
Combined with Lemma~\ref{lem: approx}, the value of any polynomial $g(\rho)$ can be evaluated with error $\epsilon$ only from local reduced density matrices on $\AGL(\zeta)$ of size $\zeta\geq m\Delta=m\xi \log(\|g\|_1 mp/\epsilon)$.
This result ensures the validity of the GLQK, which learns from local information of quantum data on subsystems $\AGL(\zeta)$.

\begin{figure*}[t]
    \centering
    \includegraphics[width=\linewidth]{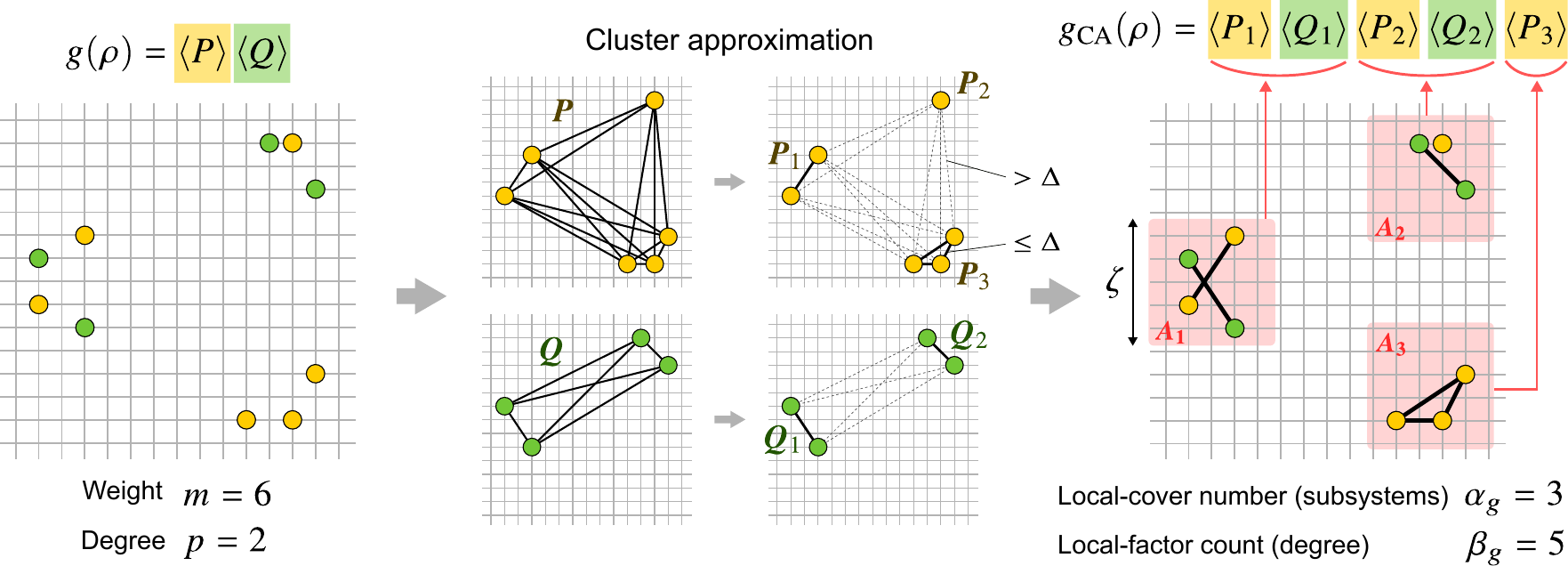}
    \caption{
    An example of cluster approximation.
    Here, we consider $g(\rho)=\braket{P}\braket{Q}$, where $P$ and $Q$ are Pauli strings acting on qubits denoted as yellow and green circles, respectively.
    The weight $m$ and degree $p$ of this polynomial are $6$ and $2$, respectively, since the number of qubits on which $P$ and $Q$ act is bounded by six, and $g(\rho)$ is the product of two quantum expectation values.
    The cluster approximation decomposes the support of each Pauli string into clusters by grouping qubits within a distance $\Delta$ and partitioning qubits separated by distances greater than $\Delta$ into distinct clusters.
    Based on this decomposition, we define $P_1,P_2,P_3, Q_1$, and $Q_2$ as Pauli strings acting on each cluster such that $P=P_1\otimes P_2\otimes P_3$ and $Q=Q_1\otimes Q_2$.
    The $\Delta$-cluster approximation of $g(\rho)$ is given by $g_{\rm CA}(\rho)=\braket{P_1}\braket{P_2}\braket{P_3}\braket{Q_1}\braket{Q_2}$.
    The local-cover number $\alpha_g$ is defined as the minimum number of local subsystems in $\AGL(\zeta)$ required to encompass the support of all Pauli strings (i.e., the number of red boxes), and the local-factor count $\beta_g$ is defined as the number of factors (i.e., degree) of $g_{\rm CA}(\rho)$.
    Then, $g_{\rm CA}(\rho)$ is represented as the product of $\alpha_g=3$ local linear/nonlinear quantities: $\braket{P_1}\braket{Q_1}$, $\braket{P_2}\braket{Q_2}$, and $\braket{P_3}$.
    }
    \label{fig: cluster}
\end{figure*}

We introduce two quantities of $g$ crucial for the learning cost scaling (see Fig.~\ref{fig: cluster} and Sec.~\ref{sec: methods_LCN} for details).
The first is the {\it local-cover number} $\alpha_g=\LCN(g;\Delta,\zeta)$, which characterizes the locality of the $\Delta$-cluster approximation $g_{\rm CA}$ relative to the scale $\zeta$. 
It is defined as the minimum number of local subsystems in $\AGL(\zeta)$ needed to cover the support of each term of $g_{\rm CA}$, satisfying $\alpha_g \leq mp$.
This means each term can be represented as the product of $\alpha_g$ local linear/nonlinear quantities.
For instance, sums of local quantities (e.g., local Hamiltonian, local purity, local entanglement entropy) correspond to $\alpha_g=1$ if $\zeta$ is sufficiently large to cover each term, while $t$-point correlation functions satisfy $\alpha_g=t$ in general. 
The second is the {\it local-factor count} $\beta_g=\LFC(g;\Delta)$, which roughly corresponds to the degree of $g_{\rm CA}$, satisfying $p\leq \beta_g \leq mp$.
This quantity takes a small value ($\sim p$) when $g(\rho)$ is local [more generally, when $\supp(P_{ij})$ is local] compared to $\Delta$.

\subsection{General learning framework} \label{sec: GLQK}
Our learning framework consists of the quantum experiment phase and the classical learning phase (Fig.~\ref{fig: overview}).
In the quantum experiment phase, we prepare quantum data $\rho$ on quantum hardware and then measure it based on a predefined protocol, extracting quantum features of data as a record of measurement bases and outcomes.
A promising approach is classical shadow tomography via random Pauli measurements~\cite{Aaronson2020-uc, Huang2020-ti}.
This method enables obtaining an efficient classical representation of $\rho$ by repeatedly measuring each qubit of $\rho$ on a random Pauli basis $W_i\in\{X_i, Y_i, Z_i\}$ and recording the outcome $o_i=\pm1$ over $T$ copies ($i$ is the qubit index). 
Let $S_T(\rho)$ denote this record.
The original quantum state $\rho$ can be reproduced from the measurement results as $\rho=\mathbb{E}[\sigma]$, where $\sigma =\sigma_1\otimes \cdots \otimes \sigma_n$ with $\sigma_i=(3o_iW_i+I)/2$ is a classical shadow for $\rho$.
Then, the training dataset $\{\rho_i^{\otimes T},y_i\}_{i=1}^N$ is converted to $\{S_T(\rho_i),y_i\}_{i=1}^N$, where $y_i=g(\rho_i)$.
While this work primarily employs classical shadows for simplicity, our framework is not restricted to them.
For example, quantum tomography based on tensor networks may offer more efficient representations of quantum states by exploiting locality in entanglement~\cite{Cramer2010-yw, Votto2025-ij, Guo2024-ly, Vermersch2024-bg, Rouze2024-gm}.

In the subsequent learning phase, we learn $g(\rho)$ on a classical computer from the quantum features obtained in the quantum experiments.
The GLQK is a general quantum kernel framework that exploits the locality of quantum data to enhance learning efficiency (see Sec.~\ref{sec: methods_kernel} for kernel methods).
The main idea is based on the observation that any polynomial $g(\rho)$ can be approximated with an alternative polynomial $g_{\rm CA}(\rho)$ in local expectation values $\tr(P_{ijk}\rho)$.
This observation motivates constructing a quantum kernel whose feature space consists of polynomials in local quantities.   
Given a set of local subsystems $\AGL(\zeta)$, we define the GLQK for classical shadows $S_T(\rho)$ and $S_T(\tilde{\rho})$ as a polynomial in local quantum kernels:
\begin{align}
    k_{\rm GL}(S_T(\rho),S_T(\tilde{\rho})) 
    = f(\{k_A(S_{T}(\rho),S_{T}(\tilde{\rho})) | A\in \AGL(\zeta)\}), \label{eq: general GLQK}
\end{align}
where $f(x_1,x_2,\ldots)=\sum_{i_1,i_2,\ldots} c_{i_1i_2\cdots}x_1^{i_1}x_2^{i_2}\cdots$ is any polynomial with non-negative coefficients $c_{i_1i_2\ldots}\geq0$ (including infinite series like exponential), and $k_A$ is any local quantum kernel defined on the subsystem $A$ (e.g., fidelity kernel~\cite{Havlicek2019-wd, Schuld2019-hm} and shadow kernel~\cite{Huang2022-ip}). 
This definition includes projected quantum kernels~\cite{Huang2021-wb}, such as $\sum_{k=1}^n \tr[\rho_k \tilde{\rho}_k]$, where $\rho_k$ and $\tilde{\rho}_k$ are the reduced density matrices at the $k$th qubit.
Furthermore, when $\zeta$ equals the entire system size, the GLQK reduces to the underlying quantum kernel $k_A$ with $A=[n]$, which does not incorporate any locality structure. This guarantees a performance lower bound given by $k_{[n]}$, up to a logarithmic overhead arising from the optimization over $\zeta$ (see Theorems~\ref{thm: main theorem 1} and \ref{thm: main theorem 2}).

To understand the capability of the GLQK, let us consider its feature space.
A straightforward calculation reveals the feature vector of the GLQK as follows:
\begin{align}
    \phi_{\rm GL}(S_T(\rho)) = \tilde{f}(\{\phi_A(S_T(\rho)) | A\in \AGL(\zeta )\}), \label{eq: feature vector}
\end{align}
where $\tilde{f}(x_1,x_2,\ldots)=\bigoplus_{i_1,i_2,\ldots} \sqrt{c_{i_1i_2\ldots}} x_1^{\otimes i_1} \otimes x_2^{\otimes i_2}\otimes \cdots$, and $\phi_A$ is the feature vector of the local kernel $k_A$.
Thus, $\phi_{\rm GL}$ incorporates polynomials of local features at the length scale $\zeta$. 
Given appropriate $f$ and $k_A$ with sufficiently large $\zeta$, this feature space structure, coupled with Lemma~\ref{lem: approx}, enables learning any polynomial $g(\rho)$ via the cluster approximation $g_{\rm CA}(\rho)$, even when nonlocal terms are present.

To improve flexibility, we adaptively optimize several hyperparameters on a classical computer for each dataset.
The hyperparameters include the local subsystem size $\zeta$, kernel parameters (as shown below), and the regularization parameter, which affect the kernel model's complexity and learning efficiency.
Optimizing these hyperparameters enables us to flexibly learn various targets without prior knowledge of their characteristics, such as correlation length, polynomial degree, and weight. 
Specifically, we optimize the hyperparameters by selecting the values that minimize validation error on a validation dataset from a finite set of candidate hyperparameter values $\mP$.
In this approach, the available dataset is partitioned into two subsets.
One is used to train each ML model with different hyperparameter values, and the other is used to validate their performance and select the best-performing model.
We denote the training and validation datasets as $Z$ and $V$, consisting of $N_z$ and $N_v$ data points, respectively.
Also, let $h_{ZV}^\ast$ be the kernel model obtained through these training and validation processes. 
Importantly, this hyperparameter optimization incurs no additional quantum computational cost, as it relies solely on the classical representation of quantum features.
This is particularly advantageous in the near term, when quantum data acquisition, rather than the additional classical computational cost from hyperparameter optimization, is the primary bottleneck.

\subsection{Polynomial GLQK}
The design of $f$ and $k_A$ in Eq.~\eqref{eq: general GLQK} is critical for achieving high learning efficiency and broad applicability.
Here, as a versatile yet powerful kernel, we propose the {\it polynomial GLQK}, equipped with the {\it truncated shadow kernel}.
Combined with the cluster approximation, this kernel can represent any polynomial $g(\rho)$ as a linear function of local quantities within the feature space, thereby enabling efficient learning.
The polynomial GLQK is defined as
\begin{align}
    k_{\rm GL}(S_{T}(\rho),S_{T}(\tilde{\rho}))
    = \left( \frac{1}{|\AGL(\zeta)|} \sum_{A\in \AGL(\zeta)} k_A(S_{T}(\rho),S_{T}(\tilde{\rho})) \right)^\kappa, \label{eq: horder GLQK}
\end{align}
where $\kappa\geq 1$ is an integer hyperparameter, and $|\AGL(\zeta)|$ is the cardinality of $\AGL(\zeta)$.
Given Eq.~\eqref{eq: feature vector}, the feature space of this GLQK includes the product of $\kappa$ local features: $\phi_{A_{1}}\otimes\cdots\otimes\phi_{A_{\kappa}}$ for $\forall A_1,\cdots,A_\kappa \in \AGL(\zeta)$.
As mentioned above, $\kappa$ can be optimized without requiring additional quantum computational resources.

As a local quantum kernel $k_A$, we propose the following truncated shadow kernel that can represent any local polynomial within its feature space ($k_A$ can be any other kernel in general):
\begin{align}
    k^{{\rm TSK}}_A(S_{T}(\rho),S_{T}(\tilde{\rho})) 
    = \exp\left( \frac{\tau}{T^2}\sum_{t,t'=1}^T \prod_{i\in A} \left[1+\frac{\gamma}{|A|} \tr(\sigma_i^{(t)}\tilde{\sigma}_i^{(t')} )\right]\right), \label{eq: truncated shadow kernel}
\end{align}
where $\sigma_i^{(t)}$ denotes the classical shadow of the $i$th qubit at the $t$th measurement shot, and $\tau,\gamma>0$ are real hyperparameters.
The classical computation time for this kernel is $\mO(|A|T^2)$, which results in the overall computation time of $\mO(n |A| T^2)$ for the polynomial GLQK.
Notably, the feature vector of this kernel incorporates arbitrarily large reduced density matrices within $A$ and arbitrarily high-degree polynomials thereof, while truncating unphysical feature components present in the standard shadow kernel (see Sec.~\ref{sec: methods_TSK}).

Given these feature space structures and Lemma~\ref{lem: approx}, the polynomial GLQK based on the truncated shadow kernel can represent any polynomial $g(\rho)$ with error $\epsilon$ as a linear function within the feature space by setting $\zeta\geq m\Delta$ and $\kappa\geq \alpha_{g}=\LCN(g;\Delta,\zeta)$ with $\Delta=\xi\log(\|g\|_1mp/\epsilon)$.
This universality is demonstrated by approximating $g(\rho)$ with $g_{\rm CA}(\rho)$, where each term is represented as the product of $\alpha_g$ local quantities, and by considering the GLQK's feature space, which consists of products of $\kappa$ local quantities. 
Consequently, the GLQK can learn any polynomial by tuning $\zeta$ and $\kappa$ for a given dataset, provided there are sufficient training samples.

\subsection{Rigorous resource estimation for general states}
Here, we show that the GLQK achieves high scalability with respect to $n$ by effectively eliminating irrelevant nonlocal contributions from the feature space.
From a technical perspective, we assume that an upper bound of $|g(\rho)|$ is known: $|g(\rho)|\leq C$ for all $\rho$.
In many physical settings, such a bound is known a priori from physical principles,
or can be reasonably estimated from the training data.
Based on this bound, we clip the prediction of the trained kernel model $h$ as $h(S_T(\rho)) \to \clip(h(S_T(\rho)))$ to improve robustness, where $\clip(t) = t$ if $|t|\leq C$ and $\clip(t) = \sgn(t) C$ if $|t|\geq C$.
In kernel ridge regression, the following theorem quantifies the learning cost scaling for the GLQK to achieve $L_{\mD}(h_{ZV}^\ast)=\mathbb{E}_{\rho\sim\mD}[(g(\rho)- \clip(h_{ZV}^\ast(S_T(\rho))))^2] \leq \epsilon^2$, where $h_{ZV}^\ast$ denotes the kernel model obtained from training and validation datasets $Z$ and $V$ (see Sec.~\ref{sec_sm: GLQK} for the formal version and proof):

\begin{thm}[Informal] \label{thm: main theorem 1}
Consider an $m$-body, degree-$p$ polynomial $g(\rho)$ of an $n$-qubit quantum state $\rho$ with a correlation length bounded by $\xi$ on a $D$-dimensional hypercubic lattice.
Assume that the candidate hyperparameter set $\mP$ contains values $\zeta = \mO(\xi \log(1/\epsilon))$ and $\kappa=\alpha_g$.
Suppose that we obtain $N=N_z+N_v$ classical shadows of size $T$ as training and validation datasets such that
\begin{align}
    &N_z = \mOt(n^{\alpha_g} \xi^{mpD}/\epsilon^4), \\
    &N_v = \mOt(\log(|\mP|)/\epsilon^2), \\
    &T= \mO(1/\epsilon^2),
\end{align}
where $\mOt(\cdot)$ hides logarithmic factors in $\delta$.
Then, kernel ridge regression trained on $N_z$ training data using the polynomial GLQK based on the truncated shadow kernel, combined with hyperparameter optimization via $N_v$ validation data, achieves
\begin{align}
    L_{\mD}(h_{ZV}^\ast) \leq \epsilon^2
\end{align}
with probability at least $1-\delta$.
\end{thm}    

This theorem has several important implications.
First, it shows that the number of training data $N_z$ required to learn any polynomial using the GLQK scales as $\mOt(n^{\alpha_g}\xi^{mpD}/\epsilon^4)$.
Focusing on the dependence on $n$, this scaling improves upon that of the conventional shadow kernel.
In Sec.~\ref{sec_sm: SK}, we perform a similar resource analysis for the shadow kernel and show that $N_z=\mOt(n^{mp}/\epsilon^4)$ and $N_v=\mOt(\log(|\mP|)/\epsilon^2)$ classical shadows of size $T=\mO(1/\epsilon^2)$ suffice to learn any $m$-body, degree-$p$ polynomial.
Since $\alpha_g \leq mp$, the GLQK improves the sample complexity of the shadow kernel with respect to $n$. 
This improvement is particularly pronounced when the target polynomial has small $\alpha_g$.
For example, sums of local linear or nonlinear quantities within a scale $\zeta = \mO(\xi \log(1/\epsilon))$, satisfying $\alpha_g=1$, can be learned with $\mOt(n\xi^{mpD}/\epsilon^4)$ training data using the GLQK.
Furthermore, the dependence on $\xi$ is polynomial.
Although the factor $\xi^{mpD}$ may appear to yield worse scaling than the shadow kernel for large (potentially infinite) $\xi$, this does not represent a fundamental disadvantage. 
This is because when $\zeta$ equals the entire system size, the GLQK reduces to the (truncated) shadow kernel, and hence its sample complexity is likewise bounded by $\mOt(n^{mp}/\epsilon^4)$.
Note that the estimated amounts of $N_z$ and $T$ are (super) exponential in $m$ and $p$ for both the GLQK and the shadow kernel.
Thus, they are efficient in learning few-body, low-degree polynomials.

Second, the number of validation samples $N_v$ scales logarithmically with the number of candidate hyperparameter values $|\mP|$, indicating that hyperparameter optimization introduces only a modest overhead in sample complexity.
In situations where the correlation length $\xi$, the degree $p$, and the weight $m$ are bounded, only a finite set of hyperparameters needs to be considered.
Even if $|\mP|$ grows polynomially with the system size $n$, the logarithmic dependence ensures that the additional validation cost incurs only a $\log(n)$ overhead.

\subsection{Rigorous resource estimation for translationally symmetric states}
Translationally symmetric systems are a vast and practically important class of quantum many-body systems, which are encountered, e.g., in solids, artificial quantum systems, and lattice gauge theories.
Imposing translation symmetry on $\rho$ further improves the learning efficiency of the GLQK, achieving constant sample complexity in $n$.
This is guaranteed by the following theorem (see Sec.~\ref{sec_sm: GLQK} for the formal version and proof):
\begin{thm}[Informal] \label{thm: main theorem 2}
Consider an $m$-body, degree-$p$ polynomial $g(\rho)$ of an $n$-qubit translationally symmetric quantum state $\rho$ with a correlation length bounded by $\xi$ on a $D$-dimensional hypercubic lattice.
Assume that the candidate hyperparameter set $\mP$ contains values $\zeta = \mO(\xi \log(1/\epsilon))$ and $\kappa=1$.
Suppose that we obtain $N=N_z+N_v$ classical shadows of size $T$ as training and validation datasets such that
\begin{align}
    &N_z = \mOt(\xi^{mpD}/\epsilon^4), \\
    &N_v = \mOt(\log(|\mP|)/\epsilon^2), \\
    &T= \mO(1/\epsilon^2),
\end{align}
where $\mOt(\cdot)$ hides logarithmic factors in $\delta$.
Then, kernel ridge regression trained on $N_z$ training data using the polynomial GLQK based on the truncated shadow kernel, combined with hyperparameter optimization via $N_v$ validation data, achieves
\begin{align}
    L_{\mD}(h_{ZV}^\ast) \leq \epsilon^2
\end{align}
with probability at least $1-\delta$.
\end{thm}    
This theorem shows the GLQK's high scalability: a constant number of data $N=N_z+N_v$, independent of $n$, is sufficient to learn any polynomial $g(\rho)$ from translationally symmetric data, provided that $|\mP|$ does not scale with $n$.
Even if $|\mP|$ grows polynomially with $n$, hyperparameter optimization incurs only a logarithmic overhead in $n$.
This significantly improves the learning cost of the shadow kernel, where $N_z=\mOt(n^{mp-\beta_g}/\epsilon^4)$ training data suffice to learn the polynomial, as shown in Sec.~\ref{sec_sm: SK}.
This improvement is particularly pronounced for polynomials with small $\beta_g$, i.e., when $g(\rho)$ is local relative to $\Delta=\mO(\xi)$.

The improved scalability over the shadow kernel in Theorems~\ref{thm: main theorem 1} and \ref{thm: main theorem 2} stems from the reduced dimensionality of the feature space. Unlike the shadow kernel, which encompasses all polynomials within its feature space, the polynomial GLQK incorporates only local features and their polynomials, resulting in efficient learning (see Sec.~\ref{sec: methods_theorems} for details). 
Notably, this restriction in the GLQK never sacrifices learning universality due to the ECP. 

\subsection{Beyond polynomial functions} \label{sec: Beyond polynomial functions}
Although our theorems focus on learning polynomial functions, the GLQK can in principle learn non-polynomial functions, provided that they admit accurate approximations by truncated polynomial expansions. Importantly, for the GLQK, the asymptotic scaling of the sample complexity with respect to $n$ does not explicitly depend on the polynomial degree $p$ required for the approximation. This contrasts with the shadow kernel, for which the $n$-dependence of the sample complexity explicitly depends on $p$. Nevertheless, for both kernels, the $n$-independent prefactor in the sample complexity bound scales super-exponentially with $p$, which can result in prohibitive data requirements in practice.

These bounds arise from worst-case analyses for arbitrary polynomials in Theorems~\ref{thm: main theorem 1} and \ref{thm: main theorem 2}. By exploiting the specific structure of the target function, one may obtain substantially more favorable sample-complexity guarantees. For instance, in Sec.~\ref{sec_sm: nonpolynomial}, we theoretically show that the GLQK with the exponential-form truncated shadow kernel can efficiently learn exponential functions of the form $g(\rho)=\sum_i c_i \exp(\tr(P_i\rho))$ directly, without relying on a truncated polynomial approximation. This advantage stems from the structural alignment between the kernel feature space and the series expansion of the target function, highlighting the importance of designing kernels that are tailored to the target problem.

\subsection{Numerical experiment (Random quantum dynamics)} \label{sec: numerical experiments}

We numerically demonstrate the GLQK's high scalability in the regression task of $g(\rho)$ for quantum states generated by random quantum dynamics~\cite{Caro2023-hv, Chinzei2025-pa} (see Sec.~\ref{sec: methods_numerics1} for details).
To investigate the impact of translation symmetry, we explore two types of random local Hamiltonians $H_1$ and $H_2$, where $H_1$ ($H_2$) is (not) translationally symmetric.
For $k=1,2$, given an initial product state $\ket{\phi_k}$ (that is translationally symmetric for $k=1$), we consider quantum dynamics $\ket{\psi_k} = e^{-iH_k t} \ket{\phi_k}$.
Here, $\ket{\psi_k}$ is used as quantum data in this task.
We generate quantum data by randomly sampling the local Hamiltonian $H_k$ and the initial product state $\ket{\phi_k}$, thereby defining the data distribution $\mD_k$ of $\ket{\psi_k}$. 
Both $N$ training data and $M$ test data are independently sampled from $\mD_k$.
The finite evolution time $t$ ensures the ECP of $\ket{\psi_k}$, suggesting that the GLQK is suitable for this task~\cite{Lieb1972-mo, Nachtergaele2006-un, Bravyi2006-df, Poulin2010-qy}. 
Furthermore, the translation symmetry of $\ket{\psi_1}$ implies that the GLQK is likely to be even more effective for $\mD_1$.
For quantum data $\rho=\ket{\psi_k}\bra{\psi_k}$, we consider three types of polynomials and two types of non-polynomials: local linear function $g_1(\rho)= \braket{X_1 Y_2}$, local quadratic function $g_2(\rho)=\braket{X_1 X_2}\braket{Y_1 Y_2}$, nonlocal linear correlation function $g_3(\rho)=\braket{X_1 Y_{n/2+1}}$, entanglement entropy $g_4(\rho)=-\tr(\rho_{1,2}\log \rho_{1,2})$, and exponential function $g_5(\rho)=\exp(\braket{X_1Y_2})$.
Kernel ridge regression~\cite{Shalev-Shwartz2014-rj} is invoked to solve this task, based on the conventional shadow kernel and the polynomial GLQK with the truncated shadow kernel of Eqs.~\eqref{eq: horder GLQK} and \eqref{eq: truncated shadow kernel}.

\begin{figure*}[h]
    \centering
    \includegraphics[width=\linewidth]{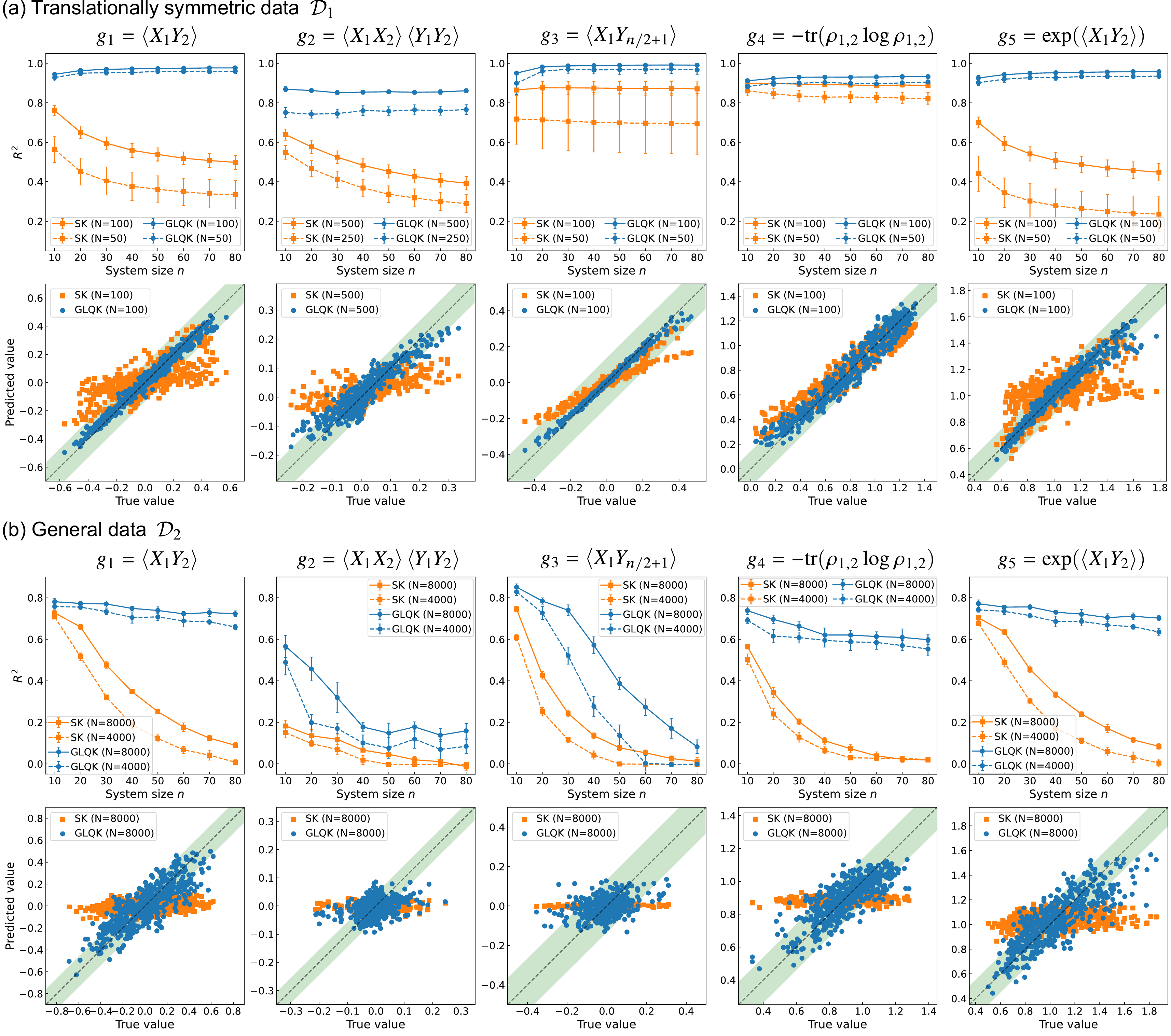}
    \caption{
Numerical results of the regression task involving random quantum dynamics.
Results are shown for (a) translationally symmetric and (b) general data distributions, comparing the shadow kernel (SK, orange squares) and GLQK (blue circles).
The top panels show the coefficient of determination $R^2=1 - \sum_{i=1}^{M} (y_i - f_i)^2/\sum_{i=1}^{M} (y_i - \bar{y})^2$, where $y_i=g(\rho_i)$ and $f_i$ are the true and predicted values for the $i$th test data $\rho_i$, and $\bar{y}=\sum_{i=1}^{M} y_i/M$.
A larger $R^2$ indicates better accuracy, with $R^2=1$ denoting perfect prediction.
Data points represent the mean value of $R^2$ across 10 randomly sampled training and test datasets, and error bars represent the corresponding standard deviation.
The bottom panels display scatter plots for $n=80$.
The horizontal and vertical axes represent the true and predicted values of $g(\rho_i)$, respectively.
The green shaded areas depict $[\sum_{i=1}^M (g(\rho_i)-g(\sigma_i))^2/M]^{1/2}$, the statistical error originating from the finite shadow size $T$, independent of the kernel ridge regression.
Here, $\sigma_i$ is the density matrix estimated from the classical shadow for test data $\rho_i$.
The number of training data points is $N$, and the shadow size is fixed at $T=500$.
    }
    \label{fig: regression}
\end{figure*}

The left three columns of Fig.~\ref{fig: regression} demonstrate improved learning efficiency of the GLQK over the shadow kernel across various qubit numbers and three target polynomials, $g_1, g_2$, and $g_3$, for both $\mD_1$ and $\mD_2$.
Although the nonlocal quantity $g_3(\rho)=\braket{X_1Y_{n/2+1}}$ is not explicitly included in the feature space of the GLQK, it can still be learned due to the cluster approximation $\braket{X_1Y_{n/2+1}}\approx \braket{X_1}\braket{Y_{n/2+1}}$.
The scatter plots of true versus predicted values on the test data further provide evidence of the improved performance of the GLQK.
The improvement is particularly obvious for $\mD_1$, where the data exhibit translation symmetry.
In Fig.~\ref{fig: regression} (a), the prediction accuracy of the GLQK remains high even as the system size $n$ increases, indicating that $N=\mO(1)$ classical shadows of size $T=\mO(1)$ suffice to learn the polynomials from translationally symmetric data. 
This constant sample complexity contrasts sharply with the shadow kernel, where the prediction accuracy for $g_1$ and $g_2$ degrades with increasing $n$. 
Note that the accuracy of the shadow kernel for $g_3$ with $m=\beta_g=2$ and $p=1$ does not significantly decrease, as indicated by $\mO(n^{mp-\beta_g})=\mO(1)$.
In Fig.~\ref{fig: regression} (b), the GLQK continues to exhibit better performance even for $\mD_2$, although its accuracy is no longer independent of $n$.
Despite the relative advantage of the GLQK, both kernels exhibit low accuracy for $g_2$ and $g_3$ at large system sizes.
We interpret this behavior as reflecting the intrinsic difficulty of learning such quantities in the absence of symmetries.

The right two columns of Fig.~\ref{fig: regression} show the results for non-polynomials, $g_4$ and $g_5$.
Remarkably, the GLQK performs well even on such non-polynomials and maintains its accuracy as $n$ increases for translationally symmetric data.
Even in the absence of translation symmetry, the decrease in the accuracy of the GLQK is not large compared to that for $g_1, g_2$, and $g_3$.
On the other hand, except for $g_4$ with translation symmetry, the accuracy of the shadow kernel decreases rapidly with $n$.
The mechanism by which the GLQK can efficiently learn the exponential function $g_5$ can be understood through the compatibility between the feature-space structure and the Taylor expansion of the target function (see Sec.~\ref{sec_sm: nonpolynomial}).
However, the mechanism underlying the high accuracy of the GLQK for the entanglement entropy $g_4$, which is of logarithmic form, is theoretically subtle, because the logarithm is non-analytic at the origin and thus admits no Taylor expansion.
Nevertheless, these results highlight the practical effectiveness of the GLQK beyond the theoretical guarantees.

In Sec.~\ref{sec_sm: shadow_size}, we also demonstrate that increasing the shadow size $T$ leads to systematic improvements in learning accuracy.
Furthermore, in Sec.~\ref{sec_sm: robustness}, we present numerical results, supported by theoretical considerations, suggesting that the constant sample complexity of the GLQK is robust against weak disorders that break translation symmetry, as long as the disorder strength is sufficiently small relative to the desired accuracy.

\subsection{Numerical experiment (Quantum phase recognition)}

We also tackle quantum phase recognition, a more practical and application-oriented task~\cite{Cong2019-ov, Chinzei2024-nm} (see Sec.~\ref{sec: methods_numerics2} for details).
Let us consider the bond-alternating XXZ model $H(J)$, where $J$ is the interaction strength.
The ground state of this Hamiltonian, $\ket{\phi(J)}$, exhibits a quantum phase transition from the trivial phase to the symmetry-protected topological (SPT) phase at $J\approx 1$. 
The task here is to classify noisy ground-state data into these two phases.
We use locally disturbed ground states as quantum data: $\ket{\tilde{\phi}(J)} = R \ket{\phi(J)}$, where $R$ is a random local unitary.
Both $N$ training data and $M$ test data are independently generated by randomly sampling $J$ and $R$.
We solve this classification task using the support vector machine~\cite{Cortes1995-ic} with the shadow kernel and the polynomial GLQK, where the target polynomial $g(\rho)$ is the effective ``order parameter,'' and $g(\rho)=0$ corresponds to the phase boundary [Fig.~\ref{fig: QPR} (a)].
The topological order parameter for this transition is defined on a local subsystem of size $\mO(\xi)$~\cite{Elben2020-ip}, suggesting the validity of the GLQK [Fig.~\ref{fig: QPR} (b)].

\begin{figure*}[t]
    \centering
    \includegraphics[width=\linewidth]{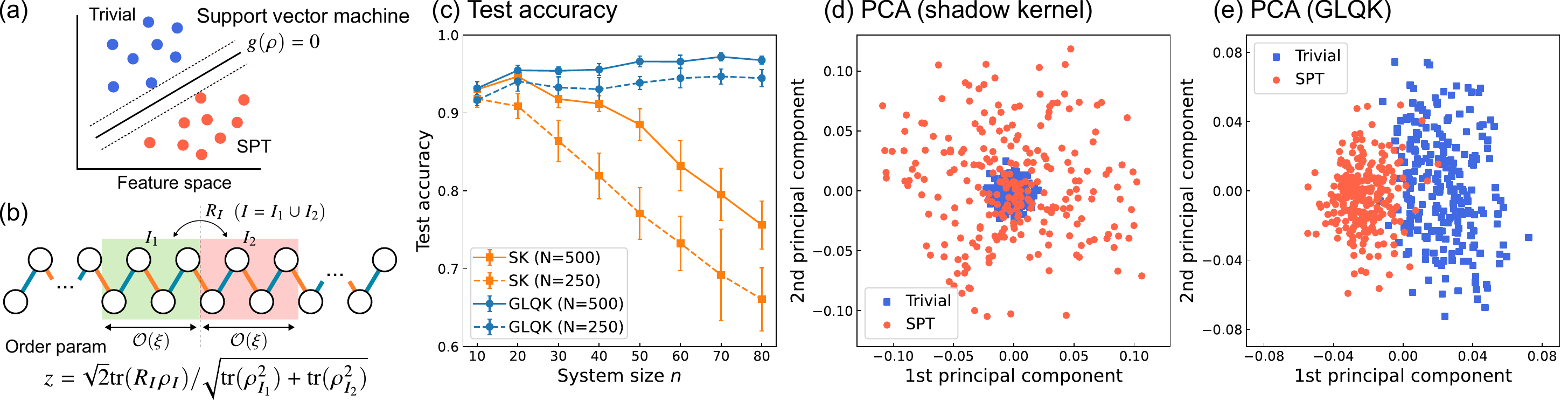}
    \caption{
    Numerical results of quantum phase recognition.
    (a) The support vector machine classifies data points using a hyperplane within the feature space.
    (b) Topological order parameter defined on a local subsystem $I$ of size $\mO(\xi)$ (see Sec.~\ref{sec: methods_numerics2} for details). 
    (c) Test accuracy for the shadow kernel (SK, orange squares) and the GLQK (blue circles) as the system size $n$ varies.
    Data points represent the mean value of test accuracy across 10 randomly sampled training and test datasets, and error bars represent the corresponding standard deviation.
    (d)--(e) Kernel PCA results obtained with the shadow kernel and the GLQK for 500 data points at $n=80$. 
    In (e), we set $\kappa=1$ and $\zeta=2$.
    The number of training data is denoted by $N$, and the shadow size is fixed at $T=500$.
    }
    \label{fig: QPR}
\end{figure*}

Figure~\ref{fig: QPR} (c) shows the test accuracy of the shadow kernel and the GLQK as the system size $n$ is varied.
The accuracy of the shadow kernel significantly decreases with increasing $n$, whereas the GLQK maintains high accuracy even with up to 80 qubits.
This underscores the high learning efficiency of the GLQK, enabling a substantial reduction in the number of training samples.
Furthermore, we perform kernel principal component analysis (PCA)~\cite{Mika1998-wh} to visualize the data geometry in the feature space [Figs.~\ref{fig: QPR} (d) and (e)].
In the shadow kernel, data points corresponding to the trivial and SPT phases overlap in the two-dimensional PCA space, indicating the difficulty in distinguishing between the two quantum phases.
Conversely, the PCA with the GLQK reveals a clear separation of data points, highlighting not only the easier classification than the shadow kernel but also the necessary and sufficient expressivity of the GLQK for this task.

It is noteworthy that the GLQK achieves high accuracy even on datasets containing states near the critical point, where the correlation length diverges. This indicates practical applicability beyond the regime covered by our theoretical guarantees. Although the true topological order parameter characterizing the SPT phase is associated with a length scale that diverges at the critical point, accurate classification does not necessarily require learning this exact nonlocal quantity. Instead, our results suggest that, within the considered data distribution, there exist sufficiently informative local features to distinguish the two phases. By optimizing the local subsystem size $\zeta$, the GLQK can capture such effective features even when $\zeta$ remains finite.

\section{Discussion} \label{sec: conclusions}
We have formulated the GLQK, a provably scalable ML framework that learns from quantum many-body experimental data by leveraging locality. 
Although this work has primarily focused on the kernel method, the underlying principle---that any polynomial $g(\rho)$ can be learned solely from local information---is broadly applicable to other ML approaches, including neural networks (NNs). 
Kernel methods are generally advantageous in data-limited regimes, as they incorporate domain knowledge into the kernel design and offer strong theoretical guarantees. 
This advantage is particularly relevant in the near term, when resources for quantum experiments remain limited. 
In contrast, NNs provide greater flexibility across a wider range of tasks by learning relevant features directly from data, though they typically require larger datasets. 
These two approaches are therefore complementary for the analysis of quantum many-body data and can be combined with the locality principle to further enhance their capabilities.

Demonstrating the quantum advantages of the GLQK in practical problems is a significant open problem.
The quantum advantages of our learning framework rely on preparing quantum data and sampling measurement outcomes, as the learning phase is performed on a classical computer.
Despite the controversy surrounding the boundary between classical and quantum computational complexities~\cite {Lee2023-tn}, quantum data preparation and measurement are believed to be classically hard in certain situations.
Even in our problems, although the ECP may allow the efficient tensor network representation of the quantum state~\cite{Brandao2013-rz}, tensor-network methods often struggle to solve actual problems in systems with more than one dimension due to the computational complexity of tensor contraction~\cite{Schuch2007-sv}.
This highlights the potential benefit of preparing such quantum states on quantum devices in the GLQK.
Exploring the quantum advantages of our method presents an intriguing opportunity to identify how ML affects the computational complexity of classical and quantum algorithms for finitely correlated quantum systems.

\section{Methods}\label{sec: methods}

\subsection{Local-Cover Number and Local-Factor Count} \label{sec: methods_LCN}
Consider an $m$-body, degree-$p$ polynomial $g(\rho)$, its $\Delta$-cluster approximation $g_{\rm CA}(\rho)=\sum_i c_i \prod_j\prod_k\tr(P_{ijk}\rho)$, and the set of local subsystems $\AGL(\zeta)$.
The local-cover number is a function of $g$, $\Delta$, and $\zeta$, defined as
\begin{align}
    \alpha_g = \LCN(g;\Delta,\zeta)\equiv\max_i (a_i). \label{eq: LCN}
\end{align}
Here, $a_i$ is the minimum number of subsystems in $\AGL(\zeta)$ required to encompass the support of the $i$th term of the $\Delta$-cluster approximation.
That is, there exists a partition $\mP_i\equiv\{P_{ijk}\}_{j,k}=\mP_{i,1}\sqcup \cdots \sqcup \mP_{i,a_i}$ and local subsystems $\{A_{i,1},\ldots,A_{i,a_i}\}$ ($A_{i,j}\in\AGL(\zeta)$) such that $\supp(P) \subseteq A_{i,j}$ for all $P\in\mP_{i,j}$, where $a_i$ is minimized among all possible partitions.
We have assumed that for any $P_{ijk}$, there exists $A\in\AGL(\zeta)$ such that $\supp(P_{ijk})\subseteq  A$ (this is necessarily satisfied if $\zeta\geq m\Delta$).
Then, we can rewrite $g_{\rm CA}(\rho)$ as
\begin{align}
    g_{\rm CA}(\rho) = \sum_i c_i \prod_{j=1}^{a_i} \prod_{P\in\mP_{i,j}} \tr(P\rho) = \sum_i c_i \prod_{j=1}^{a_i} \ell_{ij}(\rho),
\end{align}
where 
\begin{align}
    \ell_{ij}(\rho)=\prod_{P\in\mP_{i,j}} \tr(P\rho)
\end{align}
is a local quantity of $\rho$ on the subsystem $A_{i,j}$.
This means that each term of $g_{\rm CA}(\rho)$ can be represented as the product of at most $\alpha_g=\max_i(a_i)$ local quantities.
The local-cover number satisfies $\alpha_g \leq mp$ because the degree of $g_{\rm CA}$ (i.e., $|\mP_i|$) is bounded by $mp$.

The local-factor count is a function of $g$ and $\Delta$, defined as follows:
\begin{align}
    \beta_g = \LFC(g;\Delta)\equiv\max(p,\min_i(b_i)), \label{eq: LFC}
\end{align}
where $b_i=|\mP_i|$ is the degree of the $i$th term in $g_{\rm CA}$.
By definition, the local-factor count satisfies $p\leq \beta_g \leq mp$.

\subsection{Kernel method} \label{sec: methods_kernel}
The kernel method~\cite{Shalev-Shwartz2014-rj} addresses a nonlinear learning problem by mapping data $\bx$ to a high-dimensional feature vector $\phi(\bx)$ and solving a linear optimization problem in the feature space.
This method approximates a target function $g(\bx)$ with a linear function in the feature space $\braket{\bw,\phi(\bx)}$, where $\bw$ is a weight vector and $\braket{\cdot,\cdot}$ represents an inner product.
Instead of explicitly mapping to the high-dimensional space, the kernel method computes the inner product of feature vectors as the kernel function, $k(\bx,\bx')=\braket{\phi(\bx),\phi(\bx')}$, allowing us to utilize potentially infinite-dimensional feature space.
Formally, given training data $\bx_1,\ldots,\bx_N$, we consider the following model $h_{\balpha}(\bx)$ that approximates the target function $g(\bx)$:
\begin{align}
    h_{\balpha}(\bx) = \sum_i \alpha_i k(\bx_i, \bx), \label{eq: kernel approx}
\end{align}
where $\balpha=(\alpha_1,\ldots,\alpha_N)$ are trainable parameters.
Remarkably, the representer theorem ensures that this kernel model of Eq.~\eqref{eq: kernel approx} contains the optimal solution that minimizes the regularized empirical loss in the entire feature space via $\bw^\ast = \sum_i \alpha_i^\ast \phi(\bx_i)$, where $^\ast$ indicates that it is optimal.
Moreover, the optimal $\balpha^\ast$ can be obtained efficiently by solving an $N$-dimensional optimization problem on a classical computer.
Thus, if the target function $g(\bx)$ can be represented as a linear function in the feature space $g(\bx)=\braket{\bw,\phi(\bx)}$ with some weight vector $\bw$, the kernel method can learn $g(\bx)$ with high accuracy, given a sufficient amount of training data.
The statistical learning theory guarantees this fact, for example, in kernel ridge regression with $\ell_2$ regularization (see Sec.~\ref{sec_sm: ridge regression}).

\subsection{Truncated shadow kernel} \label{sec: methods_TSK}
The feature vector of the truncated shadow kernel is given by (see Sec.~\ref{sec_sm: shadow} for derivation)
\begin{align}
    \phi^{\rm TSK}_A(S_{T}(\rho))
    = \bigoplus_{d=0}^\infty \sqrt{ \frac{\tau^d}{d!}} \left( \bigoplus_{r=0}^{|A|} \sqrt{\left(\frac{\gamma}{|A|}\right)^{r}} \bigoplus_{\{i_1,\cdots,i_r\} \subseteq A} \vectorize\left(\sigma_{\{i_1,\cdots,i_r\}}\right) \right)^{\otimes d},
\end{align}
where $\vectorize(X)$ denotes the vectorization of the matrix $X$, and $\sigma_{\{i_1,\cdots,i_r\}}$ is the reduced density matrix on the subsystem $\{i_1,\cdots,i_r\}$ estimated from the classical shadow:
\begin{align}
    \sigma_{\{i_1,\cdots,i_r\}} = \frac{1}{T}\sum_{t=1}^T \sigma_{i_1}^{(t)} \otimes \cdots \otimes \sigma_{i_r}^{(t)}.
\end{align}
This feature vector includes arbitrarily large reduced density matrices on $A$ and their arbitrarily high-degree polynomials.
Compared to the shadow kernel, the truncated one excludes ``unphysical'' elements like $\sigma_{\{i_1, \cdots, i_r\}}$ with duplicated indices, thereby potentially improving learning efficiency.

\subsection{Intuitive mechanism for Theorems~\ref{thm: main theorem 1} and \ref{thm: main theorem 2}} \label{sec: methods_theorems}
In Theorem~\ref{thm: main theorem 1} for general quantum data, the improved sample complexity can be understood by counting the number of linearly independent polynomials included in the feature space.
Consider the independent basis of $m$-body, degree-$p$ polynomials represented as $\prod_{j=1}^p\braket{P_j}$, where $P_j$ is an $m$-weight Pauli string.
Then, the number of bases (i.e., the number of combinations in choosing $\{P_j\}$) is $\mO(n^{mp})$.
The shadow kernel includes all these bases within the feature space~\cite{Huang2022-ip}, resulting in the sample complexity of $\mO(n^{mp})$.
In contrast, the polynomial GLQK with $\kappa=\alpha_g$ includes only $\mO(n^{\alpha_g})$ bases because its feature space consists of degree-$\alpha_g$ polynomials in $\mO(n)$ local features.
This leads to the sample complexity of $\mO(n^{\alpha_g})$.

In Theorem~\ref{thm: main theorem 2} for translationally symmetric data, the constant sample complexity can be explained by the following argument.
Let $A^\ast \in \AGL(\zeta)$ be a representative local subsystem.
Then, using translation symmetry, the cluster approximation $g_{\rm CA}(\rho)=\sum_i c_i \prod_j\prod_k\tr(P_{ijk}\rho)$ can be written as a polynomial in the local reduced density matrix $\rho_{A^\ast}$, $g_{\rm CA}(\rho)=\sum_i c_i \prod_{j}\prod_{k} \tr_{A^\ast}(\tilde{P}_{ijk}\rho_{A^\ast})$, where $\tilde{P}_{ijk}$ is a Pauli string obtained by translating $P_{ijk}$ such that $\supp(\tilde{P}_{ijk})\subseteq A^\ast$.
Meanwhile, translation symmetry reduces the polynomial GLQK with $\kappa=1$ to the local quantum kernel on $A^\ast$, $k_{\rm GL}(\cdot,\cdot)=\sum_A k_A(\cdot,\cdot)/n \approx k_{A^\ast}(\cdot,\cdot)$, up to statistical errors originating from a finite shadow size $T$.
This consideration implies that the original learning problem on the entire system is approximately equivalent to that on the local subsystem $A^\ast$.
Since the size of $A^\ast$, $\zeta=m\xi \log(\|g\|_{1} mp/\epsilon)$, is independent of $n$, the GLQK requires only a constant number of training samples to achieve certain accuracy.

\subsection{Regression task involving random quantum dynamics} \label{sec: methods_numerics1}
We explore two types of one-dimensional local Hamiltonians: one translationally symmetric and the other not, defined as follows:
\begin{align}
H_{1} =\sum_{j=1}^n \sum_{\mu,\nu\in\{X,Y,Z\}}J^{\mu \nu}\sigma^\mu_j \sigma^{\nu}_{j+1}, \\
H_{2} =\sum_{j=1}^n \sum_{\mu,\nu\in\{X,Y,Z\}}J^{\mu \nu}_j \sigma^\mu_j \sigma^{\nu}_{j+1},
\end{align}
where $\sigma^\mu_j$ ($\mu=X,Y,Z$) is the single-qubit Pauli operator acting on the $j$th qubit, and $J^{\mu \nu}$ and $J^{\mu \nu}_j$ are the interaction strengths.
Since $J^{\mu \nu}$ does not depend on the qubit index, $H_1$ is translationally symmetric.

Given an initial product state $\ket{\phi_k}$, we consider the following quantum dynamics by $H_k$: $\ket{\psi_k} = e^{-iH_k t} \ket{\phi_k}$, where $k=1,2$.
Here, $\ket{\psi_k}$ is used as quantum data in this regression task.
We generate quantum data by randomly sampling the interaction strengths and the initial product states. 
Specifically, the interaction strengths $J^{\mu\nu}$ and $J^{\mu\nu}_j$ are drawn from the uniform distribution $[-1,1]$. 
For the initial product states, we define $\ket{\phi_1}$ as $\ket{u}\otimes \cdots \otimes \ket{u}$, where $\ket{u}$ is a single-qubit Haar random state, representing a translationally symmetric initial state. 
Alternatively, $\ket{\phi_2}$ is defined as $\ket{u_1}\otimes \cdots \otimes \ket{u_n}$, where $\ket{u_1},\ldots,\ket{u_n}$ are independent single-qubit Haar random states, representing a general initial state.
The evolution time is fixed at $t=0.5$.
The translation symmetry of $H_1$ and $\ket{\phi_1}$ ensures that $\ket{\psi_1}$ is also translationally symmetric.
For these quantum data, we consider three types of polynomials ($g_1,g_2$, and $g_3$) and two types of non-polynomials ($g_4$ and $g_5$).

To solve these learning problems, we invoke the kernel ridge regression using the shadow kernel and the polynomial GLQK with the truncated shadow kernel, defined in Eqs.~\eqref{eq: horder GLQK} and \eqref{eq: truncated shadow kernel}.
During training, some hyperparameters (the regularization parameter $\lambda$ for the shadow kernel, and $\zeta, \kappa$, and $\lambda$ for the GLQK) are optimized using grid search with cross-validation on $N$ training data.
The hyperparameters $\tau$ and $\gamma$ are fixed to 1 for both the shadow kernel and the GLQK.
We also use $M=500$ test data to evaluate the performance of the trained models.
For each quantum data, we perform $T=500$ measurement shots to obtain a classical shadow.

In this numerical experiment, we represent the quantum data $\ket{\psi_k}$ using a matrix product state (MPS) implemented with ITensor~\cite{Fishman2022-lh}, a tensor network simulation library. 
The one-dimensional nature of the Hamiltonian enables highly accurate calculations with the MPS. Additionally, we perform kernel ridge regression using scikit-learn~\cite{Pedregosa2011-mb}, an ML library. 
Further details regarding this experiment are provided in Sec.~\ref{sec_sm: numerics}.

\subsection{Quantum phase recognition} \label{sec: methods_numerics2}
We consider the following bond-alternating XXZ model:
\begin{align}
    H(J)
    = 
    \sum_{j=1}^{n/2} \left( X_{2j-1}X_{2j} + Y_{2j-1}Y_{2j} + \Delta_z Z_{2j-1}Z_{2j}\right) 
    + J \sum_{j=1}^{n/2-1} \left( X_{2j}X_{2j+1} + Y_{2j}Y_{2j+1} + \Delta_z Z_{2j}Z_{2j+1}\right),
\end{align}
where $J$ and $\Delta_z$ are the parameters of the Hamiltonian.
We fix $\Delta_z=0.5$ for simplicity.
For $\Delta_z=0.5$, the ground state of this Hamiltonian, $\ket{\phi(J)}$, exhibits a quantum phase transition from the trivial phase to the SPT phase at $J\approx 1$. 
Our task is to classify noisy ground-state data into these two phases.
This SPT phase is protected by the inversion symmetry that swaps the $j$th and $(n-j+1)$th qubits for $j=1,\ldots,n/2$, and characterized by a topological order parameter $z=\sqrt{2}\tr(R_I \rho_I)/[\tr(\rho_{I_1}^2)+\tr(\rho_{I_2}^2)]^{1/2}$, where $I_1=\{n/2-a+1, \ldots, n/2\}$ and $I_2=\{n/2+1, \ldots, n/2+a\}$ are local subsystems with width $a=\mO(\xi)$, $I=I_1\cup I_2$ is the union of $I_1$ and $I_2$, and $R_I$ is the inversion operator for $I_1$ and $I_2$ with respect to the reflection center~\cite{Elben2020-ip}.
The existence of this order parameter guarantees that the GLQK can learn the phase transition when the size of local subsystems is set to at least $\zeta = \mO(\xi)$.
Note that despite the divergence of the correlation length $\xi$ at the transition point, the GLQK based on local subsystems of finite size achieves high classification accuracy, as shown in the results.

Here, we assume that the ground state is disturbed by inversion-symmetric local noise as
\begin{align}
    &\ket{\tilde{\phi}(J)} = R \ket{\phi(J)}, \\
    &R= \left( U_1 \otimes \cdots \otimes U_{n/2} \right) \otimes \left( U_{n/2} \otimes \cdots \otimes U_{1} \right),
\end{align}
where $U_j$ ($j=1,\ldots,n/2$) is a single-qubit Haar random unitary. 
Note that $\ket{\tilde{\phi}(J)}$ is not translationally symmetric.
Since $R$ is local and inversion symmetric, it does not destroy the SPT phase that is protected by the inversion symmetry.
We adopt $\ket{\tilde{\phi}(J)}$ as quantum data in this task, which is randomly generated by drawing $J$ and $U_1, \ldots, U_{n/2}$ from the uniform distribution $[0.1, 1.9]$ and the single-qubit Haar random unitary ensemble, respectively.
The class label $y$ for training data $\ket{\tilde{\phi}(J)}$ is assigned as $y=0$ for the trivial phase (i.e., $J\lesssim 1$) and $y=1$ for the SPT phase (i.e., $J\gtrsim 1$).

We solve this classification task using the support vector machine~\cite{Cortes1995-ic} with the shadow kernel and the polynomial GLQK based on the truncated shadow kernel.
The calculation conditions are almost the same as those of the first numerical experiment: we learn from $N$ training data while optimizing some hyperparameters and use $M=500$ test data to evaluate the performance of the trained model. 
Each data is a classical shadow of size $T=500$.

In this numerical experiment, we represent the quantum data $\ket{\tilde{\phi}(J)}$ using an MPS implemented with ITensor~\cite{Fishman2022-lh}, a tensor network simulation library. 
The one-dimensional nature of the Hamiltonian enables highly accurate calculations with the MPS. Additionally, we perform the support vector machine using scikit-learn~\cite{Pedregosa2011-mb}, an ML library. 
Further details regarding this experiment are provided in Sec.~\ref{sec_sm: numerics}.

\section*{Code availability}
The codes used in numerical experiments are available at \url{https://github.com/FujitsuResearch/Geometrically-Local-Quantum-Kernel}.

\section*{Acknowledgments}
Fruitful discussions with Yuichi Kamata, Nasa Matsumoto, Riki Toshio, and Shintaro Sato are gratefully acknowledged.

\newpage
\enabletoc
\begin{center}
\textbf{\Large Supplementary Information}
\end{center}

\setcounter{section}{0}
\renewcommand{\thesection}{S\arabic{section}}
\renewcommand{\thesubsection}{\Alph{subsection}}
\renewcommand{\theequation}{S\arabic{equation}}
\setcounter{equation}{0}
\renewcommand{\thefigure}{S\arabic{figure}}
\setcounter{figure}{0}
\renewcommand{\thetable}{S\arabic{table}}
\setcounter{table}{0}

\renewcommand{\theHsection}{S\arabic{section}}
\renewcommand{\theHequation}{S\arabic{equation}}
\renewcommand{\theHfigure}{S\arabic{figure}}
\renewcommand{\theHtable}{S\arabic{table}}

\tableofcontents

\section{Related works}
\label{sec_sm: related_works}

\subsection{Leveraging locality in quantum learning theory}

The constraint of locality can significantly improve the efficiency of many quantum algorithms, including simulation, tomography, and circuit compilation, sometimes exponentially~\cite{Cerezo2021-tq, Cerezo2023-hz, Thanasilp2024-qc, Cramer2010-yw, Votto2025-ij, Guo2024-ly, Vermersch2024-bg, Rouze2024-gm, Mizuta2022-ua, Kanasugi2023-te, Huang2024-vx}. 
This constraint means that quantum information, such as entanglement, correlations, and interactions, does not stretch across the entire system arbitrarily, but is confined to small neighborhoods, thereby reducing the problem for the entire Hilbert space to one concerning a small subspace.

In particular, locality plays a central role in scalable quantum tomography and quantum property estimation.
For example, a powerful approach to quantum tomography is to represent quantum states using tensor networks under the assumption of area-law entanglement~\cite{Cramer2010-yw, Votto2025-ij, Guo2024-ly}.
This assumption effectively restricts the relevant Hilbert space to grow only polynomially with system size, thereby reducing the sample complexity of quantum tomography from exponential to polynomial scaling.
The reconstruction of tensor networks from measurement data often leverages random local measurements, such as classical shadows.
Ref.~\cite{Guo2024-ly} further improves computational efficiency by locally optimizing the tensor-network representation.
Regarding quantum property estimation, Ref.~\cite{Vermersch2024-bg} proposed an efficient method for estimating global properties, such as entropies, from local measurements under approximate factorization conditions characterized by the near-vanishing of conditional mutual information between subsystems.
Also, Ref.~\cite{Rouze2024-gm} introduced a locality-aware method for identifying the Hamiltonian parameters of an unknown state from a small number of samples, as well as for predicting physical properties across a phase by exploiting data from nearby regions of parameter space.
In our GLQK framework, quantum tomography serves as a subroutine for obtaining classical representations of quantum states. Consequently, advances in locality-exploiting tomography and property estimation methods may further enhance the efficiency of the GLQK.

Beyond quantum tomography and traditional quantum algorithms, locality has also emerged as a key principle in machine learning (ML) approaches to quantum many-body systems~\cite{Lewis2024-yh, Wanner2024-ei, Wu2024-az}.
For instance, previous studies have considered the learning task of predicting a local linear property $g(\rho(\bx))=\tr(O\rho(\bx))$, where $\rho(\bx)$ is the ground state of an unknown local Hamiltonian $H(\bx)$ with parameters $\bx$, and $O$ is an unknown local observable.
The goal of this problem is to predict the value $g(\rho(\bx))$ for an unseen parameter point $\bx$ by learning from a training dataset $\{\bx_i,g(\rho(\bx_i))\}_{i=1}^N$ within the same quantum phase as $\bx$.
While an initial ML approach without utilizing locality~\cite{Huang2022-ip} has demonstrated the potential to solve this task, it suffers from poor sample complexity, requiring a number of training samples that scales polynomially with the system size $n$, and exponentially with inverse precision $1/\epsilon$.
Recent studies~\cite{Lewis2024-yh, Wanner2024-ei} have made substantial progress in overcoming these limitations by explicitly leveraging the physical principle of locality.
They have succeeded in reducing the sample complexity with respect to $n$ and $\epsilon$ exponentially.

Compared to these previous results, our ML framework is applicable to more general situations.
First, our method can be applied to more general quantum data $\rho$, extending beyond ground and thermal states, and to more general target quantities $g(\rho)$, including nonlocal and nonlinear ones.
Our theory only assumes the exponential clustering property (ECP) and eliminates the condition that all data belongs to the same quantum phase.
Moreover, we do not need the Hamiltonian parameter $\bx$ as training data, only requiring measurement outcomes from quantum experiments for $\rho$.
The methodology that utilizes measurement outcomes as data has been explored in Ref.~\cite{Wu2024-az}; however, it has not provided theoretical guarantees for applicability and sample complexity.
Our results present a provably versatile and efficient approach, thereby accelerating the utilization of quantum many-body data obtained from experiments.

\subsection{Quantum kernel}

The quantum kernel method has been proposed to harness the quantum feature space that is classically intractable, offering a potential pathway to solve problems beyond the reach of classical computation~\cite{Havlicek2019-wd, Schuld2019-hm}.
While this method has been proven to exhibit quantum speedup for artificially designed datasets~\cite{Liu2021-oe}, achieving quantum advantages for practical problems is still challenging.
One bottleneck is the exponential concentration phenomenon~\cite{Thanasilp2024-qc}: the value of kernel functions concentrates around a fixed value exponentially with the number of qubits $n$, due to the exponentially large dimensionality of the Hilbert space.
This prevents the quantum kernel method from solving large-scale problems that cannot be addressed using classical approaches.
Several quantum kernels can overcome this difficulty in specific situations.
For instance, the projected quantum kernel~\cite{Huang2021-wb} can avoid this concentration by projecting the quantum state onto local reduced density matrices.
The shadow kernel~\cite{Huang2022-ip} also circumvents this problem by using the classical shadow instead of treating the quantum state directly.
In particular, the shadow kernel has been proven to learn quantum many-body phases with polynomial sample and computational time complexities, highlighting its potential for efficiently analyzing intricate quantum systems.
However, the polynomial complexities of the shadow kernel remain too demanding for near-term quantum devices, presenting a challenge to reduce resource requirements.

\subsection{Machine learning for quantum experimental data}

Applying classical ML to quantum measurement results, including the shadow kernel method, presents a promising approach for leveraging the advantages of quantum technologies. 
This methodology has found diverse applications across various learning tasks, including quantum phase recognition \cite{Huang2022-ip, Rem2019-zr, Cao2025-mx}, the prediction of quantum properties \cite{Wu2024-az, Che2024-ue, Du2023-hz, Tang2024-ei, Molteni2024-ja}, and the generation of quantum many-body states \cite{Wang2022-sq, Yao2024-jo, Tang2025-kp}. 
This approach often assumes the ``measure-first'' protocol, where quantum states are initially measured independently of a specific task, and the resulting measurement outcomes are subsequently used for the task. 
This contrasts with the ``fully-quantum'' protocol, where measurements are adapted during the training process. 
Recent advancements have demonstrated both the limitations and potential of these protocols \cite{Gyurik2023-qg}. 
While theoretical findings indicate that the fully-quantum protocol can efficiently resolve certain learning tasks that demand exponential resources from the measure-first protocol, the practical applicability of this distinction to real-world problems remains an open question. 
Identifying the precise boundaries of quantum advantage within these approaches constitutes a compelling research inquiry.

\section{Exponential clustering and cluster approximation}
\label{sec_sm: clustering}

This section provides the detailed definitions of the ECP and cluster approximation, and proves Lemma 1 in the main text, which quantifies the accuracy of cluster approximation.

\subsection{Exponential clustering property}

Here, we consider an $n$-qubit system on the $D$-dimensional hypercubic lattice $G\subset \mathbb{Z}^D$ with the periodic boundary condition, where each qubit is located at a lattice point (i.e., $|G|=n$).
The distance between two lattice points, $\bm{a}=(a_1,\cdots,a_D)\in G$ and $\bm{b}=(b_1,\cdots,b_D) \in G$, is defined as $\dist(\bm{a},\bm{b})=\sum_{i=1}^D |a_i-b_i|$.
The following arguments can be extended to general lattices, where $\dist(\bm{a},\bm{b})$ is defined as the length of the shortest path connecting $\bm{a}$ and $\bm{b}$.
For notational simplicity, we may represent $G$ by $[n]=\{1,2,\ldots,n\}$, where each element corresponds to a lattice point, or a qubit.
We may also denote the power set of $G$ (the set of all subsets of $G$) by $2^G$.
For a Pauli string $P\in\{I,X,Y,Z\}^{\otimes n}$ on $G$, let $\supp(P)\subseteq G$ be the support of $P$, i.e., the set of qubits on which $P$ acts nontrivially (e.g., $\supp(X_1Z_2Y_4)=\{1,2,4\}$).
Also, the Pauli weight of $P$ is defined as the number of $X, Y$, and $Z$ operators in $P$ (e.g., the weight of $X_1Z_2Y_4$ is 3).

Let $\rho$ be an $n$-qubit quantum state on $G$.
We say that $\rho$ satisfies the ECP if the following inequality holds for any observables $O_A$ and $O_B$, each acting on subsystems $A\subseteq G$ and $B\subseteq G$, respectively:
\begin{align}
    \left| \braket{O_A O_B} - \braket{O_A}\braket{O_B}\right| 
    \leq \| O_A \|_{S} \| O_B \|_{S} \, e^{-\dist(A,B)/\xi},
\end{align}
where $\dist(A,B)=\min_{\bm{a}\in A, \bm{b}\in B} \dist(\bm{a},\bm{b})$ is the shortest distance between $A$ and $B$ on the lattice, $\xi$ is the correlation length, $\braket{X}=\tr(X\rho)$ is the expectation value, and $\|X\|_{S}$ denotes the spectral norm defined as the maximum eigenvalue of $(X^\dag X)^{1/2}$.
This clustering property indicates that quantum correlations decay exponentially in space, justifying the approximation of $\braket{O_A O_B} \approx \braket{O_A}\braket{O_B}$ for any observables $O_A$ and $O_B$ with $\dist(A,B) \gg \xi$.

\subsection{Cluster approximation} \label{sec_sm: cluster approx}

The focus of this work is learning an unknown polynomial $g(\rho)$.
We characterize the polynomial as follows:
\begin{dfn}[$m$-body, degree-$p$ polynomial, Definition~\ref{def: polynomial in quantum state} in the main text]
Consider the following function $g(\rho)$ of a quantum state $\rho$:
\begin{align}
g(\rho) = \sum_i c_i \left( \prod_{j=1}^p \tr \left[ P_{ij} \rho \right] \right),
\end{align}
where $P_{ij}\in\{I,X,Y,Z\}^{\otimes n}$ is an $n$-qubit Pauli string, and $c_i$ is an expansion coefficient.
Then, if the Pauli weights of all $P_{ij}$'s are less than or equal to $m$, we say that $g(\rho)$ is an $m$-body, degree-$p$ polynomial in $\rho$.
Also, we define the $\ell_1$- and $\ell_2$-norms of Pauli coefficients as $\|g\|_{1}=\sum_i |c_i|$ and $\|g\|_{2}= (\sum_i |c_i|^2)^{1/2}$, respectively.
\end{dfn}

Here, we introduce the cluster approximation of the polynomial $g(\rho)$.
This is defined as follows:
\begin{dfn}[Cluster approximation]
Let $g(\rho) = \sum_i c_i \prod_{j=1}^p \tr \left[ P_{ij} \rho \right]$ be an $m$-body, degree-$p$ polynomial.
Given a distance $\Delta$, we decompose $P_{ij}$ in the following manner.
Define a graph $Q_{ij}$ consisting of nodes and edges, where each node corresponds to an element in $\supp(P_{ij})$, and two nodes $\bm{a},\bm{b}\in \supp(P_{ij})$ are connected by an edge if and only if $\dist(\bm{a},\bm{b})\leq \Delta$.
Then, let $Q_{ij}$ be separated into $d_{ij}$ connected subgraphs, called clusters, $Q_{ij1},Q_{ij2},\cdots,Q_{ijd_{ij}}$ ($1\leq d_{ij} \leq m$).
That is, there exists a path connecting any pair of nodes within each cluster, and there are no edges connecting different clusters.
Based on this graph, we decompose the Pauli string $P_{ij}$ as $P_{ij}=P_{ij1}\otimes P_{ij2} \otimes\cdots\otimes P_{ijd_{ij}}$, where $P_{ijk}$ is the partial Pauli string of $P_{ij}$ acting on the cluster $Q_{ijk}$ ($k=1,\ldots,d_{ij}$).
This decomposition defines the $\Delta$-cluster approximation of $g(\rho)$ as
\begin{align}
g_{\rm CA}(\rho) 
= \sum_i c_i \prod_{j=1}^p \prod_{k=1}^{d_{ij}} \tr \left[ P_{ijk} \rho \right].
\end{align}
\end{dfn}    

Intuitively, this approximation decomposes $\supp(P_{ij})$ by grouping spatially close qubits together and separating distant qubits into different clusters.
If $g(\rho)$ is an $m$-body, degree-$p$ polynomial, its cluster approximation $g_{\rm CA}(\rho)$ is at most $m$-body and degree-$mp$.

To tightly evaluate the $\ell_1$-norm of $g_{\rm CA}(\rho)$, we combine its duplicated terms as follows.
Let $\mP_i^0=\{P_{ijk}\}_{jk}$.
We partition the domain of the index $i$, $\{1,2,3,\ldots\}$, into $N_1 \sqcup N_2 \sqcup \cdots$ such that $\mP_i^0=\mP_j^0$ if $i$ and $j$ are in the same $N_k$ and $\mP_i^0\neq\mP_j^0$ if $i$ and $j$ are in different $N_k$'s.
Then, we combine duplicated terms in $g_{\rm CA}(\rho)$ as $\sum_{i\in N_k} c_i \prod_{P\in\mP_{i}^0}\tr[P\rho] = \hc_k \prod_{P\in\mP_{k}}\tr[P\rho]$, where we have defined $\hc_k = \sum_{i\in N_k} c_i$ and $\mP_k=\mP^0_i$ for some $i\in N_k$.
As a result, we obtain
\begin{align}
g_{\rm CA}(\rho) 
= \sum_i \hc_i \prod_{P\in\mP_i} \tr \left[ P \rho \right], \label{eq_sm: combined cluster approximation}
\end{align}
where we have rewritten the index $k$ as $i$.
This expression for $g_{\rm CA}(\rho)$ will be used in what follows.
The $\ell_1$-norm of $g_{\rm CA}(\rho)$ is defined as $\|g_{\rm CA}\|_1=\sum_i |\hc_i|$, which is smaller than that of the original polynomial:
\begin{align}
\|g_{\rm CA}\|_1 = \sum_i |\hc_i| \leq \sum_i |c_i| = \|g\|_1 \label{eq_sm: cluster norm}
\end{align}
because of the triangle inequality $|x+y|\leq|x|+|y|$.

We show the following inequalities for later use: 
\begin{align}
    \sum_{P\in\mP_i}|\supp(P)|\leq mp \quad \text{and} \quad b_i\equiv |\mP_i|\leq mp.\label{eq_sm: convenient ineq}
\end{align}
These inequalities hold because $\sum_{P\in\mP_i}|\supp(P)|=\sum_{j=1}^p \sum_{k=1}^{d_{ij}} |\supp(P_{ijk})|$, $\sum_{k=1}^{d_{ij}}|\supp(P_{ijk})| = |\supp(P_{ij})| \leq m$, and $|\mP_i| =\sum_{P\in\mP_i} 1 \leq\sum_{P\in\mP_i}|\supp(P)|\leq mp$.

For any $m$-body polynomial, each cluster $\supp(P_{ijk})$ in the cluster approximation is encompassed by a local subsystem of size $\zeta\geq m\Delta$, since the number of qubits included in each cluster is at most $m$, and the distance between neighboring qubits within the cluster is less than $\Delta$.
Considering this, we define a set of local subsystems $\AGL(\zeta)$ as follows:
\begin{align}
    \AGL(\zeta) = \{ A_{\bm{a}}(\zeta) \,|\, \bm{a}\in G\}, \label{eq_sm: set of subsystem}
\end{align}
where $A_{\bm{a}}(\zeta) = \{ \bm{b}\in G \,|\,  a_j\leq b_j < a_j+\zeta, \,\forall j \}$ is a local subsystem of width $\zeta$ whose corner is located at $\bm{a}\in G$. 
By definition, $|\AGL(\zeta)|=n$ and $|A_{\bm{a}}(\zeta)|=\zeta^D$ hold.

\subsection{Accuracy in cluster approximation}

For quantum states exhibiting the ECP, $g_{\rm CA}(\rho)$ well approximates the original polynomial $g(\rho)$ if $\Delta$ is sufficiently large.
To show this, we first prove two lemmas for preliminaries:
\begin{lem} \label{lem_sm: simultaneous inequalities}
Let $X_1,X_2,\cdots \in [-1,1]$ and $y_1,y_2,\cdots \in [-1,1]$ be real numbers satisfying $X_1=y_1$.
If $|X_{i+1} -  X_{i}y_{i+1}| \leq \epsilon$ for any $i\in\{1,2,\cdots\}$, then
\begin{align}
    |X_k - Y_k|\leq (k-1)\epsilon \label{eq_sm: lemma2_2}
\end{align}
holds for any $k\in\{1,2,\cdots\}$, where we have defined $Y_k = \prod_{i=1}^k y_{i}$.
\end{lem}

\begin{proof}
We prove this lemma by mathematical induction with respect to $k$.
\begin{enumerate}
    \item[(i)] For $k=1$, the lemma holds from $|X_1-Y_1|=|X_1-y_1|=0$, where we have used $X_1=y_1$.
    \item[(ii)] Assume $|X_k-Y_k| \leq (k-1)\epsilon$. 
    Then, we have
    \begin{align}
        |X_{k+1}-Y_{k+1}|
        &=|X_{k+1} - Y_{k}y_{k+1}| \\
        &=| (X_{k+1} - X_k y_{k+1}) + (X_k y_{k+1} - Y_{k}y_{k+1})| \\
        &\leq |X_{k+1} - X_k y_{k+1}| + |X_k-Y_k| \cdot |y_{k+1}| \\
        &\leq \epsilon + (k-1)\epsilon \\
        &=k\epsilon,
    \end{align}
    where we have used $|X_{k+1}-X_k y_{k+1}|\leq \epsilon$, $|X_k-Y_k| \leq (k-1)\epsilon$, and $|y_{k+1}|\leq 1$ in the third line.
\end{enumerate}
These calculations prove the lemma for any $k$ by mathematical induction.
\end{proof}

\begin{lem} \label{lem_sm: product formula}
Let $z_1,z_2,\cdots \in [-1,1]$ and $w_1,w_2,\cdots \in [-1,1]$ be real numbers.
If $|z_i-w_i| \leq \epsilon$ for any $i\in\{1,2,\ldots\}$, then \begin{align}
    \left| Z_k - W_k \right| \leq k\epsilon
\end{align}
holds for any $k\in\{1,2,\ldots\}$, where we have defined $Z_k = \prod_{i=1}^k z_i$ and $W_k = \prod_{i=1}^k w_i$.
\end{lem}

\begin{proof}
    We prove this lemma by mathematical induction with respect to $k$.
    \begin{enumerate}
        \item[(i)] For $k=1$, the lemma holds by the assumption of $|z_1-w_1|\leq \epsilon$.
        \item[(ii)] Assume $|Z_k- W_k| \leq k\epsilon$.
        Then, we have
        \begin{align}
            |Z_{k+1}-W_{k+1}| 
            &= |Z_{k}z_{k+1}-W_{k}w_{k+1}|  \\
            &= |(Z_{k}z_{k+1} - W_{k}z_{k+1}) + (W_{k}z_{k+1} - W_{k}w_{k+1})| \\
            &\leq |Z_k-W_k|\cdot|z_{k+1}| + |W_k|\cdot|z_{k+1}-w_{k+1}| \\
            &\leq k\epsilon + \epsilon  \\
            &= (k+1)\epsilon,
        \end{align}
        where we have used $|Z_k-W_k|\leq k\epsilon$, $|z_{k+1}|\leq 1$, $|W_{k}|\leq 1$, and $|z_{k+1}-w_{k+1}|\leq \epsilon$ in the third line.
    \end{enumerate}
    These prove the lemma for any $k$ by mathematical induction.
\end{proof}

Based on these, we prove the following lemma to quantify the accuracy of cluster approximation:
\begin{lem}[Lemma~\ref{lem: approx} in the main text] \label{lem_sm: approx}
    Let $g(\rho)$ be an $m$-body, degree-$p$ polynomial.
    For any $\epsilon\in(0,\infty)$ and $\xi\in(0,\infty)$, the $\Delta$-cluster approximation $g_{\rm CA}(\rho)$ with $\Delta=\xi \log(\|g\|_{1} mp/\epsilon)$ satisfies
    \begin{align}
        &\left| g(\rho) - g_{\rm CA}(\rho) \right| \leq \epsilon, 
    \end{align}
    for any $\rho$ satisfying the ECP with a correlation length less than or equal to $\xi$.
\end{lem}    

\begin{proof}

Let $g(\rho)=\sum_i c_i \prod_{j=1}^p \tr[P_{ij}\rho]$ and $g_{\rm CA}(\rho) = \sum_i c_i \prod_{j=1}^p \prod_{k=1}^{d_{ij}} \tr[P_{ijk}\rho]$.
We prove this lemma based on the ECP and Lemmas~\ref{lem_sm: simultaneous inequalities} and \ref{lem_sm: product formula}.
To this end, we define
\begin{align}
    &&&X_k^{(ij)} = \tr[P_{ij1}\cdots P_{ijk}\rho], &&&&\\
    &&&y_k^{(ij)} = \tr[P_{ijk}\rho], &&Y_k^{(ij)} = \prod_{\ell=1}^k y_\ell^{(ij)}, &&\\
    &&&z_j^{(i)} = X_{d_{ij}}^{(ij)} = \tr[P_{ij}\rho], &&Z_j^{(i)}=\prod_{\ell=1}^j z_\ell^{(i)}, &&\\
    &&&w_j^{(i)} = Y_{d_{ij}}^{(ij)} = \prod_{k=1}^{d_{ij}} \tr[P_{ijk}\rho], &&W_j^{(i)} = \prod_{\ell=1}^j w_\ell^{(i)}, &&
\end{align}
where we have used $P_{ij}=P_{ij1}\cdots P_{ijd_{ij}}$ in the equality $X_{d_{ij}}^{(ij)} = \tr[P_{ij}\rho]$.
As $P_{ij}$ and $P_{ijk}$ are Pauli strings, the absolute values of these quantities are bounded by one, and $X_1^{(ij)}=y_1^{(ij)}$ holds by definition.
These are necessary conditions for Lemmas~\ref{lem_sm: simultaneous inequalities} and \ref{lem_sm: product formula} to be applied.
The polynomials are represented as $g(\rho)=\sum_i c_i Z_p^{(i)}$ and $g_{\rm CA}(\rho)=\sum_i c_i W_p^{(i)}$.

In the cluster approximation, since the distance between clusters is more than $\Delta=\xi \log(\|g\|_{1} mp/\epsilon)$, the ECP leads to
\begin{align}
    \left|X_{k+1}^{(ij)} - X_k^{(ij)}y_{k+1}^{(ij)}\right|
    = \left| \tr\left[P_{ij1} \cdots P_{ijk+1}\rho\right] - \tr\left[P_{ij1} \cdots P_{ijk}\rho\right]\tr\left[P_{ijk+1}\rho\right] \right| \leq e^{-\Delta/\xi} = \frac{\epsilon}{\|g\|_{1} mp}, \label{eq_sm: lemma3_1}
\end{align}
where we have used $\|P_{ij1} \cdots P_{ijk}\|_S=\|P_{ijk+1}\|_S=1$.
By Lemma~\ref{lem_sm: simultaneous inequalities} and Eq.~\eqref{eq_sm: lemma3_1}, we have
\begin{align}
    \left|z_j^{(i)} - w_j^{(i)}\right| = \left|X_{d_{ij}}^{(ij)} - Y_{d_{ij}}^{(ij)}\right| \leq \frac{(d_{ij}-1)\epsilon}{\|g\|_{1} mp} \leq \frac{\epsilon}{\|g\|_{1} p}, \label{eq_sm: lemma3_2}
\end{align}
where $d_{ij}\leq m$ have been used.
Then, Lemma~\ref{lem_sm: product formula} and Eq.~\eqref{eq_sm: lemma3_2} show that 
\begin{align}
    \left|Z_p^{(i)} - W_p^{(i)}\right| \leq \frac{\epsilon}{\|g\|_{1}} = \frac{\epsilon}{\sum_i |c_i|}. \label{eq_sm: lemma3_3}
\end{align}
Therefore, we obtain
\begin{align}
\left| g(\rho)-g_{\rm CA}(\rho) \right| 
&= \left| \sum_i c_i Z_p^{(i)} - \sum_i c_i W_p^{(i)} \right|  \\
&\leq  \sum_i \left|c_i (Z_p^{(i)} - W_p^{(i)})\right|  \\
&= \sum_i |c_i| \cdot |Z_p^{(i)} - W_p^{(i)}|  \\
&\leq \sum_i |c_i| \cdot \frac{\epsilon}{\sum_i |c_i|}  \\
&= \epsilon.
\end{align}
\end{proof}

\subsection{Local-cover number and local-factor count}

Here, we introduce two quantities of the polynomial $g(\rho)$, the local-cover number and local-factor count, which are crucial for evaluating the learning cost scaling of the GLQK and shadow kernel.

We first define the local-cover number $\alpha_g=\LCN(g;\Delta,\zeta)$.
Let us consider the $\Delta$-cluster approximation of the $m$-body, degree-$p$ polynomial $g(\rho)$: $g_{\rm CA}(\rho)=\sum_i \hc_i \prod_{P\in\mP_i} \tr \left[ P \rho \right]$.
Given a set of local subsystems $\AGL(\zeta)$, assume that the support of any Pauli string in $\mP_i$ is encompassed by some subsystem in $\AGL(\zeta)$ (i.e., $\forall P\in\mP_i$, $\exists A\in\AGL(\zeta)$ s.t. $\supp(P)\subseteq A$).
This assumption is necessarily satisfied if $\zeta\geq m\Delta$.
Then, we partition $\mP_i$ as
\begin{align}
\mP_i=\mP_{i,1}\sqcup \mP_{i,2} \sqcup \cdots \sqcup \mP_{i,a_{i}},
\end{align}
such that for $\forall \mP_{i,j}$, there exists $\exists A_{i,j}\in\AGL(\zeta)$ satisfying
\begin{align}
\supp(P) \subseteq A_{i,j} \text{ for all $P\in\mP_{i,j}$}.
\end{align}
Here, $a_i$ represents the number of partitions, and this value is assumed to be minimized among all possible partitions.
Using this partition, we can rewrite the polynomial as
\begin{align}
g_{\rm CA}(\rho) 
= \sum_i \hc_i \prod_{j=1}^{a_i} \prod_{P\in \mP_{i,j}} \tr \left[ P \rho \right]=\sum_i \hc_i \prod_{j=1}^{a_i} \ell_{ij}(\rho), \label{eq_sm: gcluster LCN}
\end{align}
where $\ell_{ij}(\rho)=\prod_{P\in \mP_{i,j}} \tr \left[ P \rho \right]$ is a local quantity on the subsystem $A_{i,j}$.
Here, we define the local-cover number of $g(\rho)$ as
\begin{align}
    \alpha_g = \LCN(g;\Delta,\zeta) \equiv \max_i(a_i),
\end{align}
which is a function of $g$, $\Delta$, and $\zeta$.
This quantity describes the locality of $g_{\rm CA}(\rho)$ relative to the scale $\zeta$, satisfying $\alpha_g \leq mp$ because the degree of $g_{\rm CA}$ (i.e., $|\mP_i|$) is bounded by $mp$.
For instance, the expectation values of local observables (e.g., local Hamiltonians, magnetization) and the purity/entanglement entropy of a local subsystem both correspond to $\alpha_g = 1$ if $\zeta$ is sufficiently large to cover each local term.
Meanwhile, $t$-point correlation functions satisfy $\alpha_g = t$ in general.

The local-factor count, roughly corresponding to the degree of $g_{\rm CA}$, is defined as 
\begin{align}
    \beta_g = \LFC(g;\Delta) \equiv \max(p, \min_i (b_i)),
\end{align}
where $b_i=|\mP_i|$ is the degree of the $i$th term in $g_{\rm CA}$.
By definition, the local-factor count satisfies $p\leq \beta_g \leq mp$.
The quantity $\beta_g$ takes a large value ($\sim mp$) if $\supp(P_{ij})$ in $g(\rho)$ is dispersed across spatially distant positions compared to $\Delta$, while it takes a small value ($\sim p$) if the support is concentrated locally.
For instance, the expectation values of local observables satisfy $\beta_g=p=1$, while the purity on a local subsystem corresponds to $\beta_g=p=2$, if $\Delta$ is sufficiently large.

\section{Classical shadows for machine learning} \label{sec_sm: shadow}

This section elaborates on classical shadows and several quantum kernels based on them.
Furthermore, we derive the sample complexity required for estimating the value of $g(\rho)$ from a classical shadow of $\rho$.

\subsection{Classical shadows}

In classical shadow tomography based on random Pauli measurements~\cite{Aaronson2020-uc, Huang2020-ti}, we prepare a quantum state $\rho$ and measure each qubit of $\rho$ on a random Pauli basis, repeating this procedure $T$ times.
Let $W_i^{(t)}=X_i, Y_i, Z_i$ and $o_i^{(t)}=\pm1$ be the measurement basis and the measurement outcome at the $i$th qubit in the $t$th round.
We call $S_T(\rho)=\{(W_i^{(t)},o_i^{(t)})\}_{i=1,t=1}^{n,T}$ a classical shadow of $\rho$.
The original quantum state $\rho$ can be reconstructed from the classical shadow as
\begin{align}
    &\rho \sim \sigma = \frac{1}{T} \sum_{t=1}^T \sigma_1^{(t)}\otimes \cdots \otimes \sigma_n^{(t)},
\end{align}
where $\sigma_i^{(t)}$ is a $2\times2$ matrix acting on the $i$th qubit, defined as
\begin{align}
    &\sigma_i^{(t)} = \frac{1}{2} \left( 3o_i^{(t)} W_i^{(t)} + I \right).
\end{align}
This constructed quantum state $\sigma$ is an unbiased estimator of $\rho$ such that $\mathbb{E}[\sigma]=\rho$.
In the limit of $T\to\infty$, it approaches $\rho$: $\lim_{T\to\infty}\sigma = \rho$. 
Furthermore, the reduced density matrix on a subsystem $\{i_1,\cdots,i_r\}\subseteq [n]$ is estimated from a classical shadow as
\begin{align}
    \rho_{\{i_1,\cdots,i_r\}} \sim \sigma_{\{i_1,\cdots,i_r\}} = \frac{1}{T} \sum_{t=1}^{T} \bigotimes_{\ell=1}^r \sigma_{i_\ell}^{(t)}.
\end{align}
It is known that classical shadows based on random Pauli measurements can estimate the expectation value of an $m$-body observable with additive error 
$\epsilon$ using $T=\mO(4^m/\epsilon^2)$ samples, indicating high efficiency in estimating few-body observables.
Hereafter, let $\mD_\rho$ be the probability distribution of classical shadows for $\rho$.

\subsection{Quantum kernels based on classical shadows} 

\subsubsection{Shadow kernel}

The shadow kernel, which has been originally proposed in Ref.~\cite{Huang2022-ip}, is defined for two classical shadows $S_T(\rho)$ and $S_T(\tilde{\rho})$ as follows:
\begin{align}
    k_{\rm SK}(S_T(\rho),S_T(\tilde{\rho})) = \exp\left[ \frac{\tau}{T^2}\sum_{t,t'=1}^T \exp\left( \frac{\gamma}{n}\sum_{i=1}^n \tr\left(\sigma_i^{(t)} \tilde{\sigma}_i^{(t')} \right) \right) \right], \label{eq_sm: def of shadow kernel}
\end{align}
where $\tau$ and $\gamma$ are positive real hyperparameters.
The feature vector is given by
\begin{align}
    \phi_{\rm SK}(S_T(\rho)) 
    = \bigoplus_{d=0}^\infty \sqrt{\frac{\tau^d}{d!}} \left( \bigoplus_{r=0}^{\infty} \sqrt{\frac{1}{r!}\left(\frac{\gamma}{n}\right)^r} \bigoplus_{i_1=1}^n \cdots \bigoplus_{i_r=1}^n \vectorize\left(\sigma_{\{i_1,\cdots,i_r\}}\right) \right)^{\otimes d},
\end{align}
where we have defined the vectorized reduced density matrix as $\left(\vectorize(\sigma_{\{i_1,\cdots,i_r\}})\right)_j=\tr(P_j \sigma)/\sqrt{2^r}$ with the $j$th Pauli string $P_j \in \{I,X,Y,Z\}^{\otimes r}$ on the subsystem $\{i_1,\cdots,i_r\}$ ($j=1,\cdots,4^r$).
This indicates that the feature vector of the shadow kernel includes arbitrarily large reduced density matrices and their arbitrarily high-degree polynomials.
Note that the indices $i_1,\ldots,i_r$ can be duplicated.
See Ref.~\cite{Huang2022-ip} for the derivation of this feature vector.
This kernel is bounded as $|k_{\rm SK}(\cdot,\cdot)| \leq \exp(\tau\exp(5\gamma))$ because $\tr(\sigma_i^{(t)} \tilde{\sigma}_i^{(t')})=5,1/2,-4$.

We organize the feature vector components for later use.
Consider a set of Pauli strings $\mP=\{P_{1},P_{2},\cdots,P_{b}\}$, where $b$ is the number of Pauli strings contained in $\mP$.
Then, the feature vector of the shadow kernel has the following components:
\begin{align}
    \sqrt{\frac{\tau^{b}}{b!}} \left( \prod_{P\in \mc{P}} \sqrt{ \frac{1}{|\supp(P)|!}\left(\frac{\gamma}{2n}\right)^{|\supp(P)|}} \tr\left[P\sigma\right] \right), \label{eq_sm: SK component}
\end{align}
where $|\supp(P)|$ is the Pauli weight of $P$.

\subsubsection{Truncated shadow kernel}

We define a new quantum kernel called the truncated shadow kernel for classical shadows $S_T(\rho)$ and $S_T(\tilde{\rho})$:
\begin{align}
    k^{\rm TSK}(S_T(\rho),S_T(\tilde{\rho})) = \exp\left[ \frac{\tau}{T^2}\sum_{t,t'=1}^T \prod_{i=1}^n \left(1+\frac{\gamma}{n} \tr\left(\sigma_i^{(t)} \tilde{\sigma}_i^{(t')}\right) \right) \right],
\end{align}
where $\tau$ and $\gamma$ are positive real hyperparameters.
The feature vector of this kernel is given by
\begin{align}
    \phi^{\rm TSK}(S_T(\rho)) 
    = \bigoplus_{d=0}^\infty \sqrt{\frac{\tau^d}{d!}} \left( \bigoplus_{r=0}^{n} \sqrt{\left(\frac{\gamma}{n}\right)^{r}} \bigoplus_{\{i_1,\cdots,i_r\} \subseteq [n]} \vectorize\left(\sigma_{\{i_1,\cdots,i_r\}}\right) \right)^{\otimes d}. \label{eq_sm: TSK feature vector}
\end{align}
Indeed, this feature vector reproduces the truncated shadow kernel as
\begin{align}
    &\braket{\phi^{\rm TSK}(S_T(\rho)), \phi^{\rm TSK}(S_T(\tilde{\rho}))}  \\
    &= \sum_{d=0}^\infty \frac{\tau^d}{d!} \left( \sum_{r=0}^n \left(\frac{\gamma}{n}\right)^r \sum_{\{i_1,\cdots,i_r\}\subseteq [n]} \tr(\sigma_{\{i_1,\cdots,i_r\}} \tilde{\sigma}_{\{i_1,\cdots,i_r\}}) \right)^d\\
    &= \sum_{d=0}^\infty \frac{\tau^d}{d!} \left( \sum_{r=0}^n \left(\frac{\gamma}{n}\right)^r \sum_{\{i_1,\cdots,i_r\}\subseteq [n]} \frac{1}{T^2}\sum_{t,t'=1}^T \tr\left( (\sigma_{i_1}^{(t)}\otimes \cdots \otimes\sigma_{i_r}^{(t)})  (\tilde{\sigma}_{i_1}^{(t')}\otimes \cdots \otimes \tilde{\sigma}_{i_r}^{(t')})\right) \right)^d\\
    &= \sum_{d=0}^\infty \frac{1}{d!} \left( \frac{\tau}{T^2}\sum_{t,t'=1}^T \sum_{r=0}^n \sum_{\{i_1,\cdots,i_r\}\subseteq [n]} \left(\frac{\gamma}{n}\right)^r \tr\left( \sigma_{i_1}^{(t)} \tilde{\sigma}_{i_1}^{(t')}\right) \cdots \tr\left( \sigma_{i_r}^{(t)} \tilde{\sigma}_{i_r}^{(t')} \right) \right)^d\\
    &= \sum_{d=0}^\infty \frac{1}{d!} \left(  \frac{\tau}{T^2} \sum_{t,t'=1}^T  \left(1 + \frac{\gamma}{n}\tr(\sigma_{1}^{(t)} \tilde{\sigma}_{1}^{(t')}) \right) \cdots \left(1 + \frac{\gamma}{n}\tr(\sigma_{n}^{(t)} \tilde{\sigma}_{n}^{(t')}) \right) \right)^d \\
    &= \exp \left[ \frac{\tau}{T^2} \sum_{t,t'=1}^T  \left(1 + \frac{\gamma}{n}\tr(\sigma_{1}^{(t)} \tilde{\sigma}_{1}^{(t')}) \right) \cdots \left(1 + \frac{\gamma}{n}\tr(\sigma_{n}^{(t)} \tilde{\sigma}_{n}^{(t')})\right) \right] \\
    &=k^{\rm TSK}(S_T(\rho),S_T(\tilde{\rho})),
\end{align}
where we have used $\braket{\vectorize(\sigma_{\{i_1,\cdots,i_r\}}),  \vectorize(\tilde{\sigma}_{\{i_1,\cdots,i_r\}})}=\tr(\sigma_{\{i_1,\cdots,i_r\}}\tilde{\sigma}_{\{i_1,\cdots,i_r\}})$, $\braket{x_1\oplus x_2, y_1\oplus y_2}=\braket{x_1,y_1}+\braket{x_2,y_2}$, and $\braket{x_1\otimes x_2, y_1\otimes y_2}=\braket{x_1,y_1}\times\braket{x_2,y_2}$.
In common with the shadow kernel, the truncated one has arbitrarily large reduced density matrices and their arbitrarily high-degree polynomials within its feature space.
Meanwhile, unlike the shadow kernel, the truncated one excludes terms where some of $i_1,\cdots,i_r$ are duplicated in Eq.~\eqref{eq_sm: TSK feature vector}.
Eliminating these terms, whose physical meaning is unclear, may improve learning efficiency.
Also, this kernel is bounded as
\begin{align}
    |k^{\rm TSK}(S_T(\rho),S_T(\tilde{\rho}))| 
    &\leq \exp\left[ \frac{\tau}{T^2} \sum_{t,t'=1}^T \prod_{i=1}^n \left|1 + \frac{\gamma}{n}\tr(\sigma_i^{(t)} \tilde{\sigma}_i^{(t')})\right| \right] \\
    &\leq \exp\left[ \tau \left(1 + \frac{5\gamma}{n}\right)^n \right] \\
    &= \exp(\tau\exp(5\gamma)),
\end{align}
where we have used $\tr(\sigma_i^{(t)} \tilde{\sigma}_i^{(t')})=5,1/2,-4$ in the first line and $(1+x/n)^n \leq \exp(x)$ for $n,x\geq 0$ in the second line.

\subsubsection{Polynomial GLQK with truncated shadow kernel}

We consider the following polynomial GLQK:
\begin{align}
    &k_{\rm GL}(S_T(\rho),S_T(\tilde{\rho})) = \left[ \frac{1}{|\AGL(\zeta)|} \sum_{A\in \AGL(\zeta)} k^{\rm TSK}_{A}(S_{T}(\rho),S_{T}(\tilde{\rho})) \right]^\kappa, \label{eq_sm: GLQK}
\end{align}
where $k^{\rm TSK}_{A}(\cdot,\cdot)$ is the truncated shadow kernel limited to the subsystem $A$.
The feature vector of this kernel is given by
\begin{align}
    &\phi_{\rm GL}(S_T(\rho)) \\
    &= \frac{1}{|\AGL(\zeta)|^{\kappa/2}}\left(\bigoplus_{A\in\AGL(\zeta)} \phi^{\rm TSK}_{A}(S_{T}(\rho)) \right)^{\otimes \kappa} \\
    &= \frac{1}{|\AGL(\zeta)|^{\kappa/2}}\left(\bigoplus_{A\in\AGL(\zeta)} \left[\bigoplus_{d=0}^\infty \sqrt{\frac{\tau^d}{d!}} \left( \bigoplus_{r=0}^{|A|} \sqrt{\left(\frac{\gamma}{|A|}\right)^{r}} \bigoplus_{\{i_1,\cdots,i_r\} \subseteq A} \vectorize\left(\sigma_{\{i_1,\cdots,i_r\}}\right) \right)^{\otimes d} \right]\right)^{\otimes \kappa}. \label{eq_sm: feature vector of GLQK-TSK}
\end{align}
Also, this kernel is bounded as
\begin{align}
    |k_{\rm GL}(S_T(\rho),S_T(\tilde{\rho}))| 
    &\leq \left[ \frac{1}{|\AGL(\zeta)|} \sum_{A\in \AGL(\zeta)}  |k^{\rm TSK}_A(S_{T}(\rho),S_{T}(\tilde{\rho}))| \right]^\kappa \\
    &\leq \left[ \frac{1}{|\AGL(\zeta)|} \sum_{A\in \AGL(\zeta)}  \exp(\tau\exp(5\gamma)) \right]^\kappa \\
    &\leq \exp(\kappa \tau\exp(5\gamma)),
\end{align}
where we have used $|k^{\rm TSK}(\cdot,\cdot)|\leq \exp(\tau\exp(5\gamma))$.

We organize the feature vector components.
For an integer $a\leq \kappa$, consider a set of Pauli strings $\mP_j=\{P_{j,1},P_{j,2},\cdots,P_{j,b_j}\}$ over the index $j=1,2,\cdots,a$, where $b_j$ is the number of Pauli strings contained in $\mP_j$.
Assume that for $\forall \mP_j \in \{\mP_1,\ldots,\mP_a\}$, there exists $\exists A_j \in \AGL(\zeta)$ such that $\supp(P)\subseteq A_j$ for $\forall P\in\mP_j$.
Then, there exist the following components in the feature vector:
\begin{align}
    &\frac{1}{|\AGL(\zeta)|^{\kappa/2}}  \prod_{j=1}^a \sqrt{\frac{\tau^{b_j}}{b_j!}} \left( \prod_{P\in \mc{P}_j} \sqrt{\left(\frac{\gamma}{2|A_j|}\right)^{|\supp(P)|} }\tr\left[P\sigma\right] \right) \\
    &=\frac{1}{n^{\kappa/2}}  \prod_{j=1}^a \sqrt{\frac{\tau^{b_j}}{b_j!}} \left( \prod_{P\in \mc{P}_j} \sqrt{\left(\frac{\gamma}{2\zeta^D}\right)^{|\supp(P)|} }\tr\left[P\sigma\right] \right), \label{eq_sm: GLQK component}
\end{align}
where we have used $|\AGL(\zeta)|=n$ and $|A_j|=\zeta^D$ for $\forall A_j \in\AGL(\zeta)$.

\subsection{Estimating polynomial functions from classical shadows}

Here, we derive an error bound when estimating the value of a polynomial $g(\rho)$ from a classical shadow. 
The derivation is based on the following lemma on moment bounds for shadow estimators.
\begin{lem}
\label{lemma:explicit_moment_bound}
Let $P$ be a Pauli operator and $\hat{o}_P$ be the single-shot classical shadow estimator such that $\mathbb{E}[\hat{o}_P] = \tr(P\rho)$. Let $Z_P = \frac{1}{T}\sum_{t=1}^T \hat{o}_P^{(t)} - \tr(P\rho)$ be the estimation error from $T$ samples. 
For any even integer $q \in \{2,4,6,\ldots \}$, the $L_q$-norm of the error is bounded by
\begin{align}
    \|Z_P\|_q \equiv  (\mathbb{E}[|Z_P|^q])^{1/q} \le \sqrt{\frac{2\sigma^2(q+2)}{T}} + \frac{2 \sqrt{2} b(q+1)}{3T},
\end{align}
where $\sigma^2 = \Var[\hat{o}_P]$ is the variance and $b = \sup_{\hat{o}_P} |\hat{o}_P - \tr(P\rho)|$ is the range bound.
\end{lem}

\begin{proof}
We derive the moment bound from Bernstein's inequality via the integral representation of moments.
Applying Bernstein's inequality to $Z_P=\sum_{t=1}^T (\hat{o}_P^{(t)} - \tr(P\rho))/T$, we have the tail bound:
\begin{align}
    \Pr(|Z_P| \ge u) \le 2 \exp\left( - \frac{T u^2}{2\sigma^2 + 2bu/3} \right),
\end{align}
where $\sigma^2=\Var[\hat{o}_P^{(t)} - \tr(P\rho)]=\Var[\hat{o}_P^{(t)}]$ and $b = \sup_{\hat{o}_P} |\hat{o}_P - \tr(P\rho)|$.
We employ the splitting technique for the exponent. Using the inequality $(A+B)^{-1} \ge \frac{1}{2}\min(A^{-1}, B^{-1})$ for $A,B > 0$, we have
\begin{align}
    \frac{T u^2}{2\sigma^2 + 2bu/3} \ge \frac{1}{2} \min\left( \frac{T u^2}{2\sigma^2}, \frac{T u^2}{2bu/3} \right) = \min\left( \frac{T u^2}{4\sigma^2}, \frac{3T u}{4b} \right).
\end{align}
Thus, the probability is bounded by the sum of Gaussian and exponential tails:
\begin{align}
    \Pr(|Z_P| \ge u) \le 2 \exp\left( - \frac{T u^2}{4\sigma^2} \right) + 2 \exp\left( - \frac{3T u}{4b} \right).
\end{align}
The $q$th moment can be expressed as $\mathbb{E}[|Z_P|^q] = \int_0^\infty q u^{q-1} \Pr(|Z_P| \ge u) du$. Substituting the tail bound yields two integrals:
\begin{align}
    A &= \int_0^\infty 2q u^{q-1} \exp\left( - \frac{T u^2}{4\sigma^2} \right) du = 2 \left( \frac{2\sigma}{\sqrt{T}} \right)^q \Gamma\left(\frac{q}{2} + 1\right), \\
    B &= \int_0^\infty 2q u^{q-1} \exp\left( - \frac{3T u}{4b} \right) du = 2 \left( \frac{4b}{3T} \right)^q \Gamma(q+1),
\end{align}
where $\Gamma(\cdot)$ is the Gamma function.
Using the inequality $(x+y)^{1/q} \le x^{1/q} + y^{1/q}$ (for $x,y \ge 0, q \ge 1$), the $L_q$-norm is bounded by
\begin{align}
    \|Z_P\|_q = (A + B)^{1/q} \le 2^{1/q} \frac{2\sigma}{\sqrt{T}} \left[ \Gamma\left(\frac{q}{2} + 1\right) \right]^{1/q} + 2^{1/q} \frac{4b}{3T} [\Gamma(q+1)]^{1/q}.
\end{align}
We use the bounds $[\Gamma(q+1)]^{1/q} \le (q+1)/2$ and $[\Gamma(q/2+1)]^{1/q} \le \sqrt{(q+2)/4}$ for $q\in \{2,4,6,\ldots\}$. 
Also, noting $2^{1/q} \le \sqrt{2}$ for $q \ge 2$, we obtain
\begin{align}
    \|Z_P\|_q &\le \sqrt{\frac{2\sigma^2(q+2)}{T}} + \frac{2 \sqrt{2} b(q+1)}{3T}.
\end{align}
This completes the proof.
\end{proof}

Based on this lemma, we derive an error bound for the shadow estimator of a polynomial $g$:
\begin{lem}
\label{lemma:polynomial_moment_bound}
Consider an $m$-body, degree-$p$ polynomial $g(\rho)$ of an $n$-qubit quantum state $\rho$. 
Suppose we obtain a classical shadow $\sigma$ of size $T$ such that
\begin{align}
    T \ge 16C^2e^2 3^m p^2,
\end{align}
where $e$ is Euler's number, and $C=\sqrt{2} + \sqrt{2}/3e$.
Then, the squared error is bounded by
\begin{align}
    \mathbb{E}_{\sigma \sim \mathcal{D}_{\rho}}[|g(\rho) - g(\sigma)|^2] \leq  \frac{16C^2e^2 \|g\|_1^2 3^m p^2}{T}.
\end{align}
\end{lem}

\begin{proof}
\vspace{0.5cm}
\noindent
\textbf{(i) Generalized moment bounds for Pauli estimators.}
Let $P$ be an arbitrary Pauli operator acting on at most $m$ qubits. The classical shadow estimator for the expectation value $\tr(P\rho)$ is given by the average of single-shot estimators, $\sum_{t=1}^T \hat{o}_P^{(t)}/T$, where $\mathbb{E}[\hat{o}_P^{(t)}] = \tr(P\rho)$.
From the properties of classical shadows based on random Pauli measurements~\cite{Huang2020-ti}, the single-shot estimator $\hat{o}_P^{(t)}$ is bounded as $|\hat{o}_P^{(t)}| \le 3^m$ and has a variance $\Var[\hat{o}_P^{(t)}] \le 3^m$.

Let $Z_P = \sum_{t=1}^T \hat{o}_P^{(t)}/T - \tr(P\rho)$ be the estimation error. 
For later use, we only consider the $L_q$-norm of $Z_P$ for $q \in \{2,4,\ldots,2p\}$.
Using Lemma~\ref{lemma:explicit_moment_bound}, it is bounded by
\begin{align}
    \|Z_P\|_q 
    \leq \sqrt{\frac{2\sigma^2(q+2)}{T}} + \frac{2 \sqrt{2} b(q+1)}{3T}
    \leq \sqrt{\frac{2\cdot 3^m(q+2)}{T}} + \frac{2 \sqrt{2} (3^m+1)(q+1)}{3T},
    \label{eq_sm: ZP_proof}
\end{align}
where we have used $\sigma^2=\Var[\hat{o}_P]\leq 3^m$ and $b=\sup_{\hat{o}_P} |\hat{o}_P - \tr(P\rho)|\leq 3^m+1$.

Here, the second term is bounded as $2e(3^m+1)(q+1)/T \le \sqrt{3^m (q+2)/T}$ for the following reason. 
Since $C^2=(\sqrt{2} + \sqrt{2}/3e)^2\geq 2$, this lemma's condition on $T$ is reduced to $T\geq 16C^2 e^2 3^m p^2\geq36e^2 3^m p^2$.
Furthermore, for $m\geq0$, $p\geq1$, and $q \in \{2,4,\ldots,2p\}$, both $(3^m+1)^2/3^m \leq 4\cdot 3^m$ and $(q+1)^2/(q+2) \leq (q+1)^2/4 \leq (2p+1)^2/4 \leq (3p)^2/4=9p^2/4$ hold.
Using these inequalities, we have $T\geq 36e^2 3^m p^2 = 4e^2 (4\cdot3^m)(9p^2/4)\geq4e^2((3^m+1)^2/3^m)((q+1)^2/(q+2))$.
Simplifying this inequality, we have $4e^2 (3^m+1)^2 (q+1)^2/T^2 \leq 3^m (q+2)/T$.
Taking the square root of both sides results in $2e(3^m+1)(q+1)/T \le \sqrt{3^m (q+2)/T}$, as required.
Substituting this bound into Eq.~\eqref{eq_sm: ZP_proof}, we obtain a simplified rigorous bound:
\begin{align}
    \label{eq:moment_bound}
    \|Z_P\|_q \le \left(\sqrt{2} + \frac{\sqrt{2}}{3e}\right)\sqrt{\frac{3^m(q+2)}{T}} = C \sqrt{\frac{3^m(q+2)}{T}},
\end{align}
for $q\in\{2,4,\ldots,2p\}$.

\vspace{0.5cm}
\noindent
\textbf{(ii) Error decomposition for polynomials.}
Let the polynomial be explicitly written as $g(\rho) = \sum_i c_i \prod_{j=1}^p \tr(P_{ij}\rho)$. 
We aim to bound the $L_2$-norm of the total error. 
By Minkowski's inequality, we have
\begin{align}
    \|g(\sigma) - g(\rho)\|_2 \le \sum_i |c_i| \left\| \prod_{j=1}^p \tr(P_{ij} \sigma) - \prod_{j=1}^p \tr(P_{ij} \rho) \right\|_2.
\end{align}
For each term, let $x_j = \tr(P_{ij} \rho)$, $y_j = \tr(P_{ij} \sigma)$, and $Z_j = y_j - x_j$. Using the algebraic identity for the difference of products,
\begin{align}
    \prod_{j=1}^p y_j - \prod_{j=1}^p x_j &= \prod_{j=1}^p (x_j + Z_j) - \prod_{j=1}^p x_j \nonumber \\
    &= \sum_{\emptyset \ne S \subseteq \{1,\dots,p\}} \left( \prod_{k \in S} Z_k \right) \left( \prod_{l \notin S} x_l \right),
\end{align}
and noting that $|x_l| \le 1$, we apply Minkowski's inequality again:
\begin{align}
    \left\| \prod_{j=1}^p y_j - \prod_{j=1}^p x_j \right\|_2 \le \sum_{\emptyset \ne S \subseteq \{1,\dots,p\}} \left\| \prod_{k \in S} Z_k \right\|_2.
\end{align}
By the generalized Hölder's inequality, the $L_2$-norm of the product is bounded by the product of higher-order norms:
\begin{align}
    \left\| \prod_{k \in S} Z_k \right\|_2 \le \prod_{k \in S} \|Z_k\|_{2|S|}.
\end{align}

\vspace{0.5cm}
\noindent
\textbf{(iii) Summation and bound evaluation.}
We substitute the moment bound of Eq.~\eqref{eq:moment_bound} into the above inequality. 
For a subset $S$ with size $|S|=k\geq1$, $\|Z\|_{2|S|} \le C\sqrt{\frac{3^m (2k+2)}{T}} \le 2Ck\sqrt{\frac{3^m}{T}}$ holds.
We sum over all non-empty subsets $S$, grouped by their size $k$:
\begin{align}
    \left\|\prod_{j=1}^p y_j - \prod_{j=1}^p x_j \right\|_2 &\le \sum_{k=1}^p \binom{p}{k} \left( 2Ck\sqrt{\frac{3^m}{T}} \right)^k.
\end{align}
Using the bound $\binom{p}{k} \le \frac{p^k}{k!}$ and the lower bound from Stirling's approximation $k! \ge (k/e)^k$, we have 
\begin{align}
    \left\|\prod_{j=1}^p y_j - \prod_{j=1}^p x_j \right\|_2 \le \sum_{k=1}^p \left( 2Cep\sqrt{\frac{3^m}{T}} \right)^k,
\end{align}
which is a geometric series $\sum_{k=1}^p \alpha^k$ with $\alpha = 2Cep \sqrt{3^m/T}$. 
The assumption of $T \geq 16 C^2 e^2 3^m p^2$ ensures that $\alpha \le 1/2$. 
Thus, the sum is bounded by $\frac{\alpha(1-\alpha^p)}{1-\alpha} \le 2\alpha$.
The total estimation error is bounded by
\begin{align}
    \|g(\sigma) - g(\rho)\|_2 \le \sum_i |c_i| (2\alpha) = 2\alpha \|g\|_1 = 4Cep \|g\|_1 \sqrt{\frac{3^m}{T}}.
\end{align}
By squaring both sides, we obtain
\begin{align}
    \mathbb{E}_{\sigma\sim\mD_\rho}[|g(\rho)-g(\sigma)|^2] = 
    \|g(\sigma) - g(\rho)\|_2^2 \le  \frac{16C^2e^2 \|g\|_1^2 3^m p^2}{T}.
\end{align}
\end{proof}

From this lemma, we can derive a sufficient shadow size $T$ to estimate $g(\rho)$ with error $\epsilon$:
\begin{lem}
\label{lemma: mother_lem_shadow}
Consider an $m$-body, degree-$p$ polynomial $g(\rho)$ of an $n$-qubit quantum state $\rho$. 
Let $\epsilon\in(0,\|g\|_1)$.
Suppose we obtain a classical shadow $\sigma$ of size $T$ such that
\begin{align}
    T\geq \frac{298\|g\|_1^2 3^m p^2}{\epsilon^2}.
\end{align}
Then, the following inequality holds:
\begin{align}
\mathbb{E}_{\sigma\sim\mD_\rho} [|g(\rho)-g(\sigma)|^2] \leq \epsilon^2.
\end{align}
\end{lem}

\begin{proof}
Since $\epsilon \in(0,\|g\|_1)$ and $16C^2e^2 \leq 298$, we have $T\geq 298 \|g\|_1^2 3^m p^2/\epsilon^2 \geq 298 \cdot 3^m p^2\geq 16C^2e^23^mp^2$, which is the condition for applying Lemma~\ref{lemma:polynomial_moment_bound}.
From Lemma~\ref{lemma:polynomial_moment_bound}, the following inequality holds:
\begin{align}
\mathbb{E}_{\sigma\sim\mD_\rho} [|g(\rho)-g(\sigma)|^2]\leq \frac{16C^2e^2 \|g\|_1^2 3^m p^2}{T}\leq\frac{298 \|g\|_1^2 3^m p^2}{T}.
\end{align}
From $T\geq 298 \|g\|_1^2 3^m p^2/\epsilon^2$, we obtain $\mathbb{E}_{\sigma\sim\mD_\rho} [|g(\rho)-g(\sigma)|^2]\leq \epsilon^2$, as required.
\end{proof}

We emphasize that the number of measurement shots required to estimate $g(\rho)$ with error $\epsilon$, denoted by
\begin{align}
T(g;\epsilon)\equiv 
\frac{298 \|g\|_{1}^2 3^m p^2}{\epsilon^2},
\end{align}
is independent of the system size $n$, provided that $m, p$ and $\|g\|_1$ are fixed.
This lemma can also be applied to the cluster approximation $g_{\rm CA}(\rho)$, which is generally an $m$-body, degree-$mp$ polynomial.
The lemma claims that a classical shadow of size
\begin{align}
T_{\rm CA}(g;\epsilon) \equiv \frac{298 \|g\|_{1}^2 3^m (mp)^2}{\epsilon^2} \geq T(g_{\rm CA};\epsilon) \label{eq_sm: T_cluster}
\end{align}
suffices to estimate $g_{\rm CA}(\rho)$ with accuracy $\epsilon$, where we have replaced $p$ with $mp$ in $T(g;\epsilon)$ and used $\|g\|_1 \geq \|g_{\rm CA}\|_1$ [Eq.~\eqref{eq_sm: cluster norm}].

\section{Theory of kernel ridge regression} \label{sec_sm: ridge regression}

In this section, we review the theory of kernel ridge regression and hyperparameter optimization, deriving fundamental lemmas about generalization error that are essential for proving our main theorems.

\subsection{Ridge regression}

Regression tasks aim to learn an unknown relationship between an input $\bx\in\mX$ and an output $y\in\mY$ under a probability distribution $(\bx,y)\sim\mD$.
Assume that $\mX\subset \mathbb{R}^d$ and $\mY\subset \mathbb{R}$, and that $\|\bx\|\leq X$ and $|y|\leq Y$ for all $\bx\in\mX$ and $y\in\mY$.
In a supervised learning setting, we are given a training dataset $Z = \{\bz_i\}_{i=1}^N$ consisting of $N$ samples drawn independently from the distribution $\mD$, where each sample is defined as $\bz_i = (\bx_i, y_i)$.
Linear regression models the input--output relationship as
\begin{align}
    y \sim h_{\bw}(\bx) = \braket{\bw, \bx},
\end{align}
where $\braket{\bw, \bx}=\bw^{\rm T}\cdot \bx$ denotes the inner product of an input $\bx$ and a trainable weight vector $\bw\in\mathbb{R}^d$.
Here, we assume that the norm of $\bw$ is bounded by $\|\bw\|\leq B$, and denote the corresponding hypothesis class by $\mH=\{h_{\bw}:\mX\to\mY \,|\, \bw\in\mathbb{R}^d, \|\bw\|\leq B\}$.
This norm-restricted linear regression problem is equivalent to ridge regression with $\ell_2$ regularization via the method of Lagrange multipliers, where $B$ corresponds to the regularization parameter $\lambda$.
In practice, the regression problem is often solved using $\ell_2$ regularization rather than a norm constraint, since the former admits a closed-form solution.
Here, however, we adopt the norm-restriction formalism due to the ease of theoretical analysis.

The goal here is to minimize the following expected loss with respect to $h\in\mH$ (equivalently $\bw$):
\begin{align}
    L_{\mD}(h) = \underset{\bz \sim \mD}{\mathbb{E}} \left[\ell(h,\bz)\right], \label{eq_sm: RR expected loss}
\end{align}
where $\ell(h,\bz)=\left(y-h(\bx)\right)^2/2$.
The Lipschitz constant $\eta$ of this loss function $\ell$ is defined by $|\partial \ell/\partial h|=|y-h(\bx)|\leq \eta$.
Since the data distribution $\mD$ is unknown in general, we approximate the expected loss by the empirical loss calculated from the dataset $Z$,
\begin{align}
    L_Z(h) = \frac{1}{N} \sum_{i=1}^N \ell(h,\bz_i), \label{eq_sm: RR emprical loss}
\end{align}
and minimize it to find an optimal $h$ (or $\bw$):
\begin{align}
    h^\ast_Z = \arg\min_{h\in\mH} L_Z(h). \label{eq_sm: RR optimization}
\end{align}
This optimization problem can be efficiently solved on classical computers due to its convexity.

To improve the robustness of the trained model, we introduce a technique called clipping.
Suppose that an upper bound of $|y|$ is known such that $|y|\leq C$, where $Y\leq C$ ($C$ does not necessarily coincide with the supremum $Y$). 
The clipping function $\clip(t)$ is defined by
\begin{align}
\clip(t)=
\begin{cases}
C & \text{for $t\geq C$} \\
t & \text{for $-C<t<C$} \\
-C & \text{for $t\leq -C$}.
\end{cases}
\end{align}
Then, we clip the prediction of the trained model $h_Z^\ast(\bx)$ as $h_Z^{\rm \ast clip}(\bx) = \clip(h_Z^\ast(\bx))$.
This technique limits the output range to $[-C,C]$, thereby improving the robustness of prediction.
For this clipped model, the empirical and expected loss functions are defined by $L_Z^{\rm clip}(h)=\sum_{i=1}^N \ell^{\rm clip}(h,\bz_i)/N$ and $L_{\mD}^{\rm clip}(h)=\mathbb{E}_{\bz\sim\mD}[\ell^{\rm clip}(h,\bz)]$ with $\ell^{\rm clip}(h,\bz)=(y-\clip(h(\bx)))^2/2$.
For clarity, we also denote the unclipped loss functions as $L_Z^{\rm unclip}(h)=\sum_{i=1}^N \ell^{\rm unclip}(h,\bz_i)/N$ and $L_{\mD}^{\rm unclip}(h)=\mathbb{E}_{\bz\sim\mD}[\ell^{\rm unclip}(h,\bz)]$ with $\ell^{\rm unclip}(h,\bz)=(y-h(\bx))^2/2$.

\subsection{Kernel ridge regression}
The kernel method addresses a nonlinear learning task by mapping an input data $\bx \in \mathbb{R}^d$ to a higher-dimensional feature vector $\phi(\bx)\in \mathbb{R}^D$ and then solving a linear optimization problem in the resulting feature space.
This formulation is parallel to direct regression on $\bx$.
A linear model in the feature space is defined as $h_{\bw}(\bx) = \braket{\bw, \phi(\bx)}$ with a weight vector $\bw\in \mathbb{R}^D$.
The expected and empirical losses are defined similarly by $L_{\mD}^{\rm unclip}(h)=\mathbb{E}_{\bz\sim\mD}[\ell^{\rm unclip}(h,\bz)]$ and $L_Z^{\rm unclip}(h)=\sum_{i=1}^N \ell^{\rm unclip}(h,\bz_i)/N$ with $\ell^{\rm unclip}(h,\bz)=(y-h(\bx))^2/2$.
We optimize $h$ (or $\bw$) by minimizing the loss function within a hypothesis class $\mH=\{h_{\bw}:\mX\to\mY \,|\, \bw\in\mathbb{R}^D, \|\bw\|\leq B\}$ as $h_Z^\ast = \arg\min_{\bw\in\mH} L_Z^{\rm unclip}(\bw)$.

By the representer theorem, the optimal weight vector $\bw^\ast_Z$ corresponding to $h_Z^\ast(\bx)$ can be expressed as a linear combination of the training data as $\bw^\ast_Z = \sum_{j=1}^N (\balpha_Z^\ast)_j \phi(\bx_j)$, where $\balpha_Z^\ast \in \mathbb{R}^N$ is an $N$-dimensional vector.
The resulting linear model can be written as
\begin{align}
    h_Z^\ast(\bx)=\sum_{j=1}^N (\balpha_Z^\ast)_j k(\bx_j,\bx) \label{eq_sm: kernel model}
\end{align}
with the kernel function $k(\bx,\bx')=\braket{\phi(\bx),\phi(\bx')}$.
Therefore, given an unseen input data $\bx$, we can predict its output $y$ through Eq.~\eqref{eq_sm: kernel model} using the kernel function between $\bx$ and the training data $\bx_j$.
Substituting $\bw^\ast_Z = \sum_{j=1}^N (\balpha_Z^\ast)_j \phi(\bx_j)$ into Eq.~\eqref{eq_sm: RR optimization}, we have the optimization problem for $\balpha_Z^\ast$:
\begin{align}
    \balpha^\ast_Z 
    = \arg\min_{\balpha}\left( \frac{1}{2N} \balpha^{\rm T} KK \balpha - \frac{1}{N}\balpha^{\rm T} K \by + \frac{1}{2N} \by^{\rm T} \by \right)  \\
    \text{subject to $\balpha^{\rm T} K \balpha \leq B^2$}, \label{eq_sm: KRR optimization}
\end{align}
where we have defined $\balpha = (\alpha_1,\cdots,\alpha_N)^{\rm T}$, $\by = (y_1,\cdots,y_N)^{\rm T}$, and the kernel matrix $(K)_{ij}=k(\bx_i,\bx_j)$.
In this representation, we do not need to calculate the feature vectors explicitly, but only the $N$-dimensional kernel matrix.
This enables us to treat high-dimensional, potentially infinite-dimensional, feature spaces that cannot be computed 
directly.
Equation~\eqref{eq_sm: KRR optimization} is a convex optimization problem and thus can be solved efficiently with classical computers.

In the following subsections, we prove several lemmas for ridge regression. 
These are also applicable to the kernel method by viewing it as ridge regression in a feature space. 
In this context, the input vector $\bx$ is replaced by the feature vector $\phi(\bx)$. 
Consequently, the bound on the input norm, $\|\bx\|\leq X$, is replaced by a bound on the feature vector norm, $\|\phi(\bx)\|\leq X$, as determined by the kernel function: $k(\bx,\bx)\leq X^2$ for all $\bx\in\mX$.

\subsection{Rigorous guarantee for ridge regression}

In ridge regression, the performance of the trained model with clipping can be evaluated within the framework of statistical learning theory.
Specifically, the following lemma holds:
\begin{lem}
\label{lemma: main_krr_lemma1}
With probability at least $1-\delta$, the following inequality holds:
\begin{align}
L_{\mD}^{\rm clip}(h_Z^\ast) 
&\leq \frac{3}{2}L_{\mD}^{\rm unclip}(h^\ast) + \frac{2(C+Y)BX}{\sqrt{N}} + \frac{1}{2}(C+Y)^2 \sqrt{\frac{\log(2/\delta)}{2N}} + \frac{5(BX+Y)^2 \log(4/\delta)}{6N},
\end{align}
where $h^\ast = \arg\min_{h\in\mH} L_{\mD}^{\rm unclip}(h)$ is the optimal model.
\end{lem}

To prove this lemma, we first establish two lemmas concerning generalization error bounds.
Because these lemmas hold for general loss functions, we omit the clip or unclip symbol.
The first one is the standard bound based on the Rademacher complexity:
\begin{lem}
\label{lemma: generalization_error_1}
Let $\ell$ be an $\eta$-Lipschitz loss function taking values in $[0,b]$.
Then, with probability at least $1-\delta$, the following inequality holds:
\begin{align}
L_{\mD}(h) \leq  L_{Z}(h) + \frac{2\eta BX}{\sqrt{N}} + b \sqrt{\frac{\log(1/\delta)}{2N}}
\end{align}
for all $h\in\mH$.
\end{lem}

\begin{proof}
Let $\mG=\{g:\bz\mapsto \ell(h,\bz) | h\in\mH\}$ be the family of loss functions associated with the hypothesis class $\mH$.
According to Theorem 3.3 in Ref.~\cite{Mohri2012-rl}, the following inequality holds for all $h\in\mH$ with probability at least $1-\delta$:
\begin{align}
L_{\mD}(h) \leq L_{Z}(h) + 2 \mathfrak{R}_N(\mG) + b\sqrt{\frac{\log(1/\delta)}{2N}}, \label{eq_sm: rademacher_ineq}
\end{align}
where $\mathfrak{R}_N(\mG)$ is the Rademacher complexity of $\mG$ defined as
\begin{align}
\mathfrak{R}_N(\mG)=\mathbb{E}_{(\bz_1,\bz_2,\ldots)\sim\mD^N} \left[ \mathbb{E}_{\sigma\sim\Unif(\{\pm1\}^N)} \left[ \sup_{g\in\mG} \frac{1}{N} \sum_{i=1}^N g(\bz_i)\sigma_i \right]\right].
\end{align}
Talagrand's lemma (Lemma 5.7 in Ref.~\cite{Mohri2012-rl}) states that $\mathfrak{R}_N(\mG)\leq \eta \mathfrak{R}_N(\mH)$ holds, where $\mathfrak{R}_N(\mH)$ is the Rademacher complexity of $\mH$ defined as
\begin{align}
\mathfrak{R}_N(\mH)=\mathbb{E}_{(\bz_1,\bz_2,\ldots)\sim\mD^N} \left[ \mathbb{E}_{\sigma\sim\Unif(\{\pm1\}^N)} \left[ \sup_{h\in\mH} \frac{1}{N} \sum_{i=1}^N h(\bx_i)\sigma_i \right]\right].
\end{align}
From Theorem 6.12 in Ref.~\cite{Mohri2012-rl}, $\mathfrak{R}_N(\mH)$ satisfies $\mathfrak{R}_N(\mH)\leq \sqrt{B^2X^2/N}$.
Substituting $\mathfrak{R}_N(\mG)\leq \eta \mathfrak{R}_N(\mH) \leq \eta \sqrt{B^2X^2/N}$ into Eq.~\eqref{eq_sm: rademacher_ineq}, we obtain the inequality required.
\end{proof}

The second lemma gives a tighter bound using a concentration inequality based on the variance:
\begin{lem}
\label{lemma: generalization_error_2}
Let $\ell$ be the loss function in $[0,b]$.
Then, with probability at least $1-\delta$, the following inequality holds:
\begin{align}
|L_{\mD}(h) - L_{Z}(h)| \leq \frac{1}{2} L_{\mD}(h) + \frac{5b \log(2/\delta)}{3N}
\end{align}  
for any fixed $h\in\mH$ independent of $Z$.
\end{lem}

\begin{proof}
The proof is based on Bernstein's inequality. 
Let $X_1,\ldots,X_N$ be i.i.d. random variables such that $|X_i - \mathbb{E}[X_i]|\leq b$.
Then, the following Bernstein's inequality holds:
\begin{align}
\Pr\left[\left|\mathbb{E}[X] - \frac{1}{N}\sum_{i=1}^{N}X_i \right| \geq \epsilon\right] 
\leq 
2\exp\left[-\frac{N \epsilon^2}{2\sigma^2+(2/3)b\epsilon}\right],
\end{align}
where $\sigma^2 = \Var[X_i]$.
In our case, $X_i=\ell(h,\bz_i)$ is a random variable over $\bz_i\sim\mD$ such that $|X_i|\leq b$.
Their mean and expectation correspond to the empirical loss $\sum_{i=1}^{N}X_i/N = L_Z(h)$ and expected loss $\mathbb{E}[X_i]=L_{\mD}(h)$, respectively.
From Bernstein's inequality, we have
\begin{align}
\Pr\left[\left|L_{\mD}(h) - L_Z(h) \right| \geq \epsilon\right] 
\leq 
2\exp\left[-\frac{N \epsilon^2}{2\sigma^2+(2/3)b\epsilon}\right],
\end{align}
where $\sigma^2 = \Var_{\bz\sim \mD}[\ell(h,\bz)]$.
Letting the right-hand side be $\delta$ and solving the equation with respect to $\epsilon$, the following inequality holds with probability at least $1-\delta$:
\begin{align}
\left|L_{\mD}(h) - L_Z(h) \right| \leq \sqrt{\frac{2\sigma^2\log(2/\delta)}{N}} + \frac{2b\log(2/\delta)}{3N}.
\end{align}
Here, from $0\leq \ell(h,\bz) \leq b$, $\sigma^2$ satisfies
\begin{align}
    \sigma^2
    &= \mathbb{E}_{\bz}[\ell^2(h,\bz)] - (\mathbb{E}_{\bz}[\ell(h,\bz)])^2 \\
    &\leq \mathbb{E}_{\bz}[\ell^2(h,\bz)] \\
    &\leq b \mathbb{E}_{\bz}[\ell(h,\bz)] \\
    &= bL_{\mD}(h).
\end{align}
Therefore, we have
\begin{align}
\left|L_{\mD}(h) - L_Z(h) \right| \leq \sqrt{\frac{2b L_{\mD}(h)\log(2/\delta)}{N}} + \frac{2b\log(2/\delta)}{3N}.
\end{align}
Using $\sqrt{xy}\leq (x+y)/2$ with $x=L_{\mD}(h)$ and $y=2b\log(2/\delta)/N$ for the first term on the right-hand side, we obtain
\begin{align}
\left|L_{\mD}(h) - L_Z(h) \right| \leq \frac{1}{2}L_{\mD}(h) + \frac{5b\log(2/\delta)}{3N}.
\end{align}

\end{proof}

Based on these lemmas, we can prove Lemma~\ref{lemma: main_krr_lemma1} as follows:
\begin{proof}[Proof of Lemma~\ref{lemma: main_krr_lemma1}]
For the clipped loss function $\ell^{\rm clip}(h,\bz)=(y-\clip(h(\bx)))^2/2$, the Lipschitz constant is $|\partial \ell/\partial h|=|y-\clip(h)|\cdot |\clip'(h)|\leq C+Y \equiv \eta$.
Also, the clipped loss function is bounded as $|\ell^{\rm clip}(h,\bz)|\leq (C+Y)^2/2 \equiv b$.
Therefore, from Lemma~\ref{lemma: generalization_error_1}, the empirical and expected losses for $h_Z^\ast$ satisfy
\begin{align}
L_{\mD}^{\rm clip}(h_Z^\ast) \leq  L_{Z}^{\rm clip}(h_Z^\ast) + \frac{2(C+Y) BX}{\sqrt{N}} + \frac{1}{2}(C+Y)^2 \sqrt{\frac{\log(2/\delta)}{2N}} 
\label{eq_sm: clip_ineq}
\end{align}
with probability at least $1-\delta/2$.

The unclipped loss function $\ell^{\rm unclip}(h,\bz)=(y-h(\bx))^2/2$ is bounded as $|\ell^{\rm unclip}(h,\bz)|\leq (BX+Y)^2/2$.
From Lemma~\ref{lemma: generalization_error_2}, as $h^\ast$ is independent of the training dataset $Z$, the following inequality holds:
\begin{align}
L_{Z}^{\rm unclip}(h^\ast) \leq \frac{3}{2} L_{\mD}^{\rm unclip}(h^\ast) + \frac{5(BX+Y)^2 \log(4/\delta)}{6N}
\label{eq_sm: unclip_ineq}
\end{align} 
with probability at least $1-\delta/2$.

Since $y\in[-Y,Y]$ and $Y\leq C$, we have $(y-\clip(h(\bx)))^2 \leq (y-h(\bx))^2$.
Therefore, the following holds:
\begin{align}
L_Z^{\rm clip}(h_Z^\ast) \leq L_Z^{\rm unclip}(h_Z^\ast).
\label{eq_sm: clip_unclip_ineq}
\end{align}
Also, due to the optimality of $h_Z^\ast$ for the empirical loss $L_Z^{\rm unclip}$, we have
\begin{align}
L_Z^{\rm unclip}(h_Z^\ast) \leq L_Z^{\rm unclip}(h^\ast).
\label{eq_sm: optimality}
\end{align}

From the union bound, Eqs.~\eqref{eq_sm: clip_ineq} and \eqref{eq_sm: unclip_ineq} simultaneously hold with probability at least $1-\delta$.
Using these results, we obtain
\begin{align}
L_{\mD}^{\rm clip}(h_Z^\ast) 
&\leq L_{Z}^{\rm clip}(h_Z^\ast) + \frac{2(C+Y)BX}{\sqrt{N}} + \frac{1}{2}(C+Y)^2 \sqrt{\frac{\log(2/\delta)}{2N}}\\
&\leq L_{Z}^{\rm unclip}(h_Z^\ast) + \frac{2(C+Y)BX}{\sqrt{N}} + \frac{1}{2}(C+Y)^2 \sqrt{\frac{\log(2/\delta)}{2N}}\\
&\leq L_{Z}^{\rm unclip}(h^\ast) + \frac{2(C+Y)BX}{\sqrt{N}} + \frac{1}{2}(C+Y)^2 \sqrt{\frac{\log(2/\delta)}{2N}}\\
&\leq \frac{3}{2}L_{\mD}^{\rm unclip}(h^\ast) + \frac{2(C+Y)BX}{\sqrt{N}} + \frac{1}{2}(C+Y)^2 \sqrt{\frac{\log(2/\delta)}{2N}} + \frac{5(BX+Y)^2 \log(4/\delta)}{6N},
\end{align}
where we have used Eq.~\eqref{eq_sm: clip_ineq} in the first line, Eq.~\eqref{eq_sm: clip_unclip_ineq} in the second line, Eq.~\eqref{eq_sm: optimality} in the third line, and Eq.~\eqref{eq_sm: unclip_ineq} in the fourth line.
\end{proof}

\subsection{Hyperparameter optimization}

The kernel method often involves several hyperparameters, which should be optimized to obtain an accurate predictor. 
A standard approach to hyperparameter optimization is to use validation.
In this approach, the available dataset is partitioned into two subsets.
One is used to train each ML model with different hyperparameter values, and the other is used to validate their performance and select the best-performing model.
We denote the training and validation datasets as $Z$ and $V$, consisting of $N_z$ and $N_v$ data points, respectively.
Note that $Z$ and $V$ are independently sampled from $\mD$.

The procedure is as follows.
Let $\mP=\{\btheta_1,\btheta_2,\ldots\}$ be a finite set of candidate hyperparameter configurations, and each $\btheta \in \mP$ be associated with a hypothesis class $\mH(\btheta)$. 
For each hyperparameter value $\btheta\in\mP$, we first train a kernel model $h_{\btheta,Z}^\ast = \arg\min_{h\in\mH(\btheta)}L_{Z}^{\rm unclip}(h) \in \mH(\btheta)$ using the training dataset $Z$.
To select the best hyperparameter value from $\mP$, we evaluate the validation loss $L_V^{\rm clip}(h_{\btheta,Z}^\ast)=\sum_{i=1}^{N_v} (y_i^{\rm val} - \clip(h_{\btheta,Z}^\ast(\bx_i^{\rm val})))^2/2N_v$ on the validation dataset $V=\{(\bx_i^{\rm val},y_i^{\rm val})\}_{i=1}^{N_v}$.
Based on the validation results, we select $\btheta^\ast_V$ minimizing the validation loss: $\btheta^\ast_V =\arg\min_{\btheta\in\mP} L_V^{\rm clip}(h_{\btheta,Z}^\ast)$.

For later use, we also define the optimal hyperparameter value as $\btheta^\ast=\arg\min_{\btheta\in\mP}L_{\mD}^{\rm unclip}(h_{\btheta}^\ast)$, where $h_{\btheta}^\ast = \arg\min_{h\in\mH(\btheta)}L_{\mD}^{\rm unclip}(h)$.
This optimal hyperparameter value results in $h^\ast_{\btheta^\ast}=\arg\min_{h\in \mH}L_{\mD}^{\rm unclip}(h)$, where $\mH=\mH(\btheta_1)\cup\mH(\btheta_2)\cup\cdots$ is the union of all hypothesis classes.

The hyperparameter $\btheta_V^\ast$ selected in this way yields a reliable predictor.
This is guaranteed by the following lemma.
\begin{lem}
\label{lemma: main_krr_lemma2}
With probability at least $1-\delta$, the following inequality holds:
\begin{align}
&L_{\mD}^{\rm clip}(h_{\btheta_V^\ast,Z}^{\ast}) 
\leq 3 L_{\mD}^{\rm clip}(h_{\btheta^\ast,Z}^{\ast}) + \frac{10(C+Y)^2\log(2|\mP|/\delta)}{3N_v}.
\end{align}
\end{lem}

This lemma shows that the expected loss of the model selected through validation is upper bounded by that of the optimal model with a constant factor (the first term) and an additional correction due to validation (the second term).
The first term can be small if $\mP$ includes good hyperparameter values such that $L_{\mD}^{\rm clip}(h^\ast_{\btheta})\approx 0$ and the training dataset size $N_z$ is sufficiently large, while the second term can be suppressed by increasing the number of validation data as $N_v \sim (C+Y)^2 \log(|\mP|/\delta)$.
The logarithmic dependence of $|\mP|$ implies that this hyperparameter optimization process does not greatly degrade the generalization performance of the regression.

\begin{proof}

From Lemma~\ref{lemma: generalization_error_2}, since $V$ is independent of $Z$, we have
\begin{align}
|L_{\mD}^{\rm clip}(h_{\btheta,Z}^\ast) - L_{V}^{\rm clip}(h_{\btheta,Z}^\ast)| \leq \frac{1}{2} L_{\mD}^{\rm clip}(h_{\btheta,Z}^\ast) + \frac{5b \log(2|\mP|/\delta)}{3N_v} \label{eq: DV_ineq}
\end{align}  
for any fixed $\btheta\in\mP$, with probability at least $1-\delta/|\mP|$.
Here $b=(C+Y)^2/2$ is the upper bound of the loss function.
Note that the assumption on the independence of the validation dataset satisfies the condition of this lemma.
Using the union bound, Eq.~\eqref{eq: DV_ineq} simultaneously holds for all $\btheta\in\mP$ with probability $1-\delta$.

We evaluate the difference between the two loss functions; one involves hyperparameters $\btheta_V^\ast$ determined from the validation data $V$, and the other is calculated from the optimal ones $\btheta^\ast$.
We have
\begin{align}
&L_{\mD}^{\rm clip}(h_{\btheta_V^\ast,Z}^\ast) - L_{\mD}^{\rm clip}(h_{\btheta^\ast,Z}^\ast) \\
&\leq 
\left( L_{\mD}^{\rm clip}(h_{\btheta_V^\ast,Z}^\ast) - L_{V}^{\rm clip}(h_{\btheta_V^\ast,Z}^\ast) \right) 
 +\left( L_{V}^{\rm clip}(h_{\btheta_V^\ast,Z}^\ast) - L_{V}^{\rm clip}(h_{\btheta^\ast,Z}^\ast)\right) 
 + \left( L_{V}^{\rm clip}(h_{\btheta^\ast,Z}^\ast) - L_{\mD}^{\rm clip}(h_{\btheta^\ast,Z}^\ast) \right) \\
&\leq \frac{1}{2}L_{\mD}^{\rm clip}(h_{\btheta_V^\ast,Z}^\ast) + \frac{1}{2}L_{\mD}^{\rm clip}(h_{\btheta^\ast,Z}^\ast) + \frac{10b \log(2|\mP|/\delta)}{3N_v},
\end{align}
where we have used $L_{V}^{\rm clip}(h_{\btheta_V^\ast,Z}^\ast) - L_{V}^{\rm clip}(h_{\btheta^\ast,Z}^\ast) \leq 0$ and Eq.~\eqref{eq: DV_ineq}.
Simplifying this equation and substituting $b=(C+Y)^2/2$, we obtain
\begin{align}
&L_{\mD}^{\rm clip}(h_{\btheta_V^\ast,Z}^\ast) 
\leq 3 L_{\mD}^{\rm clip}(h_{\btheta^\ast,Z}^\ast) + \frac{10(C+Y)^2\log(2|\mP|/\delta)}{3N_v}.
\end{align}

\end{proof}

\subsection{Rigorous guarantee for ridge regression with hyperparameter optimization}

So far, we have theoretically evaluated kernel ridge regression and hyperparameter optimization separately in Lemmas~\ref{lemma: main_krr_lemma1} and \ref{lemma: main_krr_lemma2}.
The following lemma quantifies the performance of the entire training procedure in a single lemma:
\begin{lem}
\label{lemma: main_krr_lemma3}
With probability at least $1-\delta$, the following inequality holds:
\begin{align}
L_{\mD}^{\rm clip}(h_{\btheta_V^\ast,Z}^\ast) 
\leq 
&\frac{9}{2} L_{\mD}^{\rm unclip}(h_{\btheta^\ast}^\ast)  \\
&+ \frac{6(C+Y)BX}{\sqrt{N_z}} 
+ \frac{3}{2}(C+Y)^2 \sqrt{\frac{\log(4/\delta)}{2N_z}} 
+ \frac{5(BX+Y)^2 \log(8/\delta)}{2N_z}  \\
&+ \frac{10(C+Y)^2\log(4|\mP|/\delta)}{3N_v},
\end{align}
where $X$ and $B$ are the norm bound of the feature vector $\phi(\bx)$ and the weight vector $\bw$ for hyperparameters $\btheta^\ast$, respectively.
\end{lem}

Here, $h_{\btheta_V^\ast,Z}^\ast$ is the model actually obtained via training and validation, while $h_{\btheta^\ast}^\ast = \arg\min_{h\in\mH}L_{\mD}^{\rm unclip}(h)$ denotes the optimal model minimizing the unclipped loss function over all hypothesis classes $\mH=\mH(\btheta_1)\cup\mH(\btheta_2)\cup\cdots$.
Therefore, this lemma upper bounds the accuracy of the obtained model by the optimal one.

\begin{proof}
By Lemma~\ref{lemma: main_krr_lemma2}, we have
\begin{align}
&L_{\mD}^{\rm clip}(h_{\btheta_V^\ast,Z}^{\ast}) 
\leq 3 L_{\mD}^{\rm clip}(h_{\btheta^\ast,Z}^{\ast}) + \frac{10(C+Y)^2\log(4|\mP|/\delta)}{3N_v}
\end{align}
with probability at least $1-\delta/2$.
Furthermore, Lemma~\ref{lemma: main_krr_lemma1} implies that
\begin{align}
L_{\mD}^{\rm clip}(h_{\btheta^\ast,Z}^\ast) 
\leq 
&\frac{3}{2} L_{\mD}^{\rm unclip}(h_{\btheta^\ast}^\ast) 
+ \frac{2(C+Y)BX}{\sqrt{N_z}} 
+ \frac{1}{2}(C+Y)^2 \sqrt{\frac{\log(4/\delta)}{2N_z}} 
+ \frac{5(BX+Y)^2 \log(8/\delta)}{6N_z} 
\end{align}
with probability at least $1-\delta/2$.
From the union bound, these two inequalities hold simultaneously with probability at least $1-\delta$.
By combining them, we readily obtain the required inequality.
\end{proof}

\begin{lem}
\label{lemma: mother_lem}
For any $\epsilon$, suppose we obtain $N_z$ training data and $N_v$ validation data such that
\begin{align}
&N_z \geq \max\left(\frac{400 C^2 B^2 X^2}{\epsilon^4}, \frac{200C^4 \log(4/\delta)}{\epsilon^4}, \frac{50(BX+C)^2 \log(8/\delta)}{\epsilon^2}\right), \\
&N_v \geq \frac{800C^2 \log(4|\mP|/\delta)}{3\epsilon^2},
\end{align}
where $X$ and $B$ are the norm bound of the feature vector $\phi(\bx)$ and the weight vector $\bw$ for hyperparameters $\btheta^\ast$, respectively.
Then, the following inequality holds with probability at least $1-\delta$:
\begin{align}
L_{\mD}^{\rm clip}(h_{\btheta_V^\ast,Z}^\ast) 
\leq 
&\frac{9}{2} L_{\mD}^{\rm unclip}(h_{\btheta^\ast}^\ast) + \epsilon^2. \label{eq_sm: mother_ineq}
\end{align}
\end{lem}

\begin{proof}
From $N_z \geq 400C^2B^2X^2/\epsilon^4, 200C^4 \log(4/\delta)/\epsilon^4, 50(BX+C)^2 \log(8/\delta)/\epsilon^2$ and $Y\leq C$, we have
\begin{align}
&\frac{6(C+Y)BX}{\sqrt{N_z}} \leq \frac{6(C+C)BX}{\sqrt{\frac{400C^2 B^2 X^2}{\epsilon^4}}} = \frac{6\epsilon^2}{10}, \\
&\frac{3}{2}(C+Y)^2\sqrt{\frac{\log(4/\delta)}{2N_z}} \leq \frac{3}{2}(C+C)^2\sqrt{\frac{\log(4/\delta)}{2 \cdot \frac{200C^4 \log(4/\delta)}{\epsilon^4}}} = \frac{3\epsilon^2}{10}, \\
&\frac{5(BX+Y)^2\log(8/\delta)}{2N_z} \leq \frac{5(BX+C)^2\log(8/\delta)}{2\cdot\frac{50(BX+C)^2 \log(8/\delta)}{\epsilon^2}} = \frac{\epsilon^2}{20}, 
\end{align}
Furthermore, since $N_v\geq800C^2 \log(4|\mP|/\delta)/3\epsilon^2$, 
\begin{align}
&\frac{10(C+Y)^2\log(4|\mP|/\delta)}{3N_v} \leq \frac{10(C+C)^2\log(4|\mP|/\delta)}{3\cdot \frac{800C^2 \log(4|\mP|/\delta)}{3\epsilon^2}} = \frac{\epsilon^2}{20}
\end{align}
holds.
Combining these inequalities with Lemma~\ref{lemma: main_krr_lemma3}, we obtain
\begin{align}
L_{\mD}^{\rm clip}(h_{\btheta_V^\ast,Z}^\ast) 
\leq 
&\frac{9}{2} L_{\mD}^{\rm unclip}(h_{\btheta^\ast}^\ast) + \frac{6\epsilon^2}{10}  + \frac{3\epsilon^2}{10}  + \frac{\epsilon^2}{20}  + \frac{\epsilon^2}{20} = \frac{9}{2} L_{\mD}^{\rm unclip}(h_{\btheta^\ast}^\ast) + \epsilon^2.
\end{align}
\end{proof}

\section{Rigorous guarantee for GLQK} \label{sec_sm: GLQK}

In this section, we evaluate the amount of quantum resources sufficient for GLQK to learn an unknown polynomial $g(\rho)$ from classical shadow data.
Let $\mD_S$ be the probability distribution over $\mX\times\mY$, where $\mX$ is the input domain of $n$-qubit classical shadows of size $T$, and $\mY\subset\mathbb{R}$ denotes the output range. 
This distribution is constructed as follows. 
First, a quantum state $\rho$ is drawn from a distribution $\mD$. 
Then, its classical shadow $S_T(\rho)$ is generated from $\mD_\rho$ by performing random Pauli measurements on $T$ copies of $\rho$. 
This procedure induces the final data distribution $\mD_S$ over pairs $(S_T(\rho), g(\rho))$, where $g(\rho)$ is the corresponding target label. 
We assume that any state $\rho$ drawn from $\mD$ satisfies the ECP with a correlation length bounded by $\xi$. 
We further assume the output is bounded by a constant $C$, i.e., $|y|\leq C$ for all $y\in\mY$.

Suppose that a training dataset $Z\sim \mD_S^{N_z}$ of size $N_z$ and a validation dataset $V \sim \mD_S^{N_v}$ of size $N_v$ are given.
Let $\mP$ be a set of candidate hyperparameter configurations, where each $\btheta \in \mP$ specifies parameters such as the subsystem size $\zeta$, the kernel exponent $\kappa$, and the regularization parameter $B$.
For each $\btheta\in\mP$, we first obtain a GLQK model $h_{\btheta,Z}^\ast$ by minimizing the empirical loss on the training dataset $Z$.
The predictions of this model are then clipped with the bound $C$.
Finally, we select the best hyperparameter values by minimizing the validation loss over $\mP$.
Let $h_{ZV}^\ast$ denote the final model obtained through this procedure.

Below, we evaluate sufficient quantum resources to guarantee that the final model achieves $L_{\mD_S}^{\rm clip}(h_{ZV}^\ast)\leq\epsilon^2$.
The details of the GLQK and training procedure are described in Secs.~\ref{sec_sm: shadow} and \ref{sec_sm: ridge regression}.

\subsection{General cases}

In this learning task, the performance of GLQK is guaranteed by the following theorem:
\begin{thm}[Theorem~\ref{thm: main theorem 1} in the main text] \label{thm_sm: main_nontranslation}
Consider an $m$-body, degree-$p$ polynomial $g(\rho)$ and a distribution $\mD_S$ over $\mX\times\mY$ such that the sampled quantum state has correlation length at most $\xi$ on a $D$-dimensional hypercubic lattice.
Assume that an upper bound of $|g(\rho)|$ is given: $|g(\rho)|\leq C$ for all $\rho$ drawn from the distribution.
For $\epsilon\in (0,\|g\|_1)$, let $\Delta = \xi \log(\sqrt{18} \|g\|_{1} mp/\epsilon)$, $\zeta \geq m\Delta$, $\kappa\geq \LCN(g;\Delta,\zeta)$, and $B^2 \geq \|g\|_{1}^2 ( 2m p \zeta^D/\tau\gamma )^{mp} n^\kappa$.
Let $\mP=\{(\zeta_i,\kappa_i,B_i)\}_{i}$ be a finite set of candidate hyperparameter values such that $(\zeta,\kappa,B)\in\mP$, where $\zeta_i$ denotes the subsystem size, $\kappa_i$ the GLQK exponent, and $B_i$ the regularization parameter.
For $\delta\in(0,1)$, suppose that we obtain $N=N_z+N_v$ classical shadows of size $T$, together with their target labels, as training and validation datasets such that
\begin{align}
    &N_z \geq \max\left(\frac{1600C^2 B^2X^2}{\epsilon^4}, \frac{800 C^4 \log(4/\delta)}{\epsilon^4}, \frac{100(BX+C)^2 \log(8/\delta)}{\epsilon^2}\right), \\
    &N_v \geq \frac{1600C^2 \log(4|\mP|/\delta)}{3\epsilon^2}, \\
    &T\geq \frac{5364 \|g\|_1^2 3^m (mp)^2}{\epsilon^2},
\end{align}
where $X^2=\exp(\kappa\tau \exp(5\gamma))$.
Then, kernel ridge regression trained on $N_z$ training data using the polynomial GLQK based on the truncated shadow kernel, combined with hyperparameter optimization via $N_v$ validation data, achieves
\begin{align}
    L_{\mD_S}^{\rm clip}(h_{ZV}^\ast) \leq \epsilon^2
\end{align}
with probability at least $1-\delta$.
Here, we have assumed that $2\zeta^D/\gamma \geq 1$ and $mp/\tau \geq 1$.
\end{thm}    

\begin{proof}
We prove this theorem based on Lemma~\ref{lemma: mother_lem}.

\vspace{0.5cm}
\noindent
{\bf (i) Error in estimating the polynomial from classical shadows:}
First, we evaluate the first term on the right-hand side in Eq.~\eqref{eq_sm: mother_ineq}.
Let $\Delta=\xi \log(\sqrt{18}\|g\|_1 mp/\epsilon)$ and $T=T_{\rm CA}(g;\epsilon/\sqrt{18})\geq T(g_{\rm CA};\epsilon/\sqrt{18})$.
Then, by Lemmas~\ref{lem_sm: approx} and \ref{lemma: mother_lem_shadow}, the $\Delta$-cluster approximation and its value estimated from a classical shadow $\sigma$ obey
\begin{align}
    &|g(\rho) - g_{\rm CA}(\rho)| \leq \epsilon/\sqrt{18}, \\
    &\underset{\sigma\sim\mD_\rho}{\mathbb{E}} \left[ |g_{\rm CA}(\rho) - g_{\rm CA}(\sigma)|^2 \right] \leq \epsilon^2/18
\end{align}
for any $\rho$ with a correlation length less than or equal to $\xi$.
Meanwhile, the following inequality holds:
\begin{align}
    |g(\rho)-g_{\rm CA}(\sigma)|^2/2
    &= |(g(\rho)-g_{\rm CA}(\rho)) + (g_{\rm CA}(\rho) - g_{\rm CA}(\sigma))|^2/2 \\
    &\leq |g(\rho)-g_{\rm CA}(\rho)|^2 + |g_{\rm CA}(\rho) - g_{\rm CA}(\sigma)|^2,
\end{align}
where we have used $|x+y|^2/2 \leq |x|^2+|y|^2$.
Taking expectation values with respect to $\rho\sim\mD$ and $\sigma\sim\mD_\rho$, we have
\begin{align}
    \underset{\rho\sim \mD}{\mathbb{E}} \left[\underset{\sigma \sim \mD_\rho}{\mathbb{E}} \left[ |g(\rho) - g_{\rm CA}(\sigma)|^2/2 \right]\right] \leq \epsilon^2/9.
\end{align}
If $g_{\rm CA}(\sigma)$ can be represented as a linear function in the feature space of GLQK with hyperparameters $\btheta=(\zeta, \kappa, B)$, i.e., $g_{\rm CA}(\sigma)=\braket{\tilde{\bw}, \phi_{{\rm GL}, \zeta,\kappa}(\sigma)}\equiv h_{\tilde{\bw},\zeta,\kappa}(\sigma)$ for some $\tilde{\bw}$ such that $\|\tilde{\bw}\|\leq B$ ($\phi_{{\rm GL},\zeta,\kappa}$ is the feature vector of GLQK with a subsystem size $\zeta$ and an exponent $\kappa$), the first term on the right-hand side in Eq.~\eqref{eq_sm: mother_ineq} is upper bounded by $(9/2)\times(\epsilon^2/9)=\epsilon^2/2$:
\begin{align}
\frac{9}{2}\min_{h\in\mH} L_{\mD_S}^{\rm unclip}(h) \leq \frac{9}{2}L_{\mD_S}^{\rm unclip}(h_{\tilde{\bw},\zeta,\kappa}) =  \frac{9}{2}\underset{\rho\sim \mD}{\mathbb{E}} \left[\underset{\sigma \sim \mD_\rho}{\mathbb{E}} \left[ |g(\rho) - h_{\tilde{\bw},\zeta, \kappa}(\sigma)|^2/2 \right]\right] \leq \epsilon^2/2. \label{eq_sm: main proof 1_1}
\end{align}
Here, $\mH=\mH(\btheta_1)\cup\mH(\btheta_2)\cup\cdots$ is the entire hypothesis class, and $\mH(\btheta_i)$ is the hypothesis class with hyperparameter $\btheta_i\in\mP$.

\vspace{0.5cm}
\noindent
{\bf (ii) Evaluating learning cost:}
We verify that $g_{\rm CA}(\sigma)$ can be represented as a linear function in the feature space with hyperparameters $\btheta=(\zeta, \kappa, B)$ and then evaluate the norm of the weight vector $\tilde{\bw}$.
Recall that the $\Delta$-cluster approximation $g_{\rm CA}(\sigma)=\sum_i \hc_i \prod_{P\in \mP_{i}} \tr \left[ P \sigma \right]$ is written as [see Eq.~\eqref{eq_sm: gcluster LCN}]
\begin{align}
g_{\rm CA}(\sigma)&= \sum_i \hc_i \prod_{j=1}^{a_i} \ell_{ij}(\sigma)
\end{align}
with the local quantity $\ell_{ij}(\sigma)=\prod_{P\in \mP_{i,j}} \tr \left[ P \sigma \right]$ on the subsystem $A_{i,j} \in \AGL(m\Delta)$, where $\supp(P)\subseteq A_{i,j}$ for all $P\in\mP_{i,j}$.
Since $\zeta\geq m\Delta$, the support of $\ell_{ij}(\sigma)$ is also encompassed by $\exists A_{i,j}\in\AGL(\zeta)$.
Furthermore, each term of $g_{\rm CA}(\sigma)$ is the product of at most $\max_i(a_i)=\LCN(g;\Delta,m\Delta)\geq\LCN(g;\Delta,\zeta)$ local quantities $\ell_{ij}(\sigma)$.
Thus, $g_{\rm CA}(\sigma)$ can be represented as a linear function in the feature space of GLQK with $\zeta\geq m\Delta$ and $\kappa\geq\LCN(g;\Delta,\zeta)$ as follows [see Eqs.~\eqref{eq_sm: feature vector of GLQK-TSK} and \eqref{eq_sm: GLQK component}]:
\begin{align}
g_{\rm CA}(\sigma) 
&= \sum_i \hc_i \left[n^{\kappa/2} \prod_{j=1}^{a_i} \left(\frac{b_{ij}!}{\tau^{b_{ij}}}\right)^{1/2}\prod_{P\in \mP_{i,j}}  \left(\frac{2\zeta^D}{\gamma}\right)^{|\supp(P)|/2} \right]  \\
&\hspace{1cm} \times \left[ \frac{1}{n^{\kappa/2}} \prod_{j=1}^{a_i} \left(\frac{\tau^{b_{ij}}}{b_{ij}!}\right)^{1/2}\prod_{P\in \mP_{i,j}}  \left(\frac{\gamma}{2\zeta^D}\right)^{|\supp(P)|/2} \tr[P \sigma ]\right] \\
&\equiv \braket{\tilde{\bw},\phi_{{\rm GL},\zeta,\kappa}(\sigma)},
\end{align}
where $b_{ij}=|\mP_{i,j}|$ is the number of Pauli strings included in $\mP_{i,j}$.
The norm of the weight vector $\tilde{\bw}$ is bounded as 
\begin{align}
\braket{\tilde{\bw},\tilde{\bw}}
&= \sum_i |\hc_i|^2 \left[ n^{\kappa/2} \prod_{j=1}^{a_i} \left(\frac{b_{ij}!}{\tau^{b_{ij}}}\right)^{1/2}\prod_{P\in \mP_{i,j}}  \left(\frac{2\zeta^D}{\gamma}\right)^{|\supp(P)|/2} \right]^2 \\
&= \sum_{i} |\hc_i|^2  \left[ n^\kappa \left(\prod_{j=1}^{a_i} \frac{b_{ij}!}{\tau^{b_{ij}}} \right) \left( 
\prod_{j=1}^{a_i} \prod_{P\in \mP_{i,j}}  \left(\frac{2\zeta^D}{\gamma}\right)^{|\supp(P)|} \right) \right] \\
&\leq \sum_{i} |\hc_i|^2  \left[ n^\kappa \left(\prod_{j=1}^{a_i} \frac{b_{ij}^{b_{ij}}}{\tau^{b_{ij}}} \right) \left(\frac{2\zeta^D}{\gamma}\right)^{\sum_{j=1}^{a_i} \sum_{P\in\mP_{i,j}} |\supp(P)|}  \right] \\
&\leq \sum_{i} |\hc_i|^2  \left[ n^\kappa \left(\prod_{j=1}^{a_i} \frac{(mp)^{b_{ij}}}{\tau^{b_{ij}}} \right) \left(\frac{2\zeta^D}{\gamma}\right)^{mp}  \right] \\
&\leq \sum_{i} |\hc_i|^2  \left[ n^\kappa \left( \frac{mp}{\tau} \right)^{mp}  
\left(\frac{2\zeta^D}{\gamma}\right)^{mp}  \right] \\
&\leq \|g\|_{1}^2 \left(\frac{2mp \zeta^D}{\tau\gamma}\right)^{mp} n^\kappa, 
\end{align}
where we have used $b_{ij}! \leq b_{ij}^{b_{ij}}$ in the second line, $b_{ij}\leq mp$ and $\sum_{j=1}^{a_i} \sum_{P\in\mP_{i,j}} |\supp(P)| \leq mp$ in the third line, $\sum_{j=1}^{a_i} b_{ij}\leq mp$ in the fourth line, and $\sum_i |\hc_i|^2 \leq (\sum_i |\hc_i|)^2=\|g_{\rm CA}\|_1^2\leq \|g\|_1^2$ in the fifth line [see also Eqs.~\eqref{eq_sm: cluster norm} and \eqref{eq_sm: convenient ineq}].
We have also used the assumptions of $2\zeta^D/\gamma \geq 1$ and $mp/\tau \geq 1$ in the third and fourth lines.
Therefore, $\zeta\geq m\Delta$, $\kappa\geq\LCN(g;\Delta,\zeta)$, and $B^2\geq\|g\|_1^2 (2mp\zeta^D/\tau\gamma)^{mp} n^\kappa$ ensure Eq.~\eqref{eq_sm: main proof 1_1}.

We set
\begin{align}
    &N_z \geq \max\left(\frac{400C^2 B^2X^2}{(\epsilon^2/2)^2}, \frac{200 C^4 \log(4/\delta)}{(\epsilon^2/2)^2}, \frac{50(BX+C)^2 \log(8/\delta)}{\epsilon^2/2}\right), \\
    &N_v \geq \frac{800C^2 \log(4|\mP|/\delta)}{3(\epsilon^2/2)},
\end{align}
where $B^2\geq\|g\|_1^2 (2mp\zeta^D/\tau\gamma)^{mp} n^\kappa$ and $X^2=\exp(\kappa\tau \exp(5\gamma))\geq \braket{\phi_{{\rm GL},\zeta,\kappa}(\sigma), \phi_{{\rm GL},\zeta,\kappa}(\sigma)}$.
Then, from Lemma~\ref{lemma: mother_lem}, the following inequality holds with probability at least $1-\delta$:
\begin{align}
L_{\mD_S}^{\rm clip}(h_{ZV}^\ast) 
&\leq \frac{9}{2}\min_{h\in\mH} L_{\mD_S}^{\rm unclip}(h) + \frac{\epsilon^2}{2}. \label{eq_sm: main proof 1_2}
\end{align}
Combining Eqs.~\eqref{eq_sm: main proof 1_1} and \eqref{eq_sm: main proof 1_2}, we obtain
\begin{align}
L_{\mD_S}^{\rm clip}(h_{ZV}^\ast)
&\leq \epsilon^2.
\end{align}
This completes the proof.
\end{proof}

This theorem provides a sufficient condition on the quantum resources for GLQK to achieve high accuracy with high probability.
Suppose that the hyperparameter set $\mP$ contains $(\zeta,\kappa,B)$ such that $\zeta\sim m\Delta$, $\kappa=\LCN(g;\Delta,\zeta)\equiv\alpha_g$, and $B^2 \sim \|g\|_{1}^2 ( 2m p \zeta^D/\tau\gamma )^{mp} n^{\alpha_g}$.
Then, the GLQK can achieve accuracy $\epsilon$ with probability at least $1-\delta$ using $N_z=\mO(n^{\alpha_g}/\epsilon^4 + \log(1/\delta)/\epsilon^4 + n^{\alpha_g}\log(1/\delta)/\epsilon^2)$ training data, $N_v=\mO(\log(|\mP|/\delta)/\epsilon^2)$ validation data, and a shadow size $T=\mO(1/\epsilon^2)$, in terms of $n$, $\epsilon$, and $\delta$.
In many physically relevant situations, where both the correlation length and the degree of the target polynomial are bounded, the number of candidate hyperparameter values $|\mP|$ can be finite independently of $n$.
In this case, the total sample complexity of GLQK is $(N_z+N_v)T=\mO(n^{\alpha_g}/\epsilon^6 + \log(1/\delta)/\epsilon^6 + n^{\alpha_g}\log(1/\delta)/\epsilon^4)$.
Thus, this method is efficient for target functions with small $\alpha_g$.

\subsection{Translationally symmetric cases}

Remarkably, when quantum states exhibit translation symmetry, GLQK requires only a constant number of training data, independent of $n$, to achieve a given accuracy. 
Here, the translation symmetry is defined by $T_\mu \rho T_\mu^\dag = \rho$, where $T_\mu$ denotes the translation operator in the $\mu$th direction ($\mu=1,\cdots,D$).
The following theorem establishes this constant scaling:

\begin{thm}[Theorem~\ref{thm: main theorem 2} in the main text] \label{thm_sm: main_translation}
Consider an $m$-body, degree-$p$ polynomial $g(\rho)$ and a distribution $\mD_S$ over $\mX\times\mY$ such that the sampled quantum state is translationally symmetric and has correlation length at most $\xi$ on a $D$-dimensional hypercubic lattice.
Assume that an upper bound of $|g(\rho)|$ is given: $|g(\rho)|\leq C$ for all $\rho$ drawn from the distribution.
For $\epsilon\in (0,\|g\|_1)$, let $\Delta = \xi \log(\sqrt{18} \|g\|_{1} mp/\epsilon)$, $\zeta \geq m\Delta$, $\kappa\geq 1$, and $B^2 \geq \|g\|_{1}^2 ( 2m p \zeta^D/\tau\gamma )^{mp} n^{\kappa-1}$.
Let $\mP=\{(\zeta_i,\kappa_i,B_i)\}_{i}$ be a finite set of candidate hyperparameter values such that $(\zeta,\kappa,B)\in\mP$, where $\zeta_i$ denotes the subsystem size, $\kappa_i$ the GLQK exponent, and $B_i$ the regularization parameter.
For $\delta\in(0,1)$, suppose that we obtain $N=N_z+N_v$ classical shadows of size $T$, together with their target labels, as training and validation datasets such that
\begin{align}
    &N_z \geq \max\left(\frac{1600C^2 B^2X^2}{\epsilon^4}, \frac{800 C^4 \log(4/\delta)}{\epsilon^4}, \frac{100(BX+C)^2 \log(8/\delta)}{\epsilon^2}\right), \\
    &N_v \geq \frac{1600C^2 \log(4|\mP|/\delta)}{3\epsilon^2}, \\
    &T\geq \frac{5364 \|g\|_1^2 3^m (mp)^2}{\epsilon^2},
\end{align}
where $X^2=\exp(\kappa\tau \exp(5\gamma))$.
Then, kernel ridge regression trained on $N_z$ training data using the polynomial GLQK based on the truncated shadow kernel, combined with hyperparameter optimization via $N_v$ validation data, achieves
\begin{align}
    L_{\mD_S}^{\rm clip}(h_{ZV}^\ast) \leq \epsilon^2
\end{align}
with probability at least $1-\delta$.
Here, we have assumed that $2\zeta^D/\gamma \geq 1$ and $mp/\tau \geq 1$.
\end{thm}    

\begin{proof}
This proof considers the one-dimensional case ($D=1$) for simplicity, but the generalization to arbitrary dimensions is straightforward.

\begin{figure*}[t]
    \centering
    \includegraphics[width=0.8\linewidth]{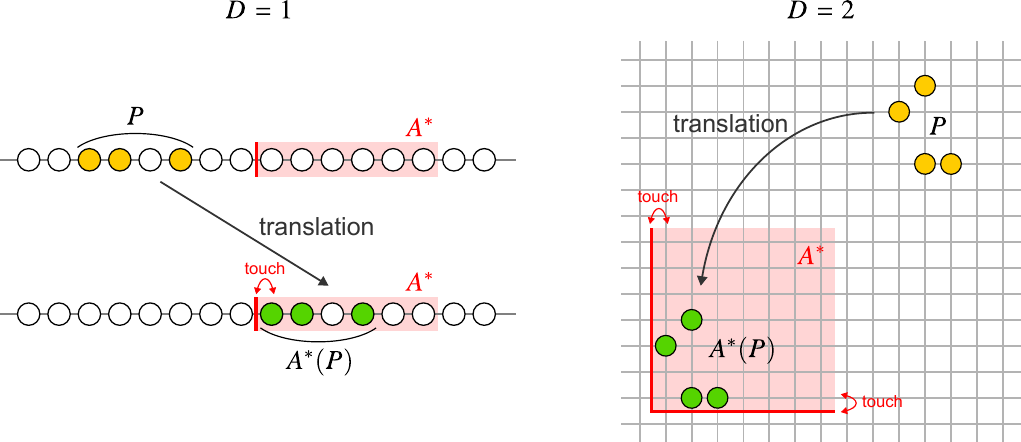}
    \caption{
    Translation of Pauli strings.
    The left and right panels illustrate the $D=1$ and $2$ cases, respectively.
    Given a Pauli string $P$ and a local subsystem $A^\ast$, we translate $P$ to $A^\ast(P)$ such that $\supp(A^\ast(P))\subseteq A^\ast$ and $\supp(A^\ast(P))$ touches the left side of $A^\ast$ (the left and bottom sides of $A^\ast$) for $D=1$ ($D=2$).
    }
    \label{fig_sm: proof1}
\end{figure*}

\vspace{0.5cm}
\noindent
{\bf (i) Deriving an easy-to-learn polynomial:}
We derive an easy-to-learn polynomial equivalent to $g_{\rm CA}(\rho)$ by ``diluting'' it in space with translation symmetry.
The derived polynomial has an $\ell_2$-norm that is $1/n$ times smaller than the original one, resulting in the constant scaling in $n$.

Let $g_{\rm CA}(\rho)=\sum_i \hc_i \prod_{P\in\mP_i} \tr[P\rho]$ be the $\Delta$-cluster approximation of $g(\rho)$ and $A^\ast \in \AGL(m\Delta)$ be a representative element.
As discussed in Eq.~\eqref{eq_sm: set of subsystem}, the support of any Pauli string $P\in\mP_i$ is encompassed by a corresponding subsystem $A\in \AGL(m\Delta)$.
Here, we consider the translation of a Pauli string $P\in\mP_i$ into $A^\ast$.
More specifically, we define the translated Pauli string $A^\ast(P)\equiv T^t P (T^\dag)^t \in\{I,X,Y,Z\}^{\otimes n}$ with some $t$ such that (i) $\supp(A^\ast(P))\subseteq A^\ast$ and (ii) $\supp(A^\ast(P))$ touches the left side of $A^\ast$.
These two conditions define $A^\ast(P)$ uniquely (see Fig.~\ref{fig_sm: proof1}).
Then, we introduce a polynomial
\begin{align}
\tilde{g}_{\rm CA}(\rho) 
= \sum_i \tc_i \prod_{\tilde{P}\in \tilde{\mP}_i} \tr \left[\tilde{P} \rho \right],
\end{align}
where $\tilde{\mP}_i=\{A^\ast(P)| P\in\mP_i \}$.
Here, we have defined new coefficients $\tc_i$ by combining duplicated terms if $\tilde{\mP}_i=\tilde{\mP}_j$ for some $i$ and $j$ [this procedure is the same as that in deriving Eq.~\eqref{eq_sm: combined cluster approximation}].
Note that combining the duplicated terms does not increase the $\ell_1$-norm: $\|g_{\rm CA}\|_1 = \sum_i |\hc_i| \geq \sum_i |\tc_i|=\|\tilde{g}_{\rm CA}\|_1$.
When $\rho$ is translationally symmetric, $g_{\rm CA}(\rho)=\tilde{g}_{\rm CA}(\rho)$ holds because $\tr[P\rho]=\tr[A^\ast(P)\rho]$.

Subsequently, we translate all Pauli strings in $\tilde{\mP}_i$ by $t$ sites, defining $\tilde{\mP}_{i,t}=\{T^t P (T^\dag)^t | P\in\tilde{\mP}_i \}$.
Then, we introduce the following polynomial:
\begin{align}
\bar{g}_{\rm CA}(\rho) 
&= \sum_i \sum_{t=1}^n \frac{\tc_i}{n} \prod_{\tilde{P}\in\tilde{\mP}_{i,t}} \tr \left[ \tilde{P} \rho \right]. 
\end{align}
Using $\prod_{\tilde{P}\in\tilde{\mP}_{i,t}} \tr [ \tilde{P} \rho ] = \prod_{\tilde{P}\in\tilde{\mP}_i} \tr [ \tilde{P} \rho ]$, we can easily show that $\tilde{g}_{\rm CA}(\rho)=\bar{g}_{\rm CA}(\rho)$ for translationally symmetric $\rho$.
By combining the indices $i$ and $t$ into a single index $i'$ and introducing $\barc_{i'}=\tc_{i}/n$ and $\bar{\mP}_{i'}=\tilde{\mP}_{i,t}$, we have
\begin{align}
\bar{g}_{\rm CA}(\rho) 
&= \sum_{i'} \barc_{i'} \prod_{\bar{P}\in\bar{\mP}_{i'}} \tr \left[ \bar{P} \rho \right].
\end{align}
Here, the coefficients satisfy $\sum_i |\tc_i|= \sum_i \sum_{t=1}^n |\tc_i/n| =\sum_{i'} |\barc_{i'}|$.
By construction, for any $\bar{\mP}_{i'}$, there exists $A_{i'}\in\AGL(m\Delta)$ such that $\supp(\bar{P})\subseteq A_{i'}$ for all $\bar{P}\in\bar{\mP}_{i'}$.
Therefore, $\bar{g}_{\rm CA}$ is the sum of local quantities $\prod_{\bar{P}\in\bar{\mP}_{i'}} \tr \left[ \bar{P} \rho \right]$.
The $\ell_1$- and $\ell_2$-norms of $\bar{g}_{\rm CA}$ satisfy
\begin{align}
&\|\bar{g}_{\rm CA}\|_1=\sum_{i'} |\barc_{i'}| = \sum_i |\tc_i| \leq \sum_i |\hc_i| = \|g_{\rm CA}\|_1 \leq \|g\|_1, \label{eq_sm: translation norm1}\\
&\|\bar{g}_{\rm CA}\|_2^2 = \sum_{i'} |\barc_{i'}|^2 = \sum_i \sum_{t=1}^n |\tc_i/n|^2 = \frac{1}{n}\sum_i |\tc_i|^2  \\
&\hspace{1.6cm}\leq  \frac{1}{n} \left( \sum_i |\tc_i|\right)^2 \leq \frac{1}{n} \left( \sum_i |\hc_i|\right)^2 = \frac{1}{n}\|g_{\rm CA}\|_1^2 \leq \frac{1}{n}\|g\|_1^2, \label{eq_sm: translation norm2}
\end{align}
where we have used $\|g_{\rm CA}\|_1 \leq \|g\|_1$.

\vspace{0.5cm}
\noindent
{\bf (ii) Error in estimating the polynomial from classical shadows:}
For $\bar{g}_{\rm CA}$, we perform the same analysis as the proof of Theorem~\ref{thm_sm: main_nontranslation}.
Let $\Delta=\xi \log(\sqrt{18}\|g\|_1 mp/\epsilon)$ and $T=T_{\rm CA}(g; \epsilon/\sqrt{18})\geq T(g_{\rm CA}; \epsilon/\sqrt{18}) \geq T(\bar{g}_{\rm CA}; \epsilon/\sqrt{18})$, where we have used Eqs.~\eqref{eq_sm: T_cluster} and \eqref{eq_sm: translation norm1}.
Then, by Lemmas~\ref{lem_sm: approx} and \ref{lemma: mother_lem_shadow}, the polynomial $\bar{g}_{\rm CA}$ and the classical shadow $\sigma$ obey
\begin{align}
    &|g(\rho) - \bar{g}_{\rm CA}(\rho)| \leq \epsilon/\sqrt{18}, \\
    &\underset{\sigma\sim\mD_\rho}{\mathbb{E}} \left[ |\bar{g}_{\rm CA}(\rho) - \bar{g}_{\rm CA}(\sigma)|^2 \right] \leq \epsilon^2/18
\end{align}
for any translationally symmetric $\rho$ with a correlation length less than or equal to $\xi$, where we have used $g_{\rm CA}(\rho)=\bar{g}_{\rm CA}(\rho)$.
In the same way as the proof of Theorem~\ref{thm_sm: main_nontranslation}, these two inequalities lead to
\begin{align}
    &\underset{\rho \sim \mD}{\mathbb{E}} \left[\underset{\sigma \sim \mD_{\rho}}{\mathbb{E}}\left[ |g(\rho) - \bar{g}_{\rm CA}(\sigma)|^2/2 \right]\right] \leq \epsilon^2/9.
\end{align}
If $\bar{g}_{\rm CA}(\sigma)$ can be represented as a linear function in the feature space of GLQK with hyperparameters $\btheta=(\zeta, \kappa, B)$, i.e., $\bar{g}_{\rm CA}(\sigma)=\braket{\tilde{\bw}, \phi_{{\rm GL}, \zeta,\kappa}(\sigma)}\equiv h_{\tilde{\bw},\zeta,\kappa}(\sigma)$ for some $\tilde{\bw}$ such that $\|\tilde{\bw}\|\leq B$ ($\phi_{{\rm GL},\zeta,\kappa}$ is the feature vector of GLQK with a subsystem size $\zeta$ and an exponent $\kappa$), the first term on the right-hand side in Eq.~\eqref{eq_sm: mother_ineq} is upper bounded by $(9/2)\times(\epsilon^2/9)=\epsilon^2/2$:
\begin{align}
\frac{9}{2}\min_{h\in\mH} L_{\mD_S}^{\rm unclip}(h) \leq \frac{9}{2}L_{\mD_S}^{\rm unclip}(h_{\tilde{\bw},\zeta,\kappa}) =  \frac{9}{2}\underset{\rho\sim \mD}{\mathbb{E}} \left[\underset{\sigma \sim \mD_\rho}{\mathbb{E}} \left[ |g(\rho) - h_{\tilde{\bw},\zeta,\kappa}(\sigma)|^2/2 \right]\right] \leq \epsilon^2/2. \label{eq_sm: main proof 2_1}
\end{align}
Here, $\mH=\mH(\btheta_1)\cup\mH(\btheta_2)\cup\cdots$ is the entire hypothesis class, and $\mH(\btheta_i)$ is the hypothesis class with hyperparameter $\btheta_i\in\mP$.

\vspace{0.5cm}
\noindent
{\bf (iii) Evaluating learning cost:}
We verify that $\bar{g}_{\rm CA}(\sigma)$ can be represented as a linear function in the feature space  with hyperparameters $\btheta=(\zeta, \kappa, B)$ and then evaluate the norm of the weight vector $\tilde{\bw}$.
As discussed above, each term of $\bar{g}_{\rm CA}$ is a local quantity $\prod_{\bar{P}\in\bar{\mP}_{i'}} \tr \left[ \bar{P} \sigma \right]$, and its support is encompassed by $\exists A_{i'}\in\AGL(m\Delta)$.
Since $\zeta\geq m\Delta$, it is also encompassed by $\exists A_{i'}\in\AGL(\zeta)$.
Therefore, $\bar{g}_{\rm CA}(\sigma)$ can be represented as a linear function in the feature space of GLQK with $\zeta\geq m\Delta$ and $\kappa\geq1$ as follows [see Eqs.~\eqref{eq_sm: feature vector of GLQK-TSK} and \eqref{eq_sm: GLQK component}]:
\begin{align}
\bar{g}_{\rm CA}(\sigma) 
&= \sum_{i'} \barc_{i'} \left[n^{\kappa/2} \left(\frac{b_{i'}!}{\tau^{b_{i'}}}\right)^{1/2}\prod_{\bar{P}\in \bar{\mP}_{i'}}  \left(\frac{2\zeta^D}{\gamma}\right)^{|\supp(\bar{P})|/2} \right]  \\
&\hspace{1cm} \times \left[ \frac{1}{n^{\kappa/2}} \left(\frac{\tau^{b_{i'}}}{b_{i'}!}\right)^{1/2}\prod_{\bar{P}\in \bar{\mP}_{i'}}  \left(\frac{\gamma}{2\zeta^D}\right)^{|\supp(\bar{P})|/2} \tr[\bar{P} \sigma ]\right] \\
&\equiv \braket{\tilde{\bw},\phi_{{\rm GL},\zeta,\kappa}(\sigma)},
\end{align}
where $b_{i'}=|\bar{\mP}_{i'}|$ is the number of Pauli strings included in $\bar{\mP}_{i'}$.
The norm of the weight vector $\tilde{\bw}$ is bounded as 
\begin{align}
\braket{\tilde{\bw},\tilde{\bw}}
&= \sum_{i'} |\barc_{i'}|^2 \left[ n^{\kappa/2} \left(\frac{b_{i'}!}{\tau^{b_{i'}}}\right)^{1/2}\prod_{\bar{P}\in \bar{\mP}_{i'}}  \left(\frac{2\zeta^D}{\gamma}\right)^{|\supp(\bar{P})|/2} \right]^2 \\
&= \sum_{i'} |\barc_{i'}|^2  \left[ n^\kappa \left(\frac{b_{i'}!}{\tau^{b_{i'}}} \right)  
\prod_{\bar{P}\in \bar{\mP}_{i'}}  \left(\frac{2\zeta^D}{\gamma}\right)^{|\supp(\bar{P})|}  \right] \\
&\leq \sum_{i'} |\barc_{i'}|^2  \left[ n^\kappa \left(\frac{b_{i'}^{b_{i'}}}{\tau^{b_{i'}}} \right) \left(\frac{2\zeta^D}{\gamma}\right)^{\sum_{\bar{P}\in\bar{\mP}_{i'}} |\supp(\bar{P})|}  \right] \\
&\leq \sum_{i'} |\barc_{i'}|^2  \left[ n^\kappa \left(\frac{(mp)^{b_{i'}}}{\tau^{b_{i'}}} \right) \left(\frac{2\zeta^D}{\gamma}\right)^{mp}  \right] \\
&\leq \sum_{i'} |\barc_{i'}|^2  \left[ n^\kappa \left( \frac{mp}{\tau} \right)^{mp}  
\left(\frac{2\zeta^D}{\gamma}\right)^{mp}  \right] \\
&\leq \|g\|_1^2 \left(\frac{2mp \zeta^D}{\tau\gamma}\right)^{mp}  n^{\kappa-1},
\end{align}
where we have used $b_{i'}! \leq b_{i'}^{b_{i'}}$ in the second line, $b_{i'}\leq mp$ and $\sum_{\bar{P}\in\bar{\mP}_{i'}} |\supp(\bar{P})| \leq mp$ in the third line,  $b_{i'}\leq mp$ in the fourth line, and Eq.~\eqref{eq_sm: translation norm2} in the fifth line [see also Eqs.~\eqref{eq_sm: cluster norm} and \eqref{eq_sm: convenient ineq}].
We have also used the assumptions of $2\zeta^D/\gamma \geq 1$ and $mp/\tau \geq 1$ in the third and fourth lines.
Therefore, $B^2\geq\|g\|_1^2 (2mp\zeta^D/\tau\gamma)^{mp}n^{\kappa-1}$ ensures Eq.~\eqref{eq_sm: main proof 2_1}.

We set
\begin{align}
    &N_z \geq \max\left(\frac{400C^2 B^2X^2}{(\epsilon^2/2)^2}, \frac{200 C^4 \log(4/\delta)}{(\epsilon^2/2)^2}, \frac{50(BX+C)^2 \log(8/\delta)}{\epsilon^2/2}\right), \\
    &N_v \geq \frac{800C^2 \log(4|\mP|/\delta)}{3(\epsilon^2/2)},
\end{align}
where $B^2\geq\|g\|_1^2 (2mp\zeta^D/\tau\gamma)^{mp} n^{\kappa-1}$ and $X^2=\exp(\kappa\tau \exp(5\gamma))\geq \braket{\phi_{{\rm GL},\zeta,\kappa}(\sigma), \phi_{{\rm GL},\zeta,\kappa}(\sigma)}$.
Then, from Lemma~\ref{lemma: mother_lem}, the following inequality holds with probability at least $1-\delta$:
\begin{align}
L_{\mD_S}^{\rm clip}(h_{ZV}^\ast) 
&\leq \frac{9}{2} \min_{h\in\mH} L_{\mD_S}^{\rm unclip}(h) + \frac{\epsilon^2}{2}. \label{eq_sm: main proof 2_2}
\end{align}
Combining Eqs.~\eqref{eq_sm: main proof 2_1} and \eqref{eq_sm: main proof 2_2}, we obtain
\begin{align}
L_{\mD_S}^{\rm clip}(h_{ZV}^\ast)
&\leq \epsilon^2.
\end{align}
This completes the proof.
\end{proof}

This theorem shows that for translationally symmetric quantum states, the GLQK can achieve constant sample complexity independent of the system size $n$.
Specifically, suppose that the hyperparameter set $\mP$ contains $(\zeta,\kappa,B)$ such that $\zeta\sim m\Delta$, $\kappa=1$, and $B^2 \sim \|g\|_{1}^2 ( 2m p \zeta^D/\tau\gamma )^{mp}$.
Then, the GLQK can learn the target function with error $\epsilon$ and probability $1-\delta$ using $N_z=\mO(\log(1/\delta)/\epsilon^4)$ training data, $N_v=\mO(\log(|\mP|/\delta)/\epsilon^2)$ validation data, and shadow size $T=\mO(1/\epsilon^2)$, with respect to $n$, $\epsilon$, and $\delta$.
If $|\mP|$ is independent of $n,\epsilon,\delta$, the total sample complexity is $(N_z+N_v)T=\mO(\log(1/\delta)/\epsilon^6)$, confirming constant sample complexity in $n$.

\section{Rigorous guarantee for shadow kernel}\label{sec_sm: SK}

This section evaluates the amount of quantum resources sufficient for the shadow kernel to ensure certain accuracy for both general data and translationally symmetric data.
The problem setting is the same as that in the GLQK.

\subsection{General cases}

\begin{thm} \label{thm_sm: SK_nontranslation}
Consider an $m$-body, degree-$p$ polynomial $g(\rho)$ and a distribution $\mD_S$ over $\mX\times\mY$ such that the sampled quantum state has correlation length at most $\xi$ on a $D$-dimensional hypercubic lattice.
Assume that an upper bound of $|g(\rho)|$ is given: $|g(\rho)|\leq C$ for all $\rho$ drawn from the distribution.
Let $\mP=\{B_i\}_{i}$ be a finite set of candidate hyperparameter values such that $B\in\mP$ and $B^2\geq \|g\|_1^2 (2m^2p^2/\tau\gamma)^{mp} n^{mp}$, where $B_i$ denotes the regularization parameter.
For $\epsilon\in (0,\|g\|_1)$ and $\delta\in(0,1)$, suppose that we obtain $N=N_z+N_v$ classical shadows of size $T$, together with their target labels, as training and validation datasets such that
\begin{align}
    &N_z \geq \max\left(\frac{1600C^2 B^2X^2}{\epsilon^4}, \frac{800 C^4 \log(4/\delta)}{\epsilon^4}, \frac{100(BX+C)^2 \log(8/\delta)}{\epsilon^2}\right), \\
    &N_v \geq \frac{1600C^2 \log(4|\mP|/\delta)}{3\epsilon^2}, \\
    &T\geq \frac{5364 \|g\|_1^2 3^m (mp)^2}{\epsilon^2},
\end{align}
where $X^2=\exp(\tau \exp(5\gamma))$.
Then, kernel ridge regression trained on $N_z$ training data using the shadow kernel, combined with hyperparameter optimization via $N_v$ validation data, achieves
\begin{align}
    L_{\mD_S}^{\rm clip}(h_{ZV}^\ast) \leq \epsilon^2
\end{align}
with probability at least $1-\delta$.
Here, we have assumed that $2n/\gamma \geq 1$ and $mp/\tau \geq 1$.
\end{thm}    

\begin{proof}
One can show that the learning cost scaling in $n$ and $\epsilon$ is independent of whether learning the original polynomial $g(\rho)$ or its cluster approximation $g_{\rm CA}(\rho)$.
To maintain consistency with the GLQK, this proof focuses on learning the cluster approximation $g_{\rm CA}(\rho)=\sum_i \hc_i \prod_{P\in\mP_i}\tr[P\rho]$.

\vspace{0.5cm}
\noindent
{\bf (i) Error in estimating the polynomial from classical shadows:}
As discussed in the proof of Theorem~\ref{thm_sm: main_nontranslation}, by setting $\Delta=\xi \log(\sqrt{18}\|g\|_1 mp/\epsilon)$ and $T=T_{\rm CA}(g;\epsilon/\sqrt{18})\geq T(g_{\rm CA};\epsilon/\sqrt{18})$, the value of the $\Delta$-cluster approximation estimated from a classical shadow $\sigma$ satisfies
\begin{align}
    &\underset{\rho\sim \mD}{\mathbb{E}} \left[\underset{\sigma \sim \mD_\rho}{\mathbb{E}} \left[ |g(\rho) - g_{\rm CA}(\sigma)|^2/2 \right]\right] \leq \epsilon^2/2
\end{align}
for any $\rho$ with a correlation length less than or equal to $\xi$.
If $g_{\rm CA}(\sigma)$ can be represented as a linear function in the feature space of shadow kernel with a hyperparameter $B\in\mP$, i.e., $g_{\rm CA}(\sigma)=\braket{\tilde{\bw}, \phi_{\rm SK}(\sigma)}\equiv h_{\tilde{\bw}}(\sigma)$ for some $\tilde{\bw}$ such that $\|\tilde{\bw}\|\leq B$, the first term on the right-hand side in Eq.~\eqref{eq_sm: mother_ineq} is upper bounded by $(9/2)\times(\epsilon^2/9)=\epsilon^2/2$:
\begin{align}
\frac{9}{2}\min_{h\in\mH} L_{\mD_S}^{\rm unclip}(h) \leq \frac{9}{2}L_{\mD_S}^{\rm unclip}(h_{\tilde{\bw}}) =  \frac{9}{2}\underset{\rho\sim \mD}{\mathbb{E}} \left[\underset{\sigma \sim \mD_\rho}{\mathbb{E}} \left[ |g(\rho) - h_{\tilde{\bw}}(\sigma)|^2/2 \right]\right] \leq \epsilon^2/2. \label{eq_sm: main proof 3_1}
\end{align}
Here, $\mH=\mH(B_1)\cup\mH(B_2)\cup\cdots$ is the entire hypothesis class, and $\mH(B_i)$ is the hypothesis class with hyperparameter $B_i\in\mP$.

\vspace{0.5cm}
\noindent
{\bf (ii) Evaluating learning cost:}
We verify that $g_{\rm CA}(\sigma)$ can be represented as a linear function in the feature space and then evaluate the norm of the weight vector $\tilde{\bw}$.
Given feature vector components of Eq.~\eqref{eq_sm: SK component}, $g_{\rm CA}(\sigma)$ can be written as
\begin{align}
g_{\rm CA}(\sigma) 
&= \sum_i \hc_i  \left[ \sqrt{\frac{b_i!}{\tau^{b_i}}}\prod_{P\in \mP_i}  \sqrt{|\supp(P)|!\left(\frac{2n}{\gamma}\right)^{|\supp(P)|}} \right]  \\
&\hspace{1cm} \times \left[ \sqrt{\frac{\tau^{b_i}}{b_i!}}\prod_{P\in \mP_i}  \sqrt{\frac{1}{|\supp(P)|!}\left(\frac{\gamma}{2n}\right)^{|\supp(P)|}} \tr[P \sigma ]\right] \\
&\equiv \braket{\tilde{\bw},\phi_{\rm SK}(\sigma))},
\end{align}
indicating that it can be described as a linear function in the feature space of the shadow kernel.
Here, $b_i=|\mP_i|$ is the number of Pauli strings included in $\mP_i$.
The norm of the weight vector $\tilde{\bw}$ is bounded as 
\begin{align}
\braket{\tilde{\bw},\tilde{\bw}}
&= \sum_i |\hc_i|^2\left[ \sqrt{\frac{b_i!}{\tau^{b_i}}}\prod_{P\in \mP_i}  \sqrt{|\supp(P)|!\left(\frac{2n}{\gamma}\right)^{|\supp(P)|}} \right]^2 \\
&= \sum_{i} |\hc_i|^2  \left[ \frac{b_i!}{\tau^{b_i}}\prod_{P\in \mP_i}  |\supp(P)|!\left(\frac{2n}{\gamma}\right)^{|\supp(P)|} \right] \\
&\leq \sum_{i} |\hc_i|^2  \left[ \frac{b_i^{b_i}}{\tau^{b_i}} (mp)! \left(\frac{2n}{\gamma}\right)^{\sum_{P\in\mP_i}|\supp(P)|} \right] \\
&\leq \sum_{i} |\hc_i|^2  \left[\frac{(mp)^{b_i}}{\tau^{b_i}} (mp)^{mp} \left(\frac{2n}{\gamma}\right)^{mp}  \right] \\
&\leq \sum_{i} |\hc_i|^2  \left[ \left( \frac{mp}{\tau} \right)^{mp}  (mp)^{mp}
\left(\frac{2n}{\gamma}\right)^{mp}  \right] \\
&\leq n^{mp} \left(\frac{2m^2p^2}{\tau\gamma}\right)^{mp} \|g\|_{1}^2,
\end{align}
where we have used $b_i! \leq b_i^{b_i}$ and $\prod_{P\in\mP_i}|\supp(P)|!\leq (mp)!$ in the second line, $b_i\leq mp$, $(mp)!\leq (mp)^{mp}$, and $\sum_{P\in\mP_{i}} |\supp(P)| \leq mp$ in the third line, $b_i\leq mp$ in the fourth line, and $\sum_i |\hc_i|^2 \leq (\sum_i |\hc_i|)^2 = \|g_{\rm CA}\|_1^2 \leq \|g\|_1^2$ in the fifth line [see also Eqs.~\eqref{eq_sm: cluster norm} and \eqref{eq_sm: convenient ineq}].
We have also used the assumptions of $2n/\gamma \geq 1$ and $mp/\tau \geq 1$ in the third and fourth lines.
Therefore, $B^2\geq\|g\|_{1}^2 ( 2m^2 p^2/\tau\gamma )^{mp} n^{mp}$ ensures Eq.~\eqref{eq_sm: main proof 3_1}.

We set
\begin{align}
    &N_z \geq \max\left(\frac{400C^2 B^2X^2}{(\epsilon^2/2)^2}, \frac{200 C^4 \log(4/\delta)}{(\epsilon^2/2)^2}, \frac{50(BX+C)^2 \log(8/\delta)}{\epsilon^2/2}\right), \\
    &N_v \geq \frac{800C^2 \log(4|\mP|/\delta)}{3(\epsilon^2/2)},
\end{align}
where $B^2\geq\|g\|_{1}^2 ( 2m^2 p^2/\tau\gamma )^{mp} n^{mp}$ and $X^2=\exp(\tau \exp(5\gamma))\geq \braket{\phi_{{\rm SK}}(\sigma), \phi_{{\rm SK}}(\sigma)}$.
Then, from Lemma~\ref{lemma: mother_lem}, the following inequality holds with probability at least $1-\delta$:
\begin{align}
L_{\mD_S}^{\rm clip}(h_{ZV}^\ast) 
&\leq \frac{9}{2} \min_{h\in\mH} L_{\mD_S}^{\rm unclip}(h) + \frac{\epsilon^2}{2}. \label{eq_sm: main proof 3_2}
\end{align}
Combining Eqs.~\eqref{eq_sm: main proof 3_1} and \eqref{eq_sm: main proof 3_2}, we obtain
\begin{align}
L_{\mD_S}^{\rm clip}(h_{ZV}^\ast)
&\leq \epsilon^2.
\end{align}
This completes the proof.
\end{proof}

This theorem shows that if the hyperparameter set $\mP$ contains $B^2\sim \|g\|_1^2 (2m^2p^2/\tau\gamma)^{mp} n^{mp}$, the shadow kernel can achieve accuracy $\epsilon$ from $N_z=\mO(n^{mp}/\epsilon^4 + \log(1/\delta)/\epsilon^4 + n^{mp}\log(1/\delta)/\epsilon^2)$ training data, $N_v=\mO(\log(|\mP|/\delta)/\epsilon^2)$ validation data, and a shadow size $T=\mO(1/\epsilon^2)$.
If $|\mP|$ is independent of $n,\epsilon,\delta$, the total sample complexity is $(N_z+N_v)T=\mO(n^{mp}/\epsilon^6 + \log(1/\delta)/\epsilon^6 + n^{mp}\log(1/\delta)/\epsilon^4)$.
This result contrasts sharply with the GLQK, where the dependence on $n$ in the sample complexity is $n^{\alpha_g}$.
Since $\alpha_g \leq mp$, this indicates that the GLQK can learn from fewer quantum resources.

\subsection{Translationally symmetric cases}

Imposing translation symmetry on quantum data improves the sample complexity, similarly to the GLQK.
\begin{thm} \label{thm_sm: SK_translation}
Consider an $m$-body, degree-$p$ polynomial $g(\rho)$ and a distribution $\mD_S$ over $\mX\times\mY$ such that the sampled quantum state is translationally symmetric and has correlation length at most $\xi$ on a $D$-dimensional hypercubic lattice.
Assume that an upper bound of $|g(\rho)|$ is given: $|g(\rho)|\leq C$ for all $\rho$ drawn from the distribution.
For $\epsilon\in (0,\|g\|_1)$, let $\Delta = \xi \log(\sqrt{18}\|g\|_{1} mp/\epsilon)$, $\beta_g=\LFC(g; \Delta)$, and $B^2\geq \|g\|_1^2 (2m^2p^2/\tau\gamma)^{mp} n^{mp-\beta_g}$.
Let $\mP=\{B_i\}_{i}$ be a finite set of candidate hyperparameter values such that $B\in\mP$, where $B_i$ denotes the regularization parameter.
For $\delta\in(0,1)$, suppose that we obtain $N=N_z+N_v$ classical shadows of size $T$, together with their target labels, as training and validation datasets such that
\begin{align}
    &N_z \geq \max\left(\frac{1600C^2 B^2X^2}{\epsilon^4}, \frac{800 C^4 \log(4/\delta)}{\epsilon^4}, \frac{100(BX+C)^2 \log(8/\delta)}{\epsilon^2}\right), \\
    &N_v \geq \frac{1600C^2 \log(4|\mP|/\delta)}{3\epsilon^2}, \\
    &T\geq \frac{5364 \|g\|_1^2 3^m (mp)^2}{\epsilon^2},
\end{align}
where $X^2=\exp(\tau \exp(5\gamma))$.
Then, kernel ridge regression trained on $N_z$ training data using the shadow kernel, combined with hyperparameter optimization via $N_v$ validation data, achieves
\begin{align}
    L_{\mD_S}^{\rm clip}(h_{ZV}^\ast) \leq \epsilon^2
\end{align}
with probability at least $1-\delta$.
Here, we have assumed that $2/\gamma \geq 1$ and $mp/\tau \geq 1$.
\end{thm}    

\begin{proof}
This proof considers the one-dimensional case ($D=1$) for simplicity, but the generalization to arbitrary dimensions is straightforward.

First, for $g_{\rm CA}(\rho)=\sum_i \hc_i \prod_{P\in\mP_i} \tr[P\rho]$, we show that
\begin{align}
    f_i\equiv\sum_{P\in\mP_{i}}|\supp(P)| - b_i \leq mp - \beta_g, \label{eq_sm: fi}
\end{align}
where $b_i=|\mP_i|$ is the number of Pauli strings included in $\mP_i$, and $\beta_g=\max(p,\min_j(b_j))$.
This can be confirmed by proving (1) $f_i \leq mp-p$ and (2) $f_i \leq mp-\min_j(b_j)$ for all $i$:
\begin{enumerate}
\item[(1)] If $b_i\leq p$, $f_i\leq mb_i - b_i\leq mp-p$ holds, where we have used $|\supp(P)|\leq m$.
Conversely, even if $b_i>p$, $f_i\leq mp - b_i\leq mp-p$ holds, where we have used $\sum_{P\in\mP_{i}}|\supp(P)|\leq mp$.
Thus, we have $f_i\leq mp-p$.
\item[(2)] We have $f_i\leq mp-b_i \leq mp - \min_j(b_j)$, where we have used $\sum_{P\in\mP_{i}}|\supp(P)|\leq mp$.
\end{enumerate}
Therefore, we obtain $f_i \leq mp - \beta_g$.

\vspace{0.5cm}
\noindent
{\bf (i) Deriving an easy-to-learn polynomial:}
We derive an easy-to-learn polynomial equivalent to $g_{\rm CA}(\rho)$ by ``diluting'' it with translation symmetry.
Similarly to the proof of Theorem~\ref{thm_sm: main_translation}, we choose a representative element $A^\ast \in \AGL(\zeta)$, where $\zeta=m\Delta$. 
(Although $\AGL(\zeta)$ is unnecessary for calculating the shadow kernel, we technically use it here to tightly evaluate the norm of $\bar{g}_{\rm CA}$ introduced below.)
By translating Pauli strings into $A^\ast$, we define
\begin{align}
\tilde{g}_{\rm CA}(\rho) 
= \sum_i \tc_i \prod_{\tilde{P}\in \tilde{\mP}_i} \tr \left[\tilde{P} \rho \right],
\end{align}
where $\tilde{\mP}_i=\{A^\ast(P)| P\in\mP_i \}$.
Here, we have introduced new coefficients $\tc_i$ by combining duplicated terms if $\tilde{\mP}_i=\tilde{\mP}_j$ for some $i$ and $j$.
When $\rho$ exhibits translation symmetry, $g_{\rm CA}(\rho)=\tilde{g}_{\rm CA}(\rho)$ holds.

For $\tilde{\mP}_i=\{\tilde{P}_{i,1}, \cdots, \tilde{P}_{i,b_i}\}$, let $\tilde{\mP}_{i,\bt}=\{ T^{t_k} \tilde{P}_{i,k} (T^\dag)^{t_k} | k=1,\cdots,b_i\}$ with $\bt=(t_1,\cdots,t_{b_i})$.
Then, we define the following polynomial:
\begin{align}
\bar{g}_{\rm CA}(\rho) 
= \sum_i \sum_{t_1=1}^n \cdots \sum_{t_{b_i}=1}^n \frac{\tc_i}{n^{b_i}} \prod_{\tilde{P}\in \tilde{\mP}_{i,\bt}} \tr \left[ \tilde{P} \rho \right].
\end{align}
By combining the indices $i$ and $t_1,\cdots,t_{b_i}$ into a single index $i'$, we have
\begin{align}
\bar{g}_{\rm CA}(\rho) 
= \sum_{i'} \barc_{i'} \prod_{\bar{P}\in \bar{\mP}_{i'}} \tr \left[ \bar{P} \rho \right],
\end{align}
where $\barc_{i'}=\tc_i/n^{b_i}$ and $\bar{\mP}_{i'}=\tilde{\mP}_{i,\bt}$.
For translationally symmetric $\rho$, we can show that $\tilde{g}_{\rm CA}(\rho)=\bar{g}_{\rm CA}(\rho)$.
Note that $\|g_{\rm CA}\|_1 = \sum_i |\hc_i| \geq  \sum_i |\tc_i| = \sum_{i'} |\barc_{i'}| = \|\bar{g}_{\rm CA}\|_1$.

\vspace{0.5cm}
\noindent
{\bf (ii) Error in estimating the polynomial from classical shadows:}
For $\bar{g}_{\rm CA}(\rho)$, we perform the same analysis as that in the proof of Theorem~\ref{thm_sm: main_nontranslation}.
Let $\Delta=\xi \log(\sqrt{18}\|g\|_1 mp/\epsilon)$ and $T=T_{\rm CA}(g; \epsilon/\sqrt{18}) \geq T(g_{\rm CA}; \epsilon/\sqrt{18}) \geq T(\bar{g}_{\rm CA}; \epsilon/\sqrt{18})$, where we have used Eq.~\eqref{eq_sm: T_cluster} and $\|g_{\rm CA}\|_1 \geq \|\bar{g}_{\rm CA}\|_1$.
Then, by Lemmas~\ref{lem_sm: approx} and \ref{lemma: mother_lem_shadow}, the value of the polynomial $\bar{g}_{\rm CA}(\rho)$ estimated from a classical shadow $\sigma$ obeys
\begin{align}
    &\underset{\rho \sim \mD}{\mathbb{E}} \left[\underset{\sigma \sim \mD_{\rho}}{\mathbb{E}}\left[ |g(\rho) - \bar{g}_{\rm CA}(\sigma)|^2/2 \right]\right] \leq \epsilon^2/2
\end{align}
for any translationally symmetric $\rho$ with a correlation length less than or equal to $\xi$.
If $\bar{g}_{\rm CA}(\sigma)$ can be represented as a linear function in the feature space of shadow kernel with a hyperparameter $B\in\mP$, i.e., $\bar{g}_{\rm CA}(\sigma)=\braket{\tilde{\bw}, \phi_{\rm SK}(\sigma)}\equiv h_{\tilde{\bw}}(\sigma)$ for some $\tilde{\bw}$ such that $\|\tilde{\bw}\|\leq B$, the first term on the right-hand side in Eq.~\eqref{eq_sm: mother_ineq} is upper bounded by $(9/2)\times(\epsilon^2/9)=\epsilon^2/2$:
\begin{align}
\frac{9}{2}\min_{h\in\mH} L_{\mD_S}^{\rm unclip}(h) \leq \frac{9}{2}L_{\mD_S}^{\rm unclip}(h_{\tilde{\bw}}) =  \frac{9}{2}\underset{\rho\sim \mD}{\mathbb{E}} \left[\underset{\sigma \sim \mD_\rho}{\mathbb{E}} \left[ |g(\rho) - h_{\tilde{\bw}}(\sigma)|^2/2 \right]\right] \leq \epsilon^2/2. \label{eq_sm: main proof 4_1}
\end{align}
Here, $\mH=\mH(B_1)\cup\mH(B_2)\cup\cdots$ is the entire hypothesis class, and $\mH(B_i)$ is the hypothesis class with hyperparameter $B_i\in\mP$.

\vspace{0.5cm}
\noindent
{\bf (iii) Evaluating learning cost:}
We verify that $\bar{g}_{\rm CA}(\sigma)$ can be represented as a linear function in the feature space and then evaluate the norm of the weight vector $\tilde{\bw}$.
Given feature vector components of Eq.~\eqref{eq_sm: SK component}, $\bar{g}_{\rm CA}(\sigma)$ can be written as
\begin{align}
\bar{g}_{\rm CA}(\sigma) 
&= \sum_{i'} \barc_{i'}  \left[ \sqrt{\frac{b_{i'}!}{\tau^{b_{i'}}}}\prod_{\bar{P}\in \bar{\mP}_{i'}}  \sqrt{|\supp(\bar{P})|!\left(\frac{2n}{\gamma}\right)^{|\supp(\bar{P})|}} \right]  \\
&\hspace{1cm} \times \left[ \sqrt{\frac{\tau^{b_{i'}}}{b_{i'}!}}\prod_{\bar{P}\in \bar{\mP}_{i'}}  \sqrt{\frac{1}{|\supp(\bar{P})|!}\left(\frac{\gamma}{2n}\right)^{|\supp(\bar{P})|}} \tr[\bar{P} \sigma ]\right] \\
&\equiv \braket{\tilde{\bw},\phi_{\rm SK}(\sigma))},
\end{align}
indicating that it can be described as a linear function in the feature space of the shadow kernel.
Here, $b_{i'}=|\bar{\mP}_{i'}|$ is the number of Pauli strings included in $\bar{\mP}_{i'}$.
The norm of the weight vector $\tilde{\bw}$ is bounded as 
\begin{align}
\braket{\tilde{\bw},\tilde{\bw}}
&= \sum_{i'} |\barc_{i'}|^2\left[ \sqrt{\frac{b_{i'}!}{\tau^{b_{i'}}}}\prod_{\bar{P}\in \bar{\mP}_{i'}}  \sqrt{|\supp(\bar{P})|!\left(\frac{2n}{\gamma}\right)^{|\supp(\bar{P})|}} \right]^2 \\
&= \sum_{i} \sum_{t_1=1}^n \cdots \sum_{t_{b_i}=1}^n \frac{|\tc_i|^2}{n^{2b_i}}  \left[ \frac{b_i!}{\tau^{b_i}}\prod_{\tilde{P}\in \tilde{\mP}_{i,\bt}}  |\supp(\tilde{P})|!\left(\frac{2n}{\gamma}\right)^{|\supp(\tilde{P})|} \right] \\
&= \sum_{i}\frac{|\tc_i|^2}{n^{b_i}}  \left[ \frac{b_i!}{\tau^{b_i}}\prod_{\tilde{P}\in \tilde{\mP}_{i}}  |\supp(\tilde{P})|!\left(\frac{2n}{\gamma}\right)^{|\supp(\tilde{P})|} \right] \\
&\leq \sum_{i} \frac{|\tc_i|^2}{n^{b_i}}  \left[ \frac{b_i^{b_i}}{\tau^{b_i}} (mp)! \left(\frac{2n}{\gamma}\right)^{\sum_{\tilde{P}\in\tilde{\mP}_{i}}|\supp(\tilde{P})|} \right] \\
&\leq \sum_{i} |\tc_i|^2  \left[\frac{(mp)^{b_i}}{\tau^{b_i}} (mp)^{mp} \left(\frac{2}{\gamma}\right)^{mp}  \right] n^{\sum_{\tilde{P}\in\tilde{\mP}_{i}}|\supp(\tilde{P})| - b_i} \\
&\leq \sum_{i} |\tc_i|^2  \left(\frac{mp}{\tau}\right)^{mp} (mp)^{mp} \left(\frac{2}{\gamma}\right)^{mp}  n^{mp-\beta_g} \\
&\leq \|g\|_1^2  \left(\frac{2m^2p^2}{\tau\gamma}\right)^{mp} n^{mp-\beta_g},
\end{align}
where we have used $b_i! \leq b_i^{b_i}$ and $\prod_{\tilde{P}\in\tilde{\mP}_i}|\supp(\tilde{P})|!\leq (mp)!$ in the third line, $b_i\leq mp$ and $(mp)!\leq (mp)^{mp}$ in the fourth line, $b_i\leq mp$ and $\sum_{\tilde{P}\in\tilde{\mP}_{i}} |\supp(\tilde{P})| - b_i\leq mp-\beta_g$ [Eq.~\eqref{eq_sm: fi}] in the fifth line, and $\sum_i |\tc_i|^2 \leq (\sum_i |\tc_i|)^2 \leq (\sum_i |\hc_i|)^2=\|g_{\rm CA}\|_1^2\leq\|g\|_1^2$ in the sixth line [see also Eqs.~\eqref{eq_sm: cluster norm} and \eqref{eq_sm: convenient ineq}].
We have also used the assumptions of $2/\gamma \geq 1$ and $mp/\tau \geq 1$ in the fourth and fifth lines.
Therefore, $B^2 \geq \|g\|_{1}^2 ( 2m^2 p^2/\tau\gamma )^{mp} n^{mp-\beta_g}$ ensures Eq.~\eqref{eq_sm: main proof 4_1}.

We set
\begin{align}
    &N_z \geq \max\left(\frac{400C^2 B^2X^2}{(\epsilon^2/2)^2}, \frac{200 C^4 \log(4/\delta)}{(\epsilon^2/2)^2}, \frac{50(BX+C)^2 \log(8/\delta)}{\epsilon^2/2}\right), \\
    &N_v \geq \frac{800C^2 \log(4|\mP|/\delta)}{3(\epsilon^2/2)},
\end{align}
where $B^2 \geq \|g\|_{1}^2 ( 2m^2 p^2/\tau\gamma )^{mp} n^{mp-\beta_g}$ and $X^2=\exp(\tau \exp(5\gamma))\geq \braket{\phi_{{\rm SK}}(\sigma), \phi_{{\rm SK}}(\sigma)}$.
Then, from Lemma~\ref{lemma: mother_lem}, the following inequality holds with probability at least $1-\delta$:
\begin{align}
L_{\mD_S}^{\rm clip}(h_{ZV}^\ast) 
&\leq \frac{9}{2} \min_{h\in\mH} L_{\mD_S}^{\rm unclip}(h) + \frac{\epsilon^2}{2}. \label{eq_sm: main proof 4_2}
\end{align}
Combining Eqs.~\eqref{eq_sm: main proof 4_1} and \eqref{eq_sm: main proof 4_2}, we obtain
\begin{align}
L_{\mD_S}^{\rm clip}(h_{ZV}^\ast)
&\leq \epsilon^2.
\end{align}
This completes the proof.
\end{proof}

This theorem shows that if the hyperparameter set $\mP$ contains $B^2\sim \|g\|_1^2 (2m^2p^2/\tau\gamma)^{mp} n^{mp-\beta_g}$, the shadow kernel can achieve accuracy $\epsilon$ from $N_z=\mO(n^{mp-\beta_g}/\epsilon^4 + \log(1/\delta)/\epsilon^4 + n^{mp-\beta_g}\log(1/\delta)/\epsilon^2)$ training data, $N_v=\mO(\log(|\mP|/\delta)/\epsilon^2)$ validation data, and a shadow size $T=\mO(1/\epsilon^2)$.
If $|\mP|$ is independent of $n,\epsilon,\delta$, the total sample complexity is $(N_z+N_v)T=\mO(n^{mp-\beta_g}/\epsilon^6 + \log(1/\delta)/\epsilon^6 + n^{mp-\beta_g}\log(1/\delta)/\epsilon^4)$.
Unlike the GLQK, whose sample complexity remains constant, the shadow kernel's grows with the system size $n$, highlighting the prominent efficiency of GLQK.

\section{Theoretical demonstration beyond polynomial functions}\label{sec_sm: nonpolynomial}

The GLQK can, in principle, be applied to non-polynomial functions when they admit accurate polynomial approximations obtained by truncating their higher-order terms. 
However, a direct application of our polynomial-learning theorems to such truncated functions can lead to a prohibitive sample complexity, because the bounds in those theorems grow rapidly with the polynomial degree $p$ of the target function. This growth can make the approach inefficient for non-polynomial functions whose faithful approximation requires a high truncation order.

This limitation, however, reflects the worst-case nature of the previous
analysis, which guarantees the learnability of arbitrary polynomials of a
given degree. A target function with a specific analytic structure may be
learned at a substantially lower cost if that structure is incorporated into
the kernel design. The GLQK framework provides this flexibility through the
choice of the local kernel $k_A$ and the function $f$ used to combine local
kernels. When prior knowledge about the target function is available, encoding
it into the kernel can concentrate the feature space on the relevant
components and suppress irrelevant ones, thereby reducing the effective
complexity that controls the sample complexity.

To demonstrate this mechanism concretely, we show below that the GLQK based on the exponential-form truncated shadow kernel can efficiently learn exponential target functions of the form
\begin{align}
  g(\rho) = \sum_i c_i \exp(\tr[P_i \rho]),
\end{align}
where each $P_i$ is a geometrically local Pauli string of weight at most $m$ supported on $A_i \in \AGL(L)$ for a finite $L$, and $c_i \in \mathbb{R}$.
Also, let the quantum state $\rho$ be non-translationally symmetric in general.
Since each term of $g(\rho)$ is supported on a local subsystem of size $L$, it suffices to take the GLQK exponent as $\kappa=1$ and the subsystem size as $\zeta=L$. 
(This is analogous to the case of $\alpha_g=1$ for polynomials, though the local-cover number is formally defined only for polynomials.)
The exponential form of the kernel is the prior knowledge exploited here: it incorporates the analytic structure of the target directly into the feature space.

The efficiency originates from the compatibility between the Taylor
coefficients of the exponential function and the degree-wise normalization of
the GLQK feature space. The feature vector of the GLQK with $\kappa=1$
[Eqs.~\eqref{eq_sm: feature vector of GLQK-TSK} and
\eqref{eq_sm: GLQK component}] assigns a normalization factor
$\sim\sqrt{1/(n d!)}$ to its degree-$d$ component, which consists of products of
$d$ local expectation values. On the other hand, the Taylor expansion $\exp(\tr[P_i\rho])=\sum_{d} (\tr[P_i\rho])^d/d!$ shows that the degree-$d$ contribution of the target carries the coefficient
$1/d!$. Since $(\tr[P_i\rho])^d$ is contained in the degree-$d$ component of
the GLQK feature space, representing $g(\rho)$ as a linear function
$\langle \bm{w},\phi_{\mathrm{GL}}\rangle$ requires the corresponding
component of the weight vector $\bm{w}$ to scale as $\sqrt{n/d!}$. Consequently, the squared norm of the weight vector scales as
\begin{align}
    \|\bm{w}\|^2
    \sim
    \sum_{d=0}^{\infty} \frac{n}{d!}
    =
    \mO(n),
\end{align}
and, crucially, this bound is independent of any finite truncation order.
This contrasts with a generic high-degree polynomial representation, whose feature-space norm can grow rapidly with the degree.
Since the sample complexity of kernel ridge regression is controlled by the
squared feature-space norm of the target function, this norm bound enables
learning such exponential targets with $N_z=\mO(n)$ training data. 
This sample complexity of $\mO(n)$ is the same as for local polynomials with $\alpha_g=1$, suggesting that the $\mO(n)$-scaling comes solely from the locality of the GLQK and the target function.

By contrast, an analogous analysis for the conventional global shadow kernel leads to an exponentially larger weight-vector norm. The degree-$d$ component of the shadow-kernel feature vector carries a normalization factor of order
$\sqrt{1/((c n)^{m d} d!)}$, where $c>0$ is a constant independent of $n$ and
$d$. Matching the Taylor coefficient $1/d!$ of the exponential target then
requires a weight-vector component whose squared norm scales as
$(c n)^{m d}/d!$. Thus,
\begin{align}
    \|\bm{w}\|^2
    \sim
    \sum_{d=0}^{\infty}
    \frac{(c n)^{m d}}{d!}
    =
    \exp\!\left((c n)^m\right),
\end{align}
which yields the sample complexity scaling as $N_z=\mO(e^{(cn)^m})$. 
This comparison highlights the role of locality: the GLQK represents a local exponential target function using local feature components, whereas the shadow kernel distributes its representation over a much larger nonlocal feature space. 
This difference results in substantially better scaling in $n$ for the GLQK than for the shadow kernel.

We also note that one could instead approximate the exponential target by a finite-degree polynomial and then apply the polynomial-learning guarantee of the conventional shadow kernel to the approximated polynomial. 
This can replace the above exponential feature-space norm by a polynomial dependence on $n$ for a fixed truncation
degree $p$. However, the truncation degree required for a given target
accuracy depends on $\epsilon$, which generally leads to an unfavorable
accuracy dependence. Thus, both from the infinite-series feature-space
perspective and from the finite-degree truncation perspective, the
target-adapted GLQK provides a more favorable scaling for this class of local
exponential targets.

Below, we make these statements rigorous.

\subsection{Estimating exponential functions from classical shadows}

First, we show that the classical shadow can efficiently estimate the value of an exponential function.
\begin{lem}
\label{lem:exponential_pauli_shadow_bound}
Let $P\in\{I,X,Y,Z\}^{\otimes n}$ be a Pauli string of weight $m$, and let
\begin{align}
g(\rho)=\exp(\tr[P\rho]).   
\end{align}
Suppose that $\sigma$ is the classical shadow estimator obtained from $T$ independent random Pauli measurements. If $T\ge 11\cdot 3^m$, then
\begin{align}
\mathbb{E}_{\sigma\sim\mathcal D_\rho}
\left[
    \left|
        g(\sigma)-g(\rho)
    \right|^2
\right]
\le
11\cdot\frac{3^m}{T}.
\end{align}
\end{lem}

\begin{proof}
Let $x=\tr[P\rho]\in[-1,1]$ and $\widehat x_T=\tr[P\sigma]$.
For a weight-$m$ Pauli string $P$, the single-shot random-Pauli classical-shadow estimator $X_t$ for $\tr[P\rho]$ takes values in $\{0,\pm 3^m\}$.
More explicitly, setting $a=3^m$, we have
\begin{align}
\Pr[X_t=a]=\frac{1+x}{2a},
\qquad
\Pr[X_t=-a]=\frac{1-x}{2a},
\qquad
\Pr[X_t=0]=1-\frac1a . \label{eq: prob distribution}
\end{align}
This random variable satisfies $\mathbb E[X_t]=x$ and $\sum_{t=1}^T X_t/T = \widehat x_T$.
We write $Z_T=\widehat x_T-x$ and $u=a/T$.
By assumption, $u\le 1/11$.

We first derive an exponential moment bound for $Z_T$. For $r\ge 0$,
\begin{align}
    \mathbb E\!\left[e^{r\widehat x_T}\right]
    &=
    \left(
        \mathbb E\!\left[e^{rX_1/T}\right]
    \right)^T                                                    \nonumber\\
    &=
    \left(
        1+
        \frac{\cosh(ru)-1+x\sinh(ru)}{a}
    \right)^T,
\end{align}
where we have used Eq.~\eqref{eq: prob distribution}.
Therefore, using $\log(1+s)\le s$,
\begin{align}
    \log \mathbb E\!\left[e^{rZ_T}\right]
    &= \log \left(e^{-rx} \mathbb E\!\left[e^{r\widehat x_T}\right]\right)\\
    &=
    -rx+
    T\log
    \left(
        1+
        \frac{\cosh(ru)-1+x\sinh(ru)}{a}
    \right)                                                     \nonumber\\
    &\le
    -rx+
    \frac{\cosh(ru)-1+x\sinh(ru)}{u}                             \nonumber\\
    &=
    \frac{\cosh(ru)-1+x(\sinh(ru)-ru)}{u}.
\end{align}
Since $x\le 1$ and $\sinh(ru)-ru\ge 0$, we obtain
\begin{align}
    \log \mathbb E\!\left[e^{rZ_T}\right]
    \le
    \frac{\cosh(ru)-1+\sinh(ru)-ru}{u}
    =
    \frac{e^{ru}-1-ru}{u}.
\end{align}
Hence, defining $d_r(u)=(e^{ru}-1-ru)/u$,
we have
\begin{align}
\mathbb E\!\left[e^{rZ_T}\right]
\le
e^{d_r(u)}. \label{eq: edr}
\end{align}
Note that $d_r(u)\geq0$ for $u\geq0$.

We decompose the mean-squared error into variance and squared bias:
\begin{align}
    \mathbb E
    \left[
        \left(e^{\widehat x_T}-e^x\right)^2
    \right]
    =
    \Var(e^{\widehat x_T})
    +
    \left(
        \mathbb E[e^{\widehat x_T}]-e^x
    \right)^2 .
\end{align}
First, for the bias term, Jensen's inequality gives
$\mathbb E[e^{\widehat x_T}]\ge e^{\mathbb E[\widehat x_T]}=e^x$. Moreover,
\begin{align}
    \mathbb E[e^{\widehat x_T}]
    =
    e^x\mathbb E[e^{Z_T}]
    \le
    e^x e^{d_1(u)} .
\end{align}
Since $x\le 1$, it follows that
\begin{align}
    \left(
        \mathbb E[e^{\widehat x_T}]-e^x
    \right)^2
    \le
    e^2\left(e^{d_1(u)}-1\right)^2 .
    \label{eq:bias_exp_pauli}
\end{align}

Next, we bound the variance term. 
Let $A_t=e^{X_t/T}, \mu=\mathbb E[A_t]$, and $\nu=\mathbb E[A_t^2]$.
Then
\begin{align}
e^{\widehat x_T}=\prod_{t=1}^T A_t,    
\end{align}
and independence implies
\begin{align}
    \Var(e^{\widehat x_T})
    =
    \nu^T-\mu^{2T}.
\end{align}
Since $\nu\ge \mu^2$,
\begin{align}
    \nu^T-\mu^{2T}
    = (\nu-\mu^2)(\nu^{T-1} + \nu^{T-2}\mu^2 + \cdots + \mu^{2T-2})
    \le
    T(\nu-\mu^2)\nu^{T-1}.
\end{align}
Using the exponential moment bound with $r=2$,
\begin{align}
    \nu^T
    =
    \mathbb E[e^{2\widehat x_T}]
    =
    e^{2x}\mathbb E[e^{2Z_T}]
    \le
    e^2e^{d_2(u)}.    
\end{align}
Hence
\begin{align}
    \nu^{T-1} = (\nu^T)^{\frac{T-1}{T}} \leq \left(e^2e^{d_2(u)}\right)^{\frac{T-1}{T}}\leq e^2e^{d_2(u)},    
\end{align}
where we have used $e^2e^{d_2(u)}\geq1$.
It remains to bound $\nu-\mu^2$. Directly from the distribution of $X_t$,
\begin{align}
    \mu
    =
    1+
    \frac{c_1}{a},
    \quad
    \nu
    =
    1+
    \frac{c_2}{a},
\end{align}
where $c_1=\cosh u-1+x\sinh u$ and $c_2 = \cosh(2u)-1+x\sinh(2u)$.
Therefore, 
\begin{align}
    \nu-\mu^2
    &= \frac{c_2-2c_1}{a} - \frac{c_1^2}{a^2}\\
    &\leq \frac{c_2-2c_1}{a} \\
    &=
    \frac{
        \cosh(2u)-1+x\sinh(2u)
        -2(\cosh u-1+x\sinh u)
    }{a}                                                        \nonumber\\
    &=
    \frac{
        2(\cosh u-1)(\cosh u+x\sinh u)
    }{a}                                                        \nonumber\\
    &\le
    \frac{2e^u(\cosh u-1)}{a},
\end{align}
where we used $x\in[-1,1]$ and hence $\cosh u+x\sinh u\le \cosh u+\sinh u=e^u$.
Combining these estimates yields
\begin{align}
    \Var(e^{\widehat x_T})
    \le
    T\frac{2e^u(\cosh u-1)}{a}
    e^2e^{d_2(u)}            
    =
    e^2e^{d_2(u)}
    \frac{2e^u(\cosh u-1)}{u}.
    \label{eq:var_exp_pauli}
\end{align}
Combining \eqref{eq:bias_exp_pauli} and \eqref{eq:var_exp_pauli}, we obtain
\begin{align}
    \mathbb E
    \left[
        \left(e^{\widehat x_T}-e^x\right)^2
    \right]
    \le
    e^2
    \left[
        e^{d_2(u)}
        \frac{2e^u(\cosh u-1)}{u}
        +
        \left(e^{d_1(u)}-1\right)^2
    \right].
\end{align}
Equivalently, 
\begin{align}
    &\mathbb E
    \left[
        \left(e^{\widehat x_T}-e^x\right)^2
    \right]
    \le
    C(u) \frac{3^m}{T},
\end{align}
where
\begin{align}
    &C(u) 
    = e^2
    \left[
        e^{d_2(u)}
        \frac{2e^u(\cosh u-1)}{u^2}
        +\frac{\left(e^{d_1(u)}-1\right)^2}{u}
    \right].
    \label{eq:pre_constant_bound_exp_pauli}
\end{align}

We finally evaluate $C(u)$ under $u\le 1/11$, which follows from $T\ge 11\cdot 3^m$. 
The function $C(u)$ is monotonically increasing for $u>0$, since each term in its Taylor expansion around $u=0$ has non-negative coefficients. 
In fact, $d_r(u)=\sum_{k=2}^\infty (ru)^k /k!$, $(\cosh u -1)/u^2 = \sum_{k=0}^\infty u^{2k}/(2k+2)!$, and $(e^{d_1(u)}-1)^2/u = [\sum_{\ell=1}^\infty (u^{\ell-(1/2)}/\ell!)(\sum_{k=2}^\infty u^{k-1}/k!)^\ell]^2$.
Therefore, we have
\begin{align}
C(u) \leq C(1/11) = 10.012\ldots < 11
\end{align}
for $u\leq1/11$.
This provides that
\begin{align}
    &\mathbb E
    \left[
        \left(e^{\widehat x_T}-e^x\right)^2
    \right]
    \le
    11 \cdot\frac{3^m}{T},
\end{align}
which means
\begin{align}
\mathbb{E}_{\sigma\sim\mathcal D_\rho}
\left[
    \left|
        g(\sigma)-g(\rho)
    \right|^2
\right]
\le
11\cdot\frac{3^m}{T},
\end{align}
as required.
\end{proof}

\begin{lem}
\label{lem:sum_exponential_pauli_shadow_bound}
Let $P_i\in\{I,X,Y,Z\}^{\otimes n}$ be a Pauli string of weight $m_i\leq m$, and let
\begin{align}
g(\rho)=\sum_i c_i \exp(\tr[P_i\rho]).    
\end{align}
Suppose that $\sigma$ is the classical shadow estimator obtained from $T$ independent random Pauli measurements. If $T\ge 11\cdot 3^m$, then
\begin{align}
\mathbb{E}_{\sigma\sim\mathcal D_\rho}
\left[
    \left|
        g(\sigma)-g(\rho)
    \right|^2
\right]
\le
\frac{11 \|g\|_{\rm exp}^2 3^{m}}{T},
\end{align}
where $\|g\|_{\rm exp} = \sum_i |c_i|$.
\end{lem}

\begin{proof}
By the triangle inequality,
\begin{align}
\sqrt{\mathbb{E}[|g(\rho)-g(\sigma)|^2]} 
&\leq \sum_{i} |c_i| \sqrt{\mathbb{E}[|\exp(\tr[P_i\rho])-\exp(\tr[P_i\sigma])|^2]}.
\end{align}
From Lemma~\ref{lem:exponential_pauli_shadow_bound}, for $T\geq 11\cdot3^m$, we have
\begin{align}
\sqrt{\mathbb{E}[|g(\rho)-g(\sigma)|^2]} 
\leq \sum_i |c_i| \sqrt{\frac{11\cdot 3^{m_i}}{T}} 
\leq \sqrt{\frac{11 \cdot 3^{m}}{T}} \|g\|_{\rm exp},
\end{align}
where we have used $m_i\leq m$.
Therefore, we obtain
\begin{align}
\mathbb{E}[|g(\rho)-g(\sigma)|^2]
&\leq \frac{11 \|g\|_{\rm exp}^2 3^{m}}{T}.
\end{align}

\end{proof}

\begin{lem}
\label{lem:sum_exponential_Tsize}
Let $P_i\in\{I,X,Y,Z\}^{\otimes n}$ be a Pauli string of weight $m_i\leq m$, and let
\begin{align}
g(\rho)=\sum_i c_i \exp(\tr[P_i\rho]).    
\end{align}
For $\epsilon\in(0,\|g\|_{\rm exp})$, suppose we obtain a classical shadow $\sigma$ of size $T$ such that
\begin{align}
T \geq \frac{11 \|g\|_{\rm exp}^2 3^m}{\epsilon^2}.
\end{align}
Then, the following inequality holds:
\begin{align}
\mathbb{E}_{\sigma\sim\mathcal D_\rho}
\left[
    \left|
        g(\sigma)-g(\rho)
    \right|^2
\right]
\le
\epsilon^2,
\end{align}
where $\|g\|_{\rm exp} = \sum_i |c_i|$.
\end{lem}

\begin{proof}
Since $\epsilon\in(0,\|g\|_{\rm exp})$, $T\geq 11 \|g\|_{\rm exp}^2 3^m/\epsilon^2 \geq 11\cdot3^m$, which is the condition for applying Lemma~\ref{lem:sum_exponential_pauli_shadow_bound}.
From Lemma~\ref{lem:sum_exponential_pauli_shadow_bound}, we obtain
\begin{align}
\mathbb{E}[|g(\rho)-g(\sigma)|^2] \leq \frac{11 \|g\|_{\rm exp}^2 3^{m}}{T} \leq 11 \|g\|_{\rm exp}^2 3^{m} \cdot \frac{\epsilon^2}{11 \|g\|_{\rm exp}^2 3^{m}} = \epsilon^2,
\end{align}
where we have used $T\geq 11 \|g\|_{\rm exp}^2 3^m/\epsilon^2$.
\end{proof}

\subsection{Theoretical guarantee of GLQK for exponential functions}

We now show that the GLQK can efficiently learn a class of local exponential
functions without truncating them to finite-degree polynomials.
\begin{thm}\label{thm_sm_rev: main_nontranslation}
Consider a distribution $\mD_S$ over $\mX\times\mY$, where $\mX$ is the input domain of $n$-qubit classical shadows of size $T$ on a $D$-dimensional hypercubic lattice, and $\mY\subset\mathbb{R}$ denotes the output range. 
The target function is $g(\rho)=\sum_i c_i \exp(\tr[P_i\rho])$, where $P_i$ is an $m_i$-weight Pauli string with $m_i\leq m$, supported on a local subsystem $A_i\in\AGL(L)$, and $P_i\neq P_j$ for $i\neq j$.
Assume that an upper bound of $|g(\rho)|$ is given: $|g(\rho)|\leq C$ for all $\rho$ drawn from the distribution.
Let $\zeta \geq L$, $\kappa\geq 1$, and $B^2 \geq n^\kappa \|g\|_{\rm exp}^2 \exp\left[(1/\tau)(2\zeta^D/\gamma)^{m} \right]$, where $\|g\|_{\rm exp}=\sum_i |c_i|$.
Let $\mP=\{(\zeta_i,\kappa_i,B_i)\}_{i}$ be a finite set of candidate hyperparameter values such that $(\zeta,\kappa,B)\in\mP$, where $\zeta_i$ denotes the subsystem size, $\kappa_i$ the GLQK exponent, and $B_i$ the regularization parameter.
For $\epsilon\in (0,\|g\|_{\rm exp})$ and $\delta\in(0,1)$, suppose that we obtain $N=N_z+N_v$ classical shadows of size $T$, together with their target labels, as training and validation datasets such that
\begin{align}
    &N_z \geq \max\left(\frac{1600C^2 B^2X^2}{\epsilon^4}, \frac{800 C^4 \log(4/\delta)}{\epsilon^4}, \frac{100(BX+C)^2 \log(8/\delta)}{\epsilon^2}\right), \\
    &N_v \geq \frac{1600C^2 \log(4|\mP|/\delta)}{3\epsilon^2}, \\
    &T\geq \frac{99 \|g\|_{\rm exp}^2 3^m}{2\epsilon^2},
\end{align}
where $X^2=\exp(\kappa\tau \exp(5\gamma))$.
Then, kernel ridge regression trained on $N_z$ training data using the polynomial GLQK based on the truncated shadow kernel, combined with hyperparameter optimization via $N_v$ validation data, achieves
\begin{align}
    L_{\mD_S}^{\rm clip}(h_{ZV}^\ast) \leq \epsilon^2
\end{align}
with probability at least $1-\delta$.
Here, we have assumed that $2\zeta^D/\gamma \geq 1$.
\end{thm}    

\begin{proof}

As in Theorem~\ref{thm_sm: main_nontranslation}, the proof is based on
Lemma~\ref{lemma: mother_lem}. Since each $P_i$ is already supported on a
local subsystem, no cluster approximation is required in the present
argument.

\vspace{0.5cm}
\noindent
{\bf (i) Error in estimating the exponential from classical shadows:}
First, we evaluate the first term on the right-hand side in Eq.~\eqref{eq_sm: mother_ineq}.
Let $T\geq 11 \|g\|_{\rm exp}^2 3^m/(2\epsilon^2/9)$.
Then, by Lemma~\ref{lem:sum_exponential_Tsize}, the classical shadow $\sigma$ obeys
\begin{align}
    &\underset{\sigma\sim\mD_\rho}{\mathbb{E}}[|g(\rho) - g(\sigma)|^2]\leq 2\epsilon^2/9, 
\end{align}
for any $\rho$.
Taking expectation values with respect to $\rho\sim\mD$, we have
\begin{align}
    \underset{\rho\sim \mD}{\mathbb{E}} \left[\underset{\sigma \sim \mD_\rho}{\mathbb{E}} \left[ |g(\rho) - g(\sigma)|^2 \right]\right] \leq 2\epsilon^2/9.
\end{align}
If $g(\sigma)$ can be represented as a linear function in the feature space of GLQK with hyperparameters $\btheta=(\zeta, \kappa, B)$, i.e., $g(\sigma)=\braket{\tilde{\bw}, \phi_{{\rm GL}, \zeta,\kappa}(\sigma)}\equiv h_{\tilde{\bw},\zeta,\kappa}(\sigma)$ for some $\tilde{\bw}$ such that $\|\tilde{\bw}\|\leq B$ ($\phi_{{\rm GL},\zeta,\kappa}$ is the feature vector of GLQK with a subsystem size $\zeta$ and an exponent $\kappa$), the first term on the right-hand side in Eq.~\eqref{eq_sm: mother_ineq} is upper bounded as
\begin{align}
\frac{9}{2}\min_{h\in\mH} L_{\mD_S}^{\rm unclip}(h) \leq \frac{9}{2}L_{\mD_S}^{\rm unclip}(h_{\tilde{\bw},\zeta,\kappa}) =  \frac{9}{2}\underset{\rho\sim \mD}{\mathbb{E}} \left[\underset{\sigma \sim \mD_\rho}{\mathbb{E}} \left[ |g(\rho) - h_{\tilde{\bw},\zeta,\kappa}(\sigma)|^2/2 \right]\right] \leq \epsilon^2/2. \label{eq_sm_rev: main proof 1_1}
\end{align}
Here, $\mH=\mH(\btheta_1)\cup\mH(\btheta_2)\cup\cdots$ is the entire hypothesis class, and $\mH(\btheta_i)$ is the hypothesis class with hyperparameter $\btheta_i\in\mP$.

\vspace{0.5cm}
\noindent
{\bf (ii) Evaluating learning cost:}
We verify that $g(\sigma)$ can be represented as a linear function in the feature space with hyperparameters $\btheta=(\zeta, \kappa, B)$ and then evaluate the norm of the weight vector $\tilde{\bw}$.
The target function can be expanded as $g(\sigma) = \sum_i c_i + \sum_i \sum_{d=1}^\infty c_i (\tr[P_i\sigma])^d/d!$.
Since each $P_i$ is supported on some $A_i\in\AGL(L)$, and since $\zeta\geq L$, the monomial $(\tr[P_i\sigma])^d$ is contained in the feature space of the GLQK with exponent $\kappa\geq1$.
Using Eqs.~\eqref{eq_sm: feature vector of GLQK-TSK} and \eqref{eq_sm: GLQK component}, we can write 
\begin{align}
g(\sigma) 
&= \left(\sum_i c_i n^{\kappa/2}\right)\frac{1}{n^{\kappa/2}}  \\
&+\sum_i \sum_{d=1}^\infty \frac{c_i}{d!} \left[n^{\kappa/2} \left(\frac{d!}{\tau^d}\right)^{1/2} \left(\frac{2\zeta^D}{\gamma}\right)^{d|\supp(P_i)|/2} \right] 
\left[ \frac{1}{n^{\kappa/2}} \left(\frac{\tau^{d}}{d!}\right)^{1/2} \left(\frac{\gamma}{2\zeta^D}\right)^{d|\supp(P_i)|/2} \left(\tr[P_i \sigma ]\right)^d \right]  \\
&\equiv \braket{\tilde{\bw},\phi_{{\rm GL},\zeta,\kappa}(\sigma)}.
\end{align}
Here, we have combined the zeroth order in the Taylor series of each $\exp(\tr[P_i\sigma])$ because they share the same feature component.
The norm of the weight vector $\tilde{\bw}$ is bounded as 
\begin{align}
\braket{\tilde{\bw},\tilde{\bw}}
&= \left( \sum_i c_i n^{\kappa/2} \right)^2  + \sum_i \sum_{d=1}^\infty \left|\frac{c_i}{d!}\right|^2 \left[n^{\kappa/2} \left(\frac{d!}{\tau^d}\right)^{1/2} \left(\frac{2\zeta^D}{\gamma}\right)^{d|\supp(P_i)|/2} \right]^2 \\
&\leq n^\kappa \left(\sum_i |c_i|\right)^2 +  n^\kappa \sum_{i} \sum_{d=1}^\infty \left|\frac{c_i}{d!}\right|^2 \left(\frac{d!}{\tau^d}\right) \left(\frac{2\zeta^D}{\gamma}\right)^{dm} \\
&= n^\kappa \left(\sum_i |c_i|\right)^2 + n^\kappa \sum_{i} |c_i|^2 \sum_{d=1}^\infty \frac{1}{d!} \left[\frac{1}{\tau}\left(\frac{2\zeta^D}{\gamma}\right)^{m}\right]^d \\
&\leq n^\kappa \left(\sum_i |c_i|\right)^2 + n^\kappa \left(\sum_{i} |c_i|\right)^2 \sum_{d=1}^\infty \frac{1}{d!} \left[\frac{1}{\tau}\left(\frac{2\zeta^D}{\gamma}\right)^{m}\right]^d \\
&= n^\kappa  \left(\sum_{i} |c_i|\right)^2 \sum_{d=0}^\infty \frac{1}{d!} \left[\frac{1}{\tau}\left(\frac{2\zeta^D}{\gamma}\right)^{m}\right]^d \\
&=n^\kappa \|g\|_{\rm exp}^2 \exp\left[\frac{1}{\tau}\left(\frac{2\zeta^D}{\gamma}\right)^{m}\right]
\end{align}
where we have used $(\sum_i c_i)^2 \leq (\sum_i |c_i|)^2$, $|\supp(P_i)|\leq m$, $2\zeta^D/\gamma \geq1$, $\sum_i|c_i|^2 \leq (\sum_i |c_i|)^2$, and $\|g\|_{\rm exp} = \sum_i |c_i|$.
Therefore, the assumed lower bound on $B^2$ ensures that $\|\tilde{\bw}\|\leq B$, and hence Eq.~\eqref{eq_sm_rev: main proof 1_1} holds.

We set
\begin{align}
    &N_z \geq \max\left(\frac{400C^2 B^2X^2}{(\epsilon^2/2)^2}, \frac{200 C^4 \log(4/\delta)}{(\epsilon^2/2)^2}, \frac{50(BX+C)^2 \log(8/\delta)}{\epsilon^2/2}\right), \\
    &N_v \geq \frac{800C^2 \log(4|\mP|/\delta)}{3(\epsilon^2/2)},
\end{align}
where $B^2\geq n^\kappa \|g\|_{\rm exp}^2 \exp\left[(1/\tau)(2\zeta^D/\gamma)^{m} \right]$ and $X^2=\exp(\kappa\tau \exp(5\gamma))\geq \braket{\phi_{{\rm GL},\zeta,\kappa}(\sigma), \phi_{{\rm GL},\zeta,\kappa}(\sigma)}$.
Then, from Lemma~\ref{lemma: mother_lem}, the following inequality holds with probability at least $1-\delta$:
\begin{align}
L_{\mD_S}^{\rm clip}(h_{ZV}^\ast) 
&\leq \frac{9}{2} \min_{h\in\mH} L_{\mD_S}^{\rm unclip}(h) + \frac{\epsilon^2}{2}. \label{eq_sm_rev: main proof 1_2}
\end{align}
Combining Eqs.~\eqref{eq_sm_rev: main proof 1_1} and \eqref{eq_sm_rev: main proof 1_2}, we obtain
\begin{align}
L_{\mD_S}^{\rm clip}(h_{ZV}^\ast)
&\leq \epsilon^2.
\end{align}
This completes the proof.
\end{proof}

This theorem proves that the GLQK can learn local exponential functions from $N_z\sim \mOt(n/\epsilon^4)$ training data, $N_v\sim \mOt(\log(|\mP|)/\epsilon^2)$ validation data, and $T\sim \mO(1/\epsilon^2)$ shadow size, by setting $\kappa=1$ and $\zeta\geq L$, provided that $\|g\|_{\rm exp}, m, L, C$ are fixed.
Importantly, this guarantee incurs no truncation error from approximating $g(\rho)$ by a finite-degree polynomial.
The result, therefore, shows that matching the feature-space structure to the analytic form of the target function, together with locality, can reduce the sample complexity compared with a direct application of the generic polynomial-learning theorem.

It is instructive to see why a naive application of
Theorem~\ref{thm_sm: main_nontranslation} to a truncated exponential target
is suboptimal. For $|\tr[P_i\rho]|\leq 1$, the Taylor tail is bounded by
$\sum_{d=p+1}^{\infty}|\tr[P_i\rho]|^d/d! \leq e/(p+1)!$. Hence, achieving
a truncation error of order $\epsilon$ requires $(p+1)! = \Omega(1/\epsilon)$,
which implies $p=\Theta(\log(1/\epsilon)/\log\log(1/\epsilon))$ up to
lower-order corrections. Substituting this degree into
Theorem~\ref{thm_sm: main_nontranslation}, the dominant $p$-dependence of
the training data size enters through the regularization bound
$B^2\gtrsim \|g\|_1^2(2mp\zeta^D/(\tau\gamma))^{mp}n^\kappa$. For the present
local exponential target, the local-cover number is one, so one may take
$\kappa=1$. Writing $c=2m\zeta^D/(\tau\gamma)$, the degree-dependent factor is
$(cp)^{mp}$, which scales as $\epsilon^{-m}$ up to subpolynomial factors in
$1/\epsilon$. Consequently, the resulting training sample complexity behaves
as $N_z=\mO(n/\epsilon^{4+m})$ up to subpolynomial factors in
$1/\epsilon$. This is worse than the
$N_z=\mO(n/\epsilon^4)$ scaling obtained in
Theorem~\ref{thm_sm_rev: main_nontranslation}, which exploits the exponential
form of the target directly. This comparison quantitatively illustrates that
incorporating prior knowledge of the target structure into the kernel design,
rather than treating the target as a generic high-degree polynomial, is
essential for efficient learning.

\subsection{Theoretical guarantee of shadow kernel for exponential functions}

We now analyze the sample complexity of the conventional shadow kernel for learning local exponential target functions.
\begin{thm}\label{thm_sm_rev: sk_nontranslation}
Consider a distribution $\mD_S$ over $\mX\times\mY$, where $\mX$ is the input domain of $n$-qubit classical shadows of size $T$ on a $D$-dimensional hypercubic lattice, and $\mY\subset\mathbb{R}$ denotes the output range. 
The target function is $g(\rho)=\sum_i c_i \exp(\tr[P_i\rho])$, where $P_i$ is an $m_i$-weight Pauli string with $m_i\leq m$, supported on a local subsystem $A_i\in\AGL(L)$, and $P_i\neq P_j$ for $i\neq j$.
Assume that an upper bound of $|g(\rho)|$ is given: $|g(\rho)|\leq C$ for all $\rho$ drawn from the distribution.
Let $B^2\geq \|g\|_{\rm exp}^2 \exp\left[(1/\tau)(2mn/\gamma)^{m} \right]$, where $\|g\|_{\rm exp}=\sum_i |c_i|$.
Let $\mP=\{B_i\}_{i}$ be a finite set of candidate hyperparameter values such that $B\in\mP$, where $B_i$ is the regularization parameter.
For $\epsilon\in (0,\|g\|_{\rm exp})$ and $\delta\in(0,1)$, suppose that we obtain $N=N_z+N_v$ classical shadows of size $T$, together with their target labels, as training and validation datasets such that
\begin{align}
    &N_z \geq \max\left(\frac{1600C^2 B^2X^2}{\epsilon^4}, \frac{800 C^4 \log(4/\delta)}{\epsilon^4}, \frac{100(BX+C)^2 \log(8/\delta)}{\epsilon^2}\right), \\
    &N_v \geq \frac{1600C^2 \log(4|\mP|/\delta)}{3\epsilon^2}, \\
    &T\geq \frac{99 \|g\|_{\rm exp}^2 3^m}{2\epsilon^2},
\end{align}
where $X^2=\exp(\tau \exp(5\gamma))$.
Then, kernel ridge regression trained on $N_z$ training data using the shadow kernel, combined with hyperparameter optimization via $N_v$ validation data, achieves
\begin{align}
    L_{\mD_S}^{\rm clip}(h_{ZV}^\ast) \leq \epsilon^2
\end{align}
with probability at least $1-\delta$.
Here, we have assumed that $2n/\gamma\geq1$.
\end{thm}    

\begin{proof}

Similar to Theorem~\ref{thm_sm: SK_nontranslation}, the proof is based on Lemma~\ref{lemma: mother_lem}.

\vspace{0.5cm}
\noindent
{\bf (i) Error in estimating the exponential from classical shadows:}
First, we evaluate the first term on the right-hand side in Eq.~\eqref{eq_sm: mother_ineq}.
Let $T\geq 11 \|g\|_{\rm exp}^2 3^m/(2\epsilon^2/9)$.
Then, by Lemma~\ref{lem:sum_exponential_Tsize}, the classical shadow $\sigma$ obeys
\begin{align}
    &\underset{\sigma\sim\mD_\rho}{\mathbb{E}}[|g(\rho) - g(\sigma)|^2]\leq 2\epsilon^2/9, 
\end{align}
for any $\rho$.
Taking expectation values with respect to $\rho\sim\mD$, we have
\begin{align}
    \underset{\rho\sim \mD}{\mathbb{E}} \left[\underset{\sigma \sim \mD_\rho}{\mathbb{E}} \left[ |g(\rho) - g(\sigma)|^2 \right]\right] \leq 2\epsilon^2/9.
\end{align}
If $g(\sigma)$ can be represented as a linear function in the feature space of shadow kernel with a hyperparameter $B\in\mP$, i.e., $g(\sigma)=\braket{\tilde{\bw}, \phi_{{\rm SK}}(\sigma)}\equiv h_{\tilde{\bw}}(\sigma)$ for some $\tilde{\bw}$ such that $\|\tilde{\bw}\|\leq B$, the first term on the right-hand side in Eq.~\eqref{eq_sm: mother_ineq} is upper bounded as
\begin{align}
\frac{9}{2}\min_{h\in\mH} L_{\mD_S}^{\rm unclip}(h) \leq \frac{9}{2}L_{\mD_S}^{\rm unclip}(h_{\tilde{\bw}}) =  \frac{9}{2}\underset{\rho\sim \mD}{\mathbb{E}} \left[\underset{\sigma \sim \mD_\rho}{\mathbb{E}} \left[ |g(\rho) - h_{\tilde{\bw}}(\sigma)|^2/2 \right]\right] \leq \epsilon^2/2. \label{eq_sm_rev: main proof 2_1}
\end{align}
Here, $\mH=\mH(B_1)\cup\mH(B_2)\cup\cdots$ is the entire hypothesis class, and $\mH(B_i)$ is the hypothesis class with hyperparameter $B_i\in\mP$.

\vspace{0.5cm}
\noindent
{\bf (ii) Evaluating learning cost:}
We verify that $g(\sigma)$ can be represented as a linear function in the feature space and then evaluate the norm of the weight vector $\tilde{\bw}$.
The target function can be expanded as $g(\sigma) = \sum_i c_i + \sum_i \sum_{d=1}^\infty c_i (\tr[P_i\sigma])^d/d!$.
Given feature vector components of Eq.~\eqref{eq_sm: SK component}, $g(\sigma)$ can be written as
\begin{align}
g(\sigma) 
&= \left(\sum_i c_i \right)  + \sum_i \sum_{d=1}^\infty \frac{c_i}{d!} \left[\left(\frac{d!}{\tau^d}\right)^{1/2} (|\supp(P_i)|!)^{d/2} \left(\frac{2n}{\gamma}\right)^{d|\supp(P_i)|/2} \right]  \\
&\hspace{3.7cm}\times\left[ \left(\frac{\tau^{d}}{d!}\right)^{1/2} \left(\frac{1}{|\supp(P_i)|!}\right)^{d/2}\left(\frac{\gamma}{2n}\right)^{d|\supp(P_i)|/2} \left(\tr[P_i \sigma ]\right)^d \right] \\
&\equiv \braket{\tilde{\bw},\phi_{{\rm SK}}(\sigma)}.
\end{align}
Here, we have combined the zeroth order in the Taylor series of each $\exp(\tr[P_i\sigma])$ because they share the same feature component.
The norm of the weight vector $\tilde{\bw}$ is bounded as 
\begin{align}
\braket{\tilde{\bw},\tilde{\bw}}
&=  \left(\sum_i c_i \right)^2 +  \sum_i \sum_{d=1}^\infty \left|\frac{c_i}{d!}\right|^2 \left[\left(\frac{d!}{\tau^d}\right)^{1/2} (|\supp(P_i)|!)^{d/2} \left(\frac{2n}{\gamma}\right)^{d|\supp(P_i)|/2} \right]^2 \\
&\leq  \left(\sum_i |c_i| \right)^2 + \sum_{i} \sum_{d=1}^\infty \left|\frac{c_i}{d!}\right|^2 \left(\frac{d!}{\tau^d}\right) (m!)^d \left(\frac{2n}{\gamma}\right)^{dm} \\
&\leq  \left(\sum_i |c_i| \right)^2 + \sum_{i} \sum_{d=1}^\infty \left|\frac{c_i}{d!}\right|^2 \left(\frac{d!}{\tau^d}\right) \left(\frac{2mn}{\gamma}\right)^{dm} \\
&= \left(\sum_i |c_i| \right)^2 + \sum_{i} |c_i|^2 \sum_{d=1}^\infty \frac{1}{d!} \left[\frac{1}{\tau}\left(\frac{2mn}{\gamma}\right)^{m}\right]^d \\
&\leq \left(\sum_i |c_i| \right)^2 + \left(\sum_i |c_i| \right)^2 \sum_{d=1}^\infty \frac{1}{d!} \left[\frac{1}{\tau}\left(\frac{2mn}{\gamma}\right)^{m}\right]^d \\
&= \left(\sum_i |c_i| \right)^2 \sum_{d=0}^\infty \frac{1}{d!} \left[\frac{1}{\tau}\left(\frac{2mn}{\gamma}\right)^{m}\right]^d \\
&= \left(\sum_i |c_i| \right)^2 \exp\left[\frac{1}{\tau}\left(\frac{2mn}{\gamma}\right)^{m}\right] \\
&= \|g\|_{\rm exp}^2 \exp\left[\frac{1}{\tau}\left(\frac{2mn}{\gamma}\right)^{m}\right],
\end{align}
where we have used $(\sum_i c_i)^2\leq(\sum_i|c_i|)^2$, $|\supp(P_i)|\leq m$, $2n/\gamma \geq1$, $m!\leq m^m$, $\sum_i |c_i|^2 \leq (\sum_i |c_i|)^2$, and $\|g\|_{\rm exp} = \sum_i |c_i|$.
Therefore, the assumed lower bound on $B^2$ ensures that $\|\tilde{\bw}\|\leq B$, and hence Eq.~\eqref{eq_sm_rev: main proof 2_1} holds.

We set
\begin{align}
    &N_z \geq \max\left(\frac{400C^2 B^2X^2}{(\epsilon^2/2)^2}, \frac{200 C^4 \log(4/\delta)}{(\epsilon^2/2)^2}, \frac{50(BX+C)^2 \log(8/\delta)}{\epsilon^2/2}\right), \\
    &N_v \geq \frac{800C^2 \log(4|\mP|/\delta)}{3(\epsilon^2/2)},
\end{align}
where $B^2\geq \|g\|_{\rm exp}^2 \exp\left[(1/\tau)(2mn/\gamma)^{m} \right]$ and $X^2=\exp(\tau \exp(5\gamma))\geq \braket{\phi_{{\rm SK}}(\sigma), \phi_{{\rm SK}}(\sigma)}$.
Then, from Lemma~\ref{lemma: mother_lem}, the following inequality holds with probability at least $1-\delta$:
\begin{align}
L_{\mD_S}^{\rm clip}(h_{ZV}^\ast) 
&\leq \frac{9}{2} \min_{h\in\mH} L_{\mD_S}^{\rm unclip}(h) + \frac{\epsilon^2}{2}. \label{eq_sm_rev: main proof 2_2}
\end{align}
Combining Eqs.~\eqref{eq_sm_rev: main proof 2_1} and \eqref{eq_sm_rev: main proof 2_2}, we obtain
\begin{align}
L_{\mD_S}^{\rm clip}(h_{ZV}^\ast)
&\leq \epsilon^2.
\end{align}
This completes the proof.
\end{proof}

This theorem establishes an efficiency separation, in terms of the scaling in $n$, between the GLQK and the shadow kernel for learning exponential-type functions: the GLQK can learn from only $\mOt(n)$ data, while the same analysis yields $\mOt(e^{(cn)^m})$ for the shadow kernel ($c\in\mathbb{R}$ is a constant independent of $n$). This gap stems from locality---the GLQK captures the local exponential target through local feature components, whereas the shadow kernel must spread its representation across a much larger nonlocal feature space.

We note that the exponential-in-$n$ scaling for the shadow kernel can be avoided by a different route: one may truncate the Taylor expansion of $g(\rho)$ at degree $p$ and then apply Theorem~\ref{thm_sm: SK_nontranslation} to the resulting finite-degree polynomial. 
This yields a polynomial dependence on the system size, scaling as $n^{mp}$ for fixed $p$. However, the required truncation degree $p$ depends on the target accuracy $\epsilon$, which worsens the accuracy dependence. 
Thus, the polynomial-truncation approach improves the $n$-scaling compared with the direct infinite-series representation in the shadow-kernel feature space, but it still leads to less favorable scaling than the GLQK, which exploits the locality of the target function.

\section{Details of numerical experiments}\label{sec_sm: numerics}

\noindent
{\bf Overall pipeline:}
To reduce the computational cost, we first prepare $N_{\rm pool}$ classical shadows of size $T=500$ for the data pool.
In the regression task involving random quantum dynamics, we use the time-evolving block-decimation (TEBD) algorithm to generate $N_{\rm pool}=1500$ data points for the translationally symmetric case and $N_{\rm pool}=8500$ data points for the non-translationally symmetric case.
In the quantum phase recognition task, we employ the density matrix renormalization group (DMRG) to generate $N_{\rm pool}=1000$ data points.
These tensor network algorithms are implemented with ITensor~\cite{Fishman2022-lh}.

The pipeline for evaluating the performance of the ML models is as follows:
\begin{enumerate}
    \item[(i)] Randomly sample $N$ training data and $M$ test data from the pool of $N_{\rm pool}$ data.
    \item[(ii)] Train the kernel model from the training data using the procedure described below.
    \item[(iii)] Calculate the prediction accuracy for the test data with the trained kernel model.
\end{enumerate}
This procedure of (i)--(iii) is repeated 10 times while changing the choice of training and test data, and the average of these 10 scores is the final test accuracy plotted in the figures.

\vspace{0.5cm}
\noindent
{\bf Training algorithm:}
We solve the two tasks with the kernel ridge regression and the support vector machine, respectively.
These algorithms are implemented using scikit-learn~\cite{Pedregosa2011-mb}. 
To align the scale of the data in the feature space, we standardize the kernel matrix:
\begin{align}
    K_{ij} \to \tilde{K}_{ij}=\frac{K_{ij}}{\sqrt{K_{ii}K_{jj}}}.
\end{align}
In our preliminary numerical experiments, this standardization improves the accuracy of the shadow kernel in the quantum phase recognition task, but has little effect on the other task and GLQK; rather, it slightly worsens the results.
Nonetheless, to ensure consistent calculation conditions, we performed the standardization across all numerical experiments.

During the training process, we use grid search combined with cross-validation to optimize some hyperparameters. 
For the GLQK, the regularization strength $\lambda$, the exponent of the polynomial GLQK $\kappa$, and the size of local subsystems $\zeta$ are optimized. 
In the shadow kernel, only $\lambda$ is optimized.
Grid search helps us find the best values for these parameters by evaluating prediction accuracy through cross-validation. 
For grid search, we adopt the parameter sets $P_\lambda=\{$0.0001, 0.001, 0.01, 0.1, 1, 10, 100, 1000, 10000$\}$ for $\lambda$, $P_\kappa=\{1,2\}$ for $\kappa$, and $P_\zeta=\{1,2,4,6\}$ (the regression task of random quantum dynamics) or $P_\zeta=\{2,4,6\}$ (the quantum phase recognition task) for $\zeta$.
We use five-fold cross-validation, which involves randomly dividing the training dataset into five parts. 
We train the kernel model using four of these parts and then calculate the prediction accuracy on the remaining part as validation data. 
This process is conducted for five possible choices of validation data, and the average of these five accuracy scores is considered the final validation accuracy.

\section{Additional numerical experiments}\label{sec_sm: addiional_numerics}

\subsection{Shadow-size dependence}\label{sec_sm: shadow_size}

\begin{figure*}[t]
    \centering
    \includegraphics[width=\linewidth]{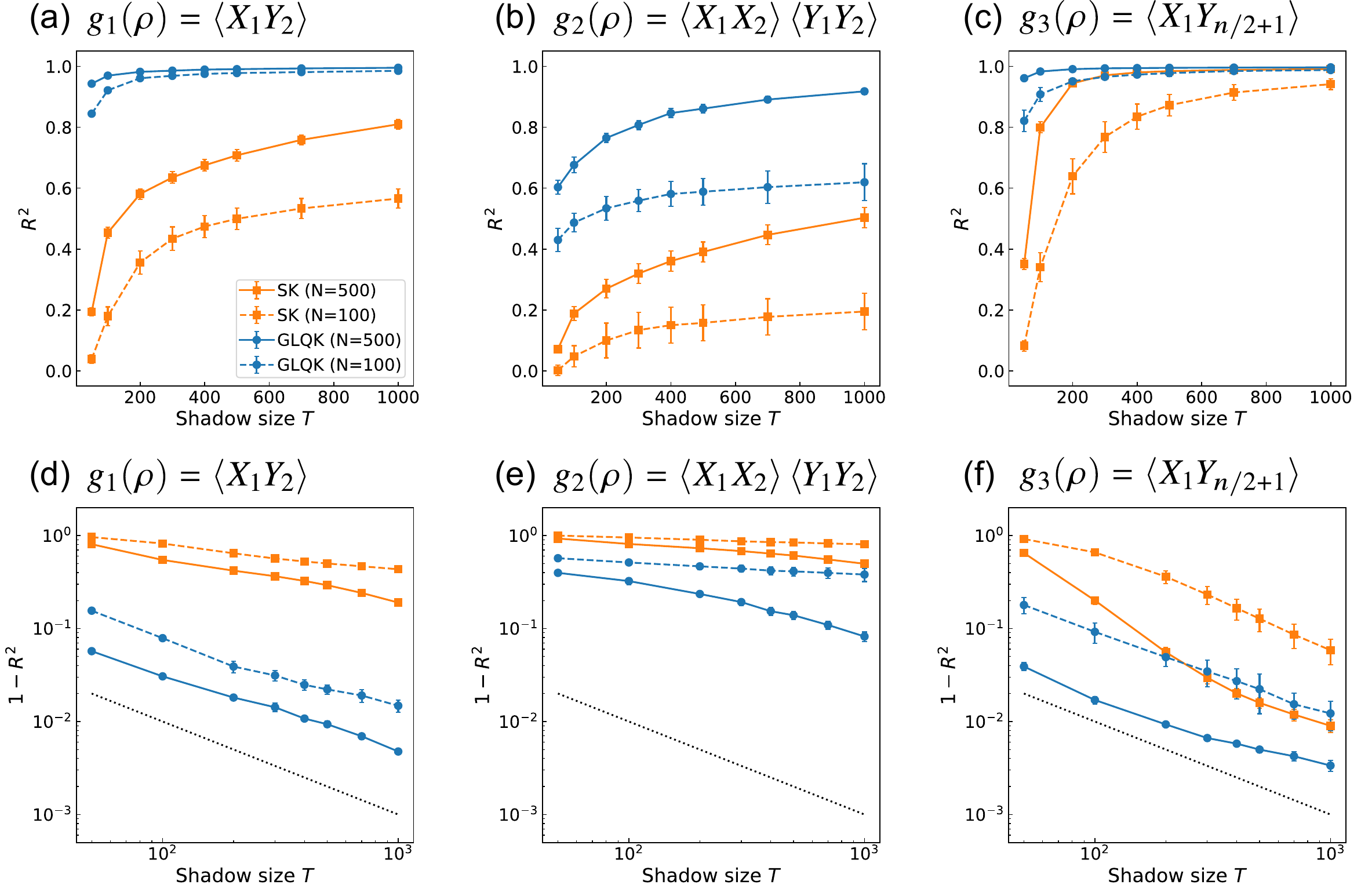}
    \caption{
Regression performance as a function of the shadow size $T$.
Panels (a)--(c) show the regression accuracy $R^2$ as a function of the shadow size $T$, while panels (d)--(f) show log-log plots of the regression error $1 - R^2$ as a function of $T$.
The black dotted lines in panels (d)--(f) indicate the reference scaling $1 - R^2 = 1/T$ for ease of comparison.
Here, $N$ denotes the number of training samples.
Error bars indicate the standard deviation over 10 independently and randomly sampled training and test datasets.
    }
    \label{fig_sm: Tdep}
\end{figure*}

Here, we investigate the $T$-dependence of the learning accuracy in the regression task for quantum states generated by random dynamics. 
The problem setting is the same as that in the main text, with the exception of the varied shadow size $T$. 
For simplicity, we focus on the translationally symmetric case and the three types of polynomials.

The results are presented in Fig.~\ref{fig_sm: Tdep}. 
For all three polynomials, the learning accuracy, quantified by $R^2$, improves as $T$ increases. 
This is the expected behavior, as a larger shadow size reduces the shot noise associated with estimating observables from classical shadows. 
Furthermore, the GLQK consistently outperforms the shadow kernel across the entire range of $T$, reaffirming its high learning efficiency.

To analyze the scaling behavior quantitatively, we show the log-log plots of the error proxy $1-R^2$ versus $T$ in Figs.~\ref{fig_sm: Tdep} (d)--(f). 
The black dotted lines represent a pure $T^{-1}$ scaling for reference. 
Theoretically, the expected loss can be bounded by the sum of a generalization error and a shot-noise error. 
Our theorems indicate that this leads to the scaling $1-R^2 \approx \mathcal{O}(1/\sqrt{N}) + \mathcal{O}(1/T)$. 
Thus, a clear $T^{-1}$ scaling would only be observed in the limit of $N \to \infty$, where the generalization error term vanishes.
For a finite $N$, the interplay between the two error terms is visible when comparing Figs.~\ref{fig_sm: Tdep} (d)--(f). 
For instance, in the GLQK for $g_1$ and $g_3$, and in the shadow kernel for $g_3$, the accuracy at $N=500$ is relatively high, suggesting that the generalization error is lower. 
Consequently, the shot-noise term remains more dominant across the range of $T$, resulting in a scaling behavior that more closely follows the ideal $T^{-1}$ guideline compared to the other results.

\subsection{Robustness against weak disorder}\label{sec_sm: robustness}

\begin{figure*}[t]
    \centering
    \includegraphics[width=\linewidth]{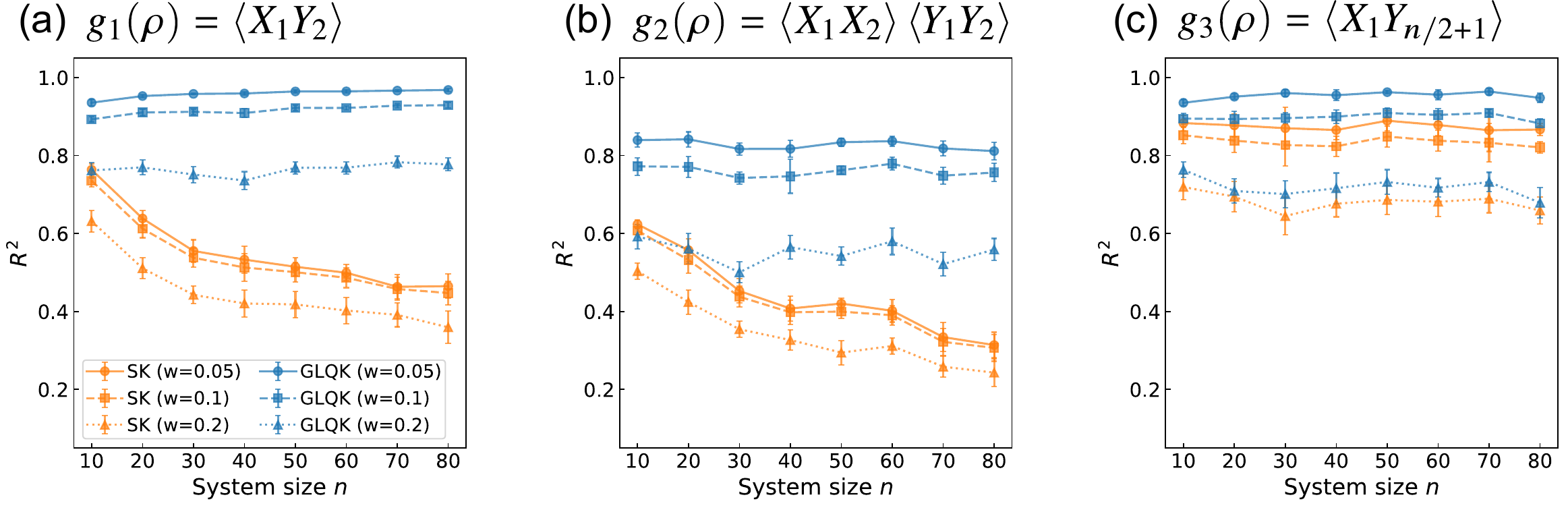}
    \caption{
Regression accuracy $R^2$ as a function of the system size $n$ for various disorder strengths $w$.
The number of training samples is $N=100$ in (a) and (c), and $N=500$ in (b).
Error bars indicate the standard deviation over 10 independently and randomly sampled training and test datasets.
    }
    \label{fig_sm: disorder}
\end{figure*}

For translationally symmetric data, the GLQK can achieve constant sample complexity.
In realistic situations, however, weak disorders within the system might break the translation symmetry of the data. 
To assess robustness against such disorders, we conduct an additional numerical experiment on the regression task for quantum states generated by translationally symmetric random dynamics.
The problem setting is the same as that in the main text, up to the Hamiltonian governing dynamics.
Here, the Hamiltonian is given by
\begin{align}
\tilde{H}_1 = \sum_{j=1}^n \sum_{\mu,\nu\in\{X,Y,Z\}} \tilde{J}^{\mu \nu}_j \sigma^\mu_j \sigma^{\nu}_{j+1}.
\end{align}
Unlike the translationally symmetric Hamiltonian $H_1$, the interaction strength $\tilde{J}_j^{\mu\nu}$ depends on the site index $j$.
This can be decomposed into average and disorder components: $\tilde{J}_j^{\mu\nu} = \bar{J}^{\mu\nu} + \delta J^{\mu\nu}_j$, where $\bar{J}^{\mu\nu}$ is the site-independent part randomly drawn from the uniform distribution $[-1+w,1-w]$ and $\delta J^{\mu\nu}_j$ is the site-dependent disorder randomly drawn from the uniform distribution $[-w,w]$.
Here, $w$ denotes the disorder strength.
We conduct regression on quantum data $\ket{\psi}$ generated by this Hamiltonian, $\ket{\psi}=e^{-i\tilde{H}_1 t}\ket{\phi}$, where $\ket{\phi}$ is a translationally symmetric random product state.
Since $\ket{\psi}$ is no longer translationally symmetric due to $\tilde{H}_1$, our theory does not necessarily ensure constant sample complexity.
For simplicity, we focus on the three types of polynomials, $g_1, g_2$, and $g_3$, as target functions.

For this experiment, numerical results show that constant sample complexity is robust against weak disorder (Fig.~\ref{fig_sm: disorder}).
In the GLQK, although increasing the disorder strength $w$ reduces overall prediction accuracy, the accuracy does not decrease significantly with system size $n$, suggesting that constant sample complexity persists even in the presence of weak disorders.
This behavior contrasts with the shadow kernel, where accuracy generally decreases with the system size, regardless of the presence of disorders.

In theory, we anticipate that the scaling behavior under weak disorders can be divided into two cases.
When the required regression accuracy $\epsilon$ is greater than a scale determined by the disorder strength $w$, the constant sample complexity will persist.
In this case, the quantum state can be regarded as approximately translationally symmetric under precision $\epsilon$, resulting in constant sample complexity.
In fact, although a finite shadow size leads to effective ``disorder'' (i.e., the density matrix reconstructed from the classical shadow is no longer the exact translationally symmetric state), the GLQK achieves constant sample complexity.
On the other hand, when the accuracy is below the disorder scale, one would need to learn the fine structure introduced by the disorder, and thus, the constant sample complexity will not hold.
In this regime, the system effectively becomes non-translationally symmetric, and we expect polynomial sample complexity as described by Theorem~\ref{thm_sm: main_nontranslation}.

%

\end{document}